\documentclass[11pt]{article}
\usepackage[utf8]{inputenc}
\usepackage{amsfonts,algorithmic}
\usepackage{xcolor}
\usepackage{qcircuit}
\usepackage{bm}
\usepackage{amsmath}
\usepackage{enumerate}
\usepackage{appendix}
\usepackage{amsthm,mathrsfs}
\usepackage{multirow}
\usepackage{enumitem}
\usepackage{tikz}
\usepackage{multicol}
\usetikzlibrary{matrix,decorations.pathreplacing}
\usepackage{mleftright}
\usepackage{amssymb}
\usepackage[linesnumbered,ruled,vlined]{algorithm2e}
\usepackage{comment}
\usepackage{subcaption}
\usepackage{url}
\usepackage{soul}
\usepackage{bbm}
\usepackage[pagebackref]{hyperref}
\hypersetup{
    colorlinks,
    linkcolor={red!100!black},
    citecolor={green!100!black},
}
\usepackage{authblk}
\usepackage[capitalize,nameinlink,noabbrev]{cleveref}
\usepackage{makecell}

\usepackage{natbib} 
\setcitestyle{square,compress,numbers,comma}
\usepackage[a4paper,bindingoffset=0cm,left=2.0cm,right=2.0cm,top=2.5cm,bottom=2.5cm,footskip=1.0cm]{geometry}

\usepackage[english]{babel}
\usepackage[T1]{fontenc}
\usepackage{float}
\usepackage[thinlines]{easytable}
\usepackage{mathtools}
\usepackage{subcaption}
\usepackage{booktabs}
\usepackage{setspace}
\usepackage{authblk}
\usepackage{hhline}

\usepackage{amssymb,amsfonts,xspace,graphicx,relsize,bm,mathtools,xcolor,stmaryrd}
\usepackage{mathrsfs}
\usepackage{enumerate}
\usepackage{bm}
\usepackage{bbm}
\usepackage{enumitem}
\usepackage[normalem]{ulem}

\usepackage{amsthm}
\usepackage{lipsum} 

\usepackage{tikz}
\usepackage{lipsum}
\usepackage{physics}

\usepackage{multicol,multirow}

\usepackage[capitalize,nameinlink,noabbrev]{cleveref}

\newtheorem{theorem}{Theorem}

\crefname{lemma}{Lemma}{lemma}
\crefname{definition}{Definition}{definition}
\crefname{corollary}{Corollary}{corollary}
\crefname{fact}{Fact}{fact}


\allowdisplaybreaks

    \newcommand{\inp}[2]{\langle{#1},{#2}\rangle} 

    \def\L{\mathcal{L}}

    \newcommand{\Z}{\ensuremath{\mathbb{Z}}}
    
    \newcommand{\R}{\ensuremath{\mathbb{R}}}

    \DeclareMathOperator{\poly}{poly}

    \newtheorem{definition}[theorem]{Definition}
    \newtheorem{fact}[theorem]{Fact}
    
    \newtheorem{lemma}[theorem]{Lemma}

\definecolor{Joa}{rgb}{0.4,0.3,0.9}
\definecolor{Ale}{rgb}{0.1,0.5,0.5}
\definecolor{Adi}{rgb}{0.9,0.0,0.0}
\definecolor{George}{rgb}{0.1,0.8,0.2}

\title{On the practicality of quantum sieving algorithms for the shortest vector problem}

\author[1,2]{Joao F. Doriguello\thanks{\scriptsize Corresponding author: doriguello@renyi.hu}}
\author[2,3]{George Giapitzakis}
\author[2]{Alessandro Luongo}
\author[2]{Aditya Morolia}
\affil[1]{HUN-REN Alfr\'ed R\'enyi Institute of Mathematics, Budapest, Hungary}
\affil[2]{Centre for Quantum Technologies, National University of Singapore, Singapore}
\affil[3]{David R. Cheriton School of Computer Science, University of Waterloo, Waterloo, Canada}

\date{\today}

\begin{document}

\maketitle

\begin{abstract}
    One of the main candidates of post-quantum cryptography is lattice-based cryptography. Its cryptographic security against quantum attackers is based on the worst-case hardness of lattice problems like the shortest vector problem (SVP), which asks to find the shortest non-zero vector in an integer lattice. Asymptotic quantum speedups for solving SVP are known and rely on Grover's search. However, to assess the security of lattice-based cryptography against these Grover-like quantum speedups, it is necessary to carry out a precise resource estimation beyond asymptotic scalings. In this work, we perform a careful analysis on the resources required to implement several sieving algorithms aided by Grover's search for dimensions of cryptographic interests. For such, we take into account fixed-point quantum arithmetic operations, non-asymptotic Grover's search, the cost of using quantum random access memory (QRAM), different physical architectures, and quantum error correction. We find that even under very optimistic assumptions like circuit-level noise of $10^{-5}$, code cycles of $100$~ns, reaction time of $1$~$\mu$s, and using state-of-the-art arithmetic circuits and quantum error-correction protocols, the best sieving algorithms require $\approx 10^{13}$ physical qubits and $\approx 10^{31}$ years to solve SVP on a lattice of dimension $400$, which is roughly the dimension in which SVP is to be solved in order to break the minimally secure post-quantum cryptographic standards currently being proposed by NIST. We estimate that a $6$-GHz-clock-rate single-core classical computer would take roughly the same amount of time to solve the same problem. We conclude that there is currently little to no quantum speedup in the dimensions of cryptographic interest and the possibility of realising a considerable quantum speedup using quantum sieving algorithms would require significant breakthroughs in theoretical protocols and hardware development.
\end{abstract}

\tableofcontents


\section{Introduction}
\label{sec:introduction}

Lattice-based cryptography~\cite{hoffstein1998ntru,regev2005lattices,regev2006lattice,micciancio2012trapdoors} has emerged as an important alternative to traditional discrete-log-based cryptosystems like RSA, DSA, and Elliptic-curve cryptography since the advent of Shor's algorithm in 1994~\cite{shor1994algorithms,shor1999polynomial}. Apart from the belief of being cryptographically secure against quantum attacks, lattice-based cryptography has several other important properties, like being based on {\em worst-case hardness} of lattice problems, e.g., the shortest vector problem (SVP)~\cite{ajtai96generating}, and allowing fully homomorphic encryption schemes~\cite{gentry2009fully,brakerski2014fully}. For these reasons, lattice-based cryptography is still considered one of the main candidates of \emph{post-quantum cryptography}~\cite{bernstein2017postquantum}, to the point of being one of the finalist in NIST's undertaking of the standardization of post-quantum cryptography schemes~\cite{nist_cfp}. It is therefore of paramount importance to understand the security level provided by lattice-based cryptography not only against classical attackers but also against quantum ones in order to determine the security guaranteed at various parameter regimes.

The security assumptions of such schemes are related to the problem of finding the shortest non-zero vector in a lattice, in the sense that the best attacks on them make use of an oracle for (approximate) SVP. There are currently three main types of algorithms to solve SVP: sieving~\cite{ajtai2001sieve,Ajtai2001overview,micciancio2010faster,aggarwal2015solving}, enumeration~\cite{fincke1985improved,Kannan1983improved,Pohst1981computation}, and constructing the Voronoi cell of the lattice~\cite{Agrell2002closest,micciancio2010deterministic}. Heuristic versions of lattice sieving and enumeration have seen a lot of success in solving SVP practically, with lattice sieving~\cite{kirshanova2023new} holding the record for breaking an NTRU~\cite{hoffstein1998ntru} challenge by Security Innovation Inc.~\cite{NTRU_challenge} with largest dimension. By using the
BKZ (Block-Korkine-Zolotarev) algorithm~\cite{Schnorr1994Lattice} with lattice sieving, Kirshanova, May, and Nowakowski~\cite{kirshanova2023new} recently broke a lattice-based construction in dimension $D=181$ in $20$ core years. Despite this and a long line of work on such algorithms~\cite{Kannan1983improved, Pohst1981computation, fincke1985improved, nguyen2008sieve, micciancio2010faster, Laarhoven2016search}, however, enumeration and sieving algorithms remain notoriously hard to analyze. The situation is further complicated by the introduction of quantum subroutines into sieving and enumeration algorithms like Grover's search~\cite{grover1996fast,grover97quantum}, which makes unclear how secure lattice-based cryptography is against these ``quantumly-enhanced'' algorithms. It is thus of critical importance to assess the actual quantum advantage that subroutines like Grover's search provide in solving SVP.

A few different works have tried to estimate the amount of resources required, and thus the computational advantage provided, by Grover's search in sieving~\cite{albrecht2020estimating,kim2024finding,kim2024quantum} and in enumeration~\cite{bai2023concrete,bindel2024quantum,prokop2024grover} algorithms. However, all of the existing work on such algorithms ignores the spacetime cost of quantum random access memory ($\mathsf{QRAM}$)~\cite{giovannetti2008architectures,giovannetti2008quantum} and/or of quantum error correction on fault tolerant quantum computers, which can add a significant overhead. In this work, we perform a very thorough analysis of the quantum resources required to enhance several \emph{sieving algorithms} with Grover's search by taking into consideration fixed-point arithmetic operations, non-asymptotic Grover's search, the cost of $\mathsf{QRAM}$, and quantum error correction.

\subsection{Previous works}

Sieving algorithms, introduced by Ajtai, Kumar, and Sivakumar~\cite{ajtai2001sieve,Ajtai2001overview}, attempt to solve SVP by sampling several vectors and combining them together in order to generate shorter vectors. The sampled vectors are thus repeatedly ``sieved'' using a norm-reducing operation until a vector with shortest norm remains. The first practical and heuristic sieving algorithm was designed by Nguyen and Vidick~\cite{nguyen2008sieve}. The Nguyen-Vidick sieve ($\mathtt{NVSieve}$) solves SVP in a $D$-dimensional lattice in time $2^{0.415D + o(D)}$ under heuristic assumptions. Shortly after, Micciancio and Voulgaris~\cite{micciancio2010faster} presented $\mathtt{GaussSieve}$, a heuristic sieving algorithm with a time complexity conjectured to be equal to that of $\mathtt{NVSieve}$, i.e., $2^{0.415D + o(D)}$, but with better performance in practice. Since then, several new sieving algorithms have been proposed~\cite{wang2011improved,zhang2014three,laarhoven15angular,laarhoven2015faster,becker2015speeding,becker2016efficient,becker2016new}. In particular, heuristic sieves like $\mathtt{NVSieve}$ and $\mathtt{GaussSieve}$ have been improved with nearest-neighbour-search methods~\cite{Indyk1998approximate} like locality sensitive hashing (LSH)~\cite{charikar2002similarity,Andoni2014beyond,Andoni2015optimal} and locality sensitive filtering (LSF)~\cite{becker2015speeding,becker2016new}. These techniques allow to reduce the number of vector comparisons by storing low-dimensional sketches (hashes) such that nearby vectors have a higher chance of sharing the same hash than far away vectors. The asymptotically best classical sieving algorithms are the $\mathtt{NVSieve}$/$\mathtt{GaussSieve}$ enhanced with spherical LSH~\cite{laarhoven2015faster} and $\mathtt{NVSieve}$/$\mathtt{GaussSieve}$ enhanced with spherical LSF~\cite{becker2016new}, which can heuristically solve SVP in time $2^{0.2971D +o(D)}$ and $2^{0.2925D + o(D)}$, respectively. For more on sieving algorithms, see the review~\cite{sun2020review}.

Quantum algorithms for SVP have recently been explored. 
Laarhoven, Mosca, and van de Pol~\cite{Laarhoven2015Quantum} studied the impact of Grover's search on the asymptotic complexity of various classical sieving algorithms, including $\mathtt{NVSieve}$ and $\mathtt{GaussSieve}$. They concluded that SVP can be heuristically solved on a quantum computer in time $2^{0.2671D + o(D)}$ by employing Grover's search on $\mathtt{NVSieve}$/$\mathtt{GaussSieve}$ with spherical LSH, a $\approx 9\%$ reduction in exponent compared to the classical complexity of~\cite{laarhoven2015faster}. Later, Laarhoven~\cite{Laarhoven2016search} improved the time complexity to $2^{0.2653D + o(D)}$ by employing Grover's search on $\mathtt{NVSieve}$/$\mathtt{GaussSieve}$ with spherical LSF, again leading to a $\approx 9\%$ reducing in exponent compared to its classical counterpart~\cite{becker2016new}. Chailloux and Loyer~\cite{chailloux21}, on the other hand, presented a modified algorithm in which Grover's search over a filtered list is replaced with a quantum random walk~\cite{magniez2007search}. This brings down the asymptotic time of the quantum algorithm to $2^{0.2570D + o(D)}$. We note that their algorithm still uses Grover's search in the update operation of the quantum random walk. A subsequent work~\cite{bonnetain2023finding} managed to reduce the complexity even further to $2^{0.2563D + o(D)}$ via improved results on quantum random walks. Other works on quantum heuristic sieving algorithms include~\cite{kirshanova2019quantum}. Regarding provably correct algorithms, Aggarwal et al.~\cite{aggarwal2020improved} more recently gave a provable quantum algorithm that solves SVP in time $2^{0.95D+o(D)}$ and requires $2^{0.5D+o(D)}$ classical memory and $\operatorname{poly}(D)$ qubits. If given access to a $\mathsf{QRAM}$ of size $2^{0.293D+o(D)}$, their algorithm requires time $2^{0.835D+o(D)}$ while using the same amount of classical memory and qubits. This improves upon the previously known fastest classical provable algorithm~\cite{aggarwal2015solving}.

Beyond asymptotic complexities, Albrecht et al.~\cite{albrecht2020estimating} analysed the cost of quantum algorithms for nearest neighbor search with focus on sieving algorithms. They presented a quantum circuit for performing a simple version of LSF using a \emph{population count} filter, which lets two vectors through the same filter whenever their hashes (using Charikar’s LSH scheme~\cite{charikar2002similarity}) have small Hamming distance. The authors then employed Grover's algorithm inside a quantum amplitude amplification routine~\cite{brassard2002quantum} to search over the filtered list of nearby vectors to a given vector. By assuming $32$ bits of precision, taking quantum arithmetic operations into consideration, disregarding the cost of $\mathsf{QRAM}$, and using a simplified quantum error-correction analysis, Albrecht et al.~\cite{albrecht2020estimating} compared the number of classical and quantum operations employed by three different sieving algorithms: the $\mathtt{NVSieve}$~\cite{nguyen2008sieve}, the bgj1 specialisation~\cite{albrecht2019general} of the Becker-Gama-Joux sieve~\cite{becker2015speeding} (which is akin to $\mathtt{NVSieve}$ with angular LSH~\cite{laarhoven15angular}), and the $\mathtt{NVSieve}$ with spherical LSF~\cite{becker2016new}. They concluded that the number of quantum operations is indeed asymptotically smaller than the number of classical operations, but are comparable at cryptographic dimensions of interest. For example, at dimension $D=400$, which is roughly the dimension in which SVP has to be solved to be able to break the minimally secure post-quantum cryptographic standards currently being standardised~\cite{Dilithium2021Algorithm}, Albrecht et al.~\cite[Figure~2]{albrecht2020estimating} estimated that $\mathtt{NVSieve}$ with spherical LSF (called ListDecodingSearch in their paper) requires either $\approx 10^{42}$ quantum operations or $\approx 10^{43}$ classical operations. 

Regarding other works on resource estimations of quantum sieving algorithms, Kim et al.~\cite{kim2024quantum} estimated the number of logical qubits and logical quantum gates required by Grover's search on $\mathtt{NVSieve}$ to solve SVP in lattices of small dimensions. As an example, by ignoring $\mathsf{QRAM}$ and quantum error correction, the authors estimated that a \emph{single} Grover's search would require $\approx 7\cdot 10^{7}$ logical quantum gates and $\approx 1.5\cdot 10^{6}$ logical qubits in dimension $D=70$ (cf.~\cite[Table~3]{kim2024quantum}). On the other hand, Prokop et al.~\cite{prokop2024grover} proposed a quantum circuit for and studied the resource requirements of a Grover oracle for SVP and analysed how to combine Grover's search with the BKZ algorithm. Beyond sieving algorithms, we briefly mention a variational quantum algorithm proposal with resource estimations for the NISQ era~\cite{albrecht2023variationalquantum} and estimations for quantum enumeration algorithms~\cite{bai2023concrete,bindel2024quantum} and for Grover's search attacks on EAS~\cite{grassl2016applying,almazrooie2018quantum,bonnetain2019quantum,jaques2020implementing} and on SHA-2/SHA-3~\cite{amy16estimating}.

\subsection{Our contributions}

In this paper, we study how practical quantum speedups for lattice sieves are by performing a precise estimate on the amount of resources required to implement Grover's search on several sieving algorithms. The sieving algorithms considered in this work are the plain $\mathtt{NVSieve}$~\cite{nguyen2008sieve} and $\mathtt{GaussSieve}$~\cite{micciancio2010faster} and their enhanced versions with angular/hyperplane LSH (also known as $\mathtt{HashSieve}$)~\cite{laarhoven15angular}, with spherical/hypercone LSH (also known as $\mathtt{SphereSieve}$)~\cite{laarhoven2015faster}, and with spherical/hypercone LSF (also known as $\mathtt{BDGL}$ sieve)~\cite{becker2016new}, to a total of $8$ different sieves. Each of these sieving algorithms perform several Grover's searches per sieving step in order to find lattice vectors that combined yield a new lattice vector with a smaller norm. We compute the amount of physical qubits and time required to perform \emph{all} Grover's searches in a typical instance of the aforementioned sieves. For such, we take into account the following. We stress that all the assumptions and values below are physically motivated and realistic, although optimistic.
\begin{enumerate}
    \item \textbf{Fixed-point quantum arithmetic.} Every entry of a $D$-dimensional vector is stored using two's-complement representation with $\kappa = 32$ (qu)bits and arithmetic operations on a quantum computer are performed modulo $2^{\kappa}$. Possible overflows are ignored. We decompose the Grover oracle behind each sieving algorithm into basic arithmetic operations like addition, comparison, and multiplication, and employ quantum circuits for each such arithmetic operation. For quantum addition and comparison, we utilise Gidney's out-of-place quantum adder~\cite{gidney2018halving}, which has the lowest $\mathsf{Toffoli}$-count of all quantum adders that we are aware of. For quantum multiplication, we utilise a simple decomposition into quantum adders based on schoolbook multiplication that has lower $\mathsf{Toffoli}$-count compared to previous works. A similar construction has appeared in~\cite{babbush2018encoding} and, very recently, in~\cite{litinski2024quantum}.
    
    \item \textbf{Non-asymptotic Grover's search.} It is well known that Grover's search requires $\lfloor \frac{\pi}{4}\sqrt{N/M}\rfloor$ iterations to find one out of $M$ marked elements in a database of size $N$ with high probability if $M$ and $N$ are known beforehand. We do \emph{not} assume to know the number of solutions to any Grover's search within a sieving algorithm. This requires an exponential search Grover's algorithm~\cite{boyer1998tight} whose complexity beyond an asymptotic scaling was analysed by Cade, Folkertsma, Niesen, and Weggemans~\cite{Cade2023quantifyinggrover} and which we borrow.

    \item \textbf{Quantum random access memory.} We take into consideration the cost of employing quantum random access memory ($\mathsf{QRAM}$) to quantumly access a classical database within Grover's search. We work exclusively with $\mathsf{QRAM}$s that access \emph{classical data} and consider the circuit implementation from Arunachalam et al.~\cite{arunachalam2015robustness} (see also~\cite{di2020fault}) of the bucket-brigade $\mathsf{QRAM}$ architecture~\cite{giovannetti2008architectures,giovannetti2008quantum}, which is conceptually simple, has shallow depth, and is noise resilient~\cite{hann2021resilience}. We assume that the memory content can be classically rewritten without affecting the $\mathsf{QRAM}$ circuit.

    \item \textbf{Physical architectures.} It is necessary to specify a physical architecture for a general-purpose fault-tolerant quantum computer. Here we assume two different types of architectures: \emph{baseline} architectures with nearest-neighbor logical two-qubit interactions on a 2D grid~\cite{Litinski2019gameofsurfacecodes,fowler2018low,Chamberland2022universal,Chamberland2022building,bombin2021interleaving}, of which Google's Sycamore processor~\cite{arute2019quantum} is an example, and the \emph{active-volume} architecture recently proposed by Litinski and Nickerson~\cite{litinski2022active} that employs a logarithmic number of non-local connections between logical qubits.

    \item \textbf{Quantum error correction.} Physical quantum computers are heavily affected by noise and a realistic resource estimate should take this into consideration. In this paper we assume an incoherent circuit-level noise model for the physical qubits with error $p_{\rm phy} = 10^{-5}$. In order to protect against errors, we use surface codes introduced by Kitaev~\cite{kitaev1997quantum,kitaev2003fault} to encode logical qubits, or more precisely, a patch-based surface-code encoding~\cite{horsman2012surface}. The time required to measure all surface-code check operators as part of error detecting and correction defines a \emph{code cycle}, which we assume to be $100$~ns. The most expensive operations on surface codes are non-Clifford gates like $\mathsf{T}$ and $\mathsf{Toffoli}$ gates, which can be performed by consuming ``magic states''~\cite{Bravyi2005universal} akin to teleportation protocols. We take into consideration space and time overheads to consume magic states by following the framework of~\cite{Litinski2019gameofsurfacecodes,litinski2022active}. In order to prepare low-error magic states, short error-detecting quantum procedures known as magic state distillation protocols~\cite{Bravyi2005universal,Reichardt2005quantum} are used. Here we employ the distillation protocols from Litinski~\cite{Litinski2019magicstate} which are one of the best we know of. More specifically, we employ a three-level concatenation distillation protocol by using two 15-to-1 punctured Reed-Muller codes~\cite{Bravyi2005universal,Haah2018codesprotocols} followed by a third and final 8-to-CCZ distillation protocol~\cite{Gidney2019efficientmagicstate} to obtain $|CCZ\rangle$ magic states with errors smaller than $10^{-40}$, which are used to perform fault tolerant $\mathsf{Toffoli}$ gates. Finally, the time required to perform a layer of measurements, feed the measurement outcomes into a classical decoder, perform a classical decoding algorithm like minimum-weight perfect matching~\cite{Edmonds1965paths,dennis2002topological} or union-find~\cite{delfosse2020lineartime,Delfosse2021almostlineartime}, and use the result to send new instructions to the quantum computer is called \emph{reaction time}. We assume a reaction time of $1~\mu$s. We note that, although the values used here for circuit-level noise, code cycle, and reaction time are not strictly impossible, they are quite optimistic~\cite{ballance2016high,gaebler2016high,madjarov2020high,clark2021high,moskalenko2022high,ruskuc2022nuclear,Chen2021exponential,Anderson2021realization,Krinner2022realizing,Battistel2023realtime}.

    \item \textbf{Classical hashing operations.} Hashing techniques can be used to decrease the time searching for reducing vectors and require purely classical operations. We take into consideration the amount of time required to classically hash vectors on top of the time required to quantumly search for reducing vectors with Grover algorithm. We break the hashing operations into additions and multiplications and assume that one addition takes $1$ cycle/instruction while one multiplication takes $4$ cycles/instructions. We consider a $6$-GHz-clock-rate single-core classical computer, i.e., it performs $6\cdot 10^9$ instructions per second. We disregard memory allocation times.
\end{enumerate}

For the sake of comparison, we also consider classical versions of $\mathtt{NVSieve}$ and $\mathtt{GaussSieve}$ in which the searching part is perform classically in a sequential manner instead of using Grover algorithm. For such, we decompose the searching operation into basic arithmetic operations like addition and multiplications (this decomposition is the same for the Grover oracle). Similarly to the classical hashing operations, we assume that one addition takes $1$ instruction and one multiplication takes $4$ instructions. We consider a $6$-GHz-clock-rate single-core classical computer.

Although resource estimates as comprehensive as ours have been carried out under similar considerations for algorithms like Shor's~\cite{Gidney_2021,litinski2023compute}, we are not aware of similar results on sieving (or enumeration) algorithms. The work of Albrecht et al.~\cite{albrecht2020estimating} is the closest to our results, but they fall short of considering $\mathsf{QRAM}$ and conducting a more rigorous analysis on quantum error correction. As an example, the scripts provided by Gidney and Eker\aa~\cite{Gidney_2021} and adapted by Albrecht et al.\ consider two-level distillation protocols which, although enough in the context of Shor's algorithm, cannot produce magic states with small enough errors for sieving algorithms in high dimensions. A three or four-level distillation protocol is required to reach errors below $10^{-40}$ or even $10^{-50}$.

Since $\mathtt{NVSieve}$ and $\mathtt{GaussSieve}$ are inherently randomised algorithms, we carried out the resource estimates under heuristic assumptions on the value of internal parameters of these sieves. As examples, we assume that the initial list size in $\mathtt{NVSieve}$ is $D\cdot 2^{0.2352D + 0.102\log_2{D} + 2.45}$ as numerically computed by Nguyen and Vidick~\cite{nguyen2008sieve}, while the maximum list size in $\mathtt{GaussSieve}$ is $2^{0.193D + 2.325}$ as calculated by us and similarly reported by Mariano et al.~\cite{mariano14comprehensive}. The number of sieving steps in $\mathtt{GaussSieve}$ has been reported to grow as $2^{0.283D + 0.335}$ by Mariano et al.~\cite{mariano14comprehensive} and independently checked by us. We refer the reader to \cref{sec:heuristics} for a complete list of assumptions on the average performance of $\mathtt{NVSieve}$ and $\mathtt{GaussSieve}$. On the other hand, the use of hashing techniques (LSH and LSF) introduces two tunable parameters: the size of the hash space and the number of hash tables. The values used for these parameters are highly heuristic in practice, while in asymptotic analyses they are chosen so to guarantee that nearby vectors collide (have the same hash) in at least one hash table with high probability and to balance out the time spent hashing and the time spent searching for reducing vectors. Here we follow a (slightly more detailed) version of the asymptotic analysis. To be more precise, we set the parameters in order to balance the classical hashing time and the quantum searching time by \emph{ignoring overall complexity constant factors}, meaning that classically hashing a list of certain size would be roughly as costly as quantumly searching the same list. Although not an entirely realistic assumption, it is optimistic in that it lessens the computational burden on hashing. We leave a more detail analysis on the hashing parameters for a future work.

\begin{figure}[t]
    \centering
    \begin{subfigure}[b]{0.49\textwidth}
        \includegraphics[width=\textwidth]{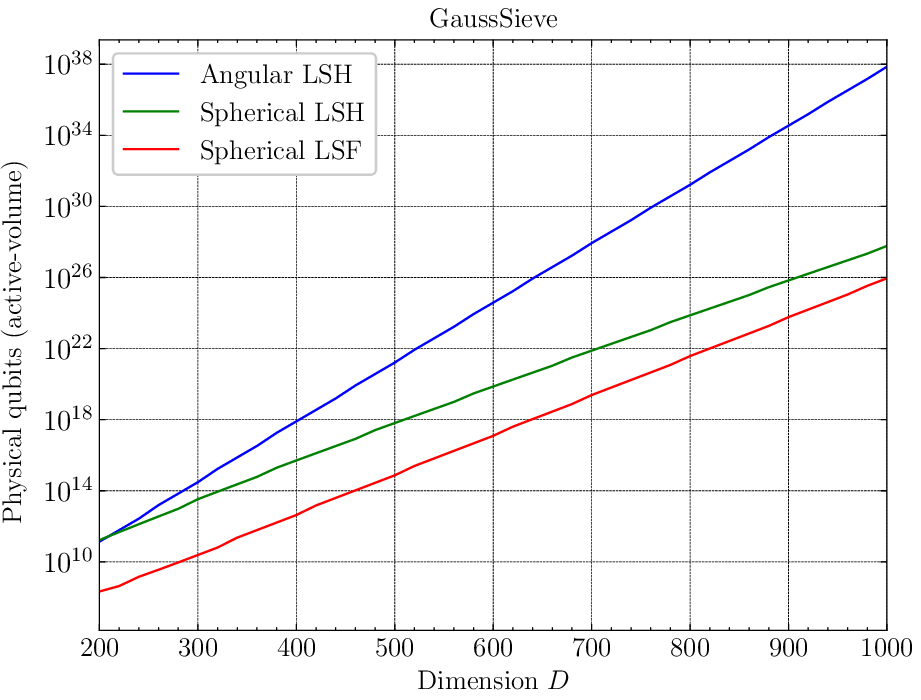}
        \caption{Active-volume physical qubits of $\mathtt{GaussSieve}$}
        \label{fig:heuristic_assumptions_intro_a}
    \end{subfigure}
    \begin{subfigure}[b]{0.49\textwidth}
        \includegraphics[width=\textwidth]{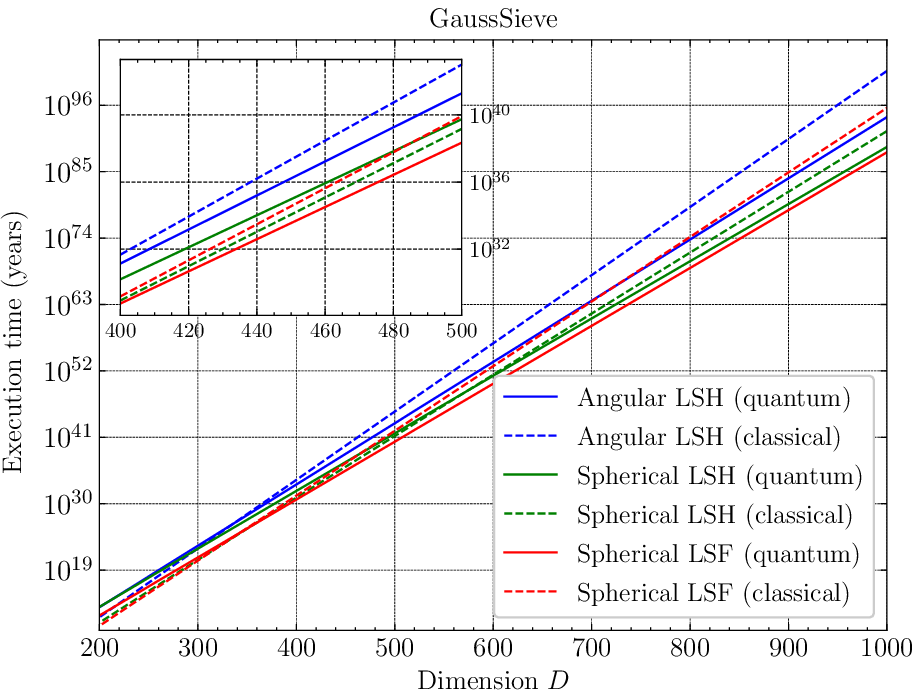}
        \caption{Execution time of $\mathtt{GaussSieve}$}
        \label{fig:heuristic_assumptions_intro_b}
    \end{subfigure}
    \caption{Number of physical qubits and execution time of all Grover's searches in $\mathtt{GaussSieve}$ with LSH/LSF as a function of the lattice dimension $D$. We assume an underlying active-volume physical architecture. The execution time is the sum of the time spent searching for pairs of reducing vectors (either quantumly or classically) and the classical time spent hashing.}
    \label{fig:heuristic_assumptions_intro}
\end{figure}

Our main results are condensed in \cref{fig:heuristic_assumptions_intro}, where we show the amount of physical qubits and time required by $\mathtt{GaussSieve}$ with LSH/LSF as a function of the lattice dimension $D$. We consider an active-volume architecture and omit results for the $\mathtt{NVSieve}$ for now as $\mathtt{GaussSieve}$ has a better performance. The number of physical qubits from \cref{fig:heuristic_assumptions_intro_a} is the number of physical qubits required to run the largest Grover's search in $\mathtt{GaussSieve}$, since physical qubits can be reused in different searches. On the other hand, \cref{fig:heuristic_assumptions_intro} shows the time required to execute both a classical and a quantum version of $\mathtt{GaussSieve}$, i.e., where the searching is performed either classically or via Grover's search. More precisely, the execution time of the classical $\mathtt{GaussSieve}$ is the sum of all searching and hashing operations, while the execution time of the quantum $\mathtt{GaussSieve}$ is the time required to \emph{sequentially} execute all Grover's searches plus the time required to classically hash all vectors.

At dimensions of cryptographic interest, e.g., at dimension $D=400$ which corresponds roughly to the first NIST Security level~\cite[Table~1]{Dilithium2021Algorithm}, $\mathtt{GaussSieve}$ with spherical LSF requires $\approx 10^{13}$ physical qubits to solve SVP in $\approx 10^{31}$ years. As shown in \cref{sec:qramcost}, most of the physical qubits are coming from the use of a bucket-brigade-style $\mathsf{QRAM}$, since it requires a number of logical qubits roughly equal to the size of the accessed database. The total time comes mostly from the quantum arithmetic circuits and the fact that Grover's search requires, at the end of the day, a deep circuit. A classical version of $\mathtt{GaussSieve}$ with spherical LSF also requires roughly the same amount of time to solve SVP.

\cref{fig:heuristic_assumptions_intro} paints a somewhat pessimistic scenario for quantum sieving algorithms, with the number of physical qubits surpassing modern transistor counts by a few orders of magnitude and a total execution time comparable to their classical counterpart and greater than the age of the universe. While Albrecht et al.~\cite{albrecht2020estimating} compared the number of (arithmetic) classical and quantum operations, which is not ideal as the cost of various elementary operations can vary significantly, here we resolve both classical and quantum operations into actual execution times. The end result as seen in \cref{fig:heuristic_assumptions_intro_b} is a small quantum advantage for dimensions beyond $400$: at $D=1000$, Grover's search provides a speedup by roughly five orders of magnitude. 

We stress that the above numbers ignore all the memory fetch operations, which although should worsen both classical and quantum runtimes, will most likely impact the quantum one more severely since, as we argue in \cref{sec:sieving_algorithms}, the use of hashing techniques yields lists of candidate vectors that require several RAM calls to be accessed via $\mathsf{QRAM}$ and thus be searched with Grover algorithm. Moreover, classical searching operations can be more easily parallelised than Grover's search~\cite{zalka99grover}, in the sense that $F$ parallel Grover algorithms running on $F$ separate search spaces have a total width that is larger by a factor of $F$ compared to a single Grover algorithm on the whole search space while only reducing the depth by a factor of $\sqrt{F}$. 

Although we have not considered the asymptotically better quantum-random-walk-based sieving algorithms from Chailloux and Loyer~\cite{chailloux21} and Bonnetain et al.~\cite{bonnetain2023finding}, we believe that the same problems highlighted above might plague these algorithms. In reality, the situation is likely to be worse as quantum random walks should have overall higher constant factors. The magnitude of such factors, though, has not be explored in the literature yet, which hinders a proper resource estimation of their sieving algorithms akin to what has been done here.

It is expected that several assumptions, numbers, and protocols used in this work will become dated in a few years and several new results on circuit design, quantum error correction, and $\mathsf{QRAM}$ will be discovered (and a few new improvements have indeed been posted online by the time this manuscript was been finalised~\cite{mukhopadhyay2024quantum,gidney2024magic,wills2024constant}), but we believe that the overall message remains that Grover's search (and quadratic improvements for that matter) offers very little advantage over classical search in sieving algorithms at dimensions of cryptographic interest. Any considerable speedups will occur on dimensions far larger than the ones needed for most cryptographic purposes or require significant breakthroughs in theoretical protocols and hardware development.

The remainder of the paper is organised as follows. In \cref{sec:prelim} we review basic concepts from quantum computation and hashing techniques like LSH and LSF. In \cref{sec:error_correction} we review several key ideas from quantum error correction like surface codes, baseline and active-volume architectures, and magic state distillation protocols. \cref{sec:arithmetic} covers all quantum arithmetic circuits employed in our paper. \cref{sec:grover_search} reviews Grover's search algorithm, while \cref{sec:qram} reviews the bucket-brigade $\mathsf{QRAM}$. In \cref{sec:sieving_algorithms} we describe the $\mathtt{NVSieve}$ and $\mathtt{GaussSieve}$ with and without LSH/LSF and construct the Grover oracles for them. In \cref{sec:results} we perform our resource estimation analysis. This section is divided into a few parts: \cref{sec:case_study} describes how the resource estimation is performed for the example when $D=400$; \cref{sec:heuristics} describes our main results; \cref{sec:qramcost} analyses the cost of $\mathsf{QRAM}$; \cref{sec:nist} explores the impact of depth restrictions as proposed by NIST post-quantum cryptography
standardisation process~\cite{nist_cfp}. Finally, we conclude in \cref{sec:discussionandopen}. The source code and data can be found in~\cite{paper_data}.


\section{Preliminaries}
\label{sec:prelim}

Given $n\in\mathbb{N} := \{1,2,\dots\}$, define $[n] := \{1, \ldots, n\}$. Let $\mathsf{X} = \bigl(\begin{smallmatrix} 0&1 \\ 1&0 \end{smallmatrix} \bigr)$, $\mathsf{Y} = \bigl(\begin{smallmatrix} 0&-i \\ i&0 \end{smallmatrix} \bigr)$, and $\mathsf{Z} = \bigl(\begin{smallmatrix} 1&0 \\ 0&-1 \end{smallmatrix} \bigr)$ be the usual Pauli matrices and $\mathsf{I}_n$ the $n$-dimensional identity matrix. We shall refer to $\mathsf{I}_n$ simply as $\mathsf{I}$ when the dimension is clear from context. Let $\mathbf{1}[\text{clause}]\in\{0,1\}$ be the indicator function that equals $1$ if the clause is true and $0$ otherwise. Given vectors $\mathbf{v},\,\mathbf{w} \in \mathbb{R}^D$, let $\|\mathbf{v}\| := (\sum_{i=1}^D v_i^2)^{1/2}$ be the Euclidean norm of $\mathbf{v}$, $\theta(\mathbf{v},\mathbf{w})$ the angle between $\mathbf{v}$ and $\mathbf{w}$, and $\inp{\mathbf{v}}{\mathbf{w}} := \sum_{i=1}^D v_iw_i$ their inner product. Let $\Gamma(z)$ be the gamma function. We denote by $\mathcal{S}^{D-1} := \{\mathbf{v}\in\mathbb{R}^D:\|\mathbf{v}\| = 1\}$ the $D$-dimensional unit hypersphere and by $\mathcal{H}_{\mathbf{v},\alpha} := \{\mathbf{x}\in\mathbb{R}^D: \langle \mathbf{v},\mathbf{x}\rangle \geq \alpha\}$ the half-spaces, where $\mathbf{v}\in\mathcal{S}^{D-1}$. Let $\mathcal{C}_D(\alpha)$ be the measure of the spherical cap $\mathcal{C}_{\mathbf{v},\alpha} := \mathcal{S}^{D-1}\cap \mathcal{H}_{\mathbf{v},\alpha}$ and $\mathcal{W}_D(\alpha,\beta,\theta)$ be the measure of the spherical wedge $\mathcal{W}_{\mathbf{v},\alpha,\mathbf{w},\beta} := \mathcal{S}^{D-1}\cap \mathcal{H}_{\mathbf{v},\alpha}\cap \mathcal{H}_{\mathbf{w},\beta}$, where $\mathbf{v},\mathbf{w}\in\mathcal{S}^{D-1}$ with $\langle \mathbf{v},\mathbf{w}\rangle = \cos\theta$. We shall need the next known facts.
\begin{fact}[{\cite[Lemma~10.7]{Laarhoven2016search}}]
    The probability density function $\Theta_{[\theta_1,\theta_2]}(\theta)$ of angles between vectors $\mathbf{v},\mathbf{w}\in\mathcal{S}^{D-1}$ drawn at random from the unit sphere and such that $\theta_1 \leq \theta(\mathbf{v},\mathbf{w}) \leq \theta_2$ is
    \begin{align*}
        \Theta_{[\theta_1,\theta_2]}(\theta) = \frac{\sin^{D-2}\theta}{\int_{\theta_1}^{\theta_2}\sin^{D-2}\phi~{\rm d}\phi}.
    \end{align*}
\end{fact}
\begin{fact}[\cite{li2010concise}]
    Let $\mathbf{v}\in\mathcal{S}^{D-1}$ and $\alpha\in(0,1)$. The measure  $\mathcal{C}_D(\alpha)$ of the spherical cap $\mathcal{C}_{\mathbf{v},\alpha}$ is
    \begin{align*}
        \mathcal{C}_D(\alpha) := \frac{\mu(\mathcal{C}_{\mathbf{v},\alpha})}{\mu(\mathcal{S}^{D-1})} = \frac{1}{\sqrt{\pi}}\frac{\Gamma(\frac{D}{2})}{\Gamma(\frac{D-1}{2})}\int_0^{\arccos\alpha} \sin^{D-2}\phi ~{\rm d}\phi. 
    \end{align*}
\end{fact}
\begin{fact}[{\cite[Case~8]{lee2014concise}}]
    Let $\mathbf{v},\mathbf{w}\in\mathcal{S}^{D-1}$ with $\langle \mathbf{v},\mathbf{w}\rangle = \cos\theta$. Let $\alpha,\beta\in(0,1)$ such that $\theta < \arccos\alpha + \arccos\beta$ and $(\alpha - \beta\cos\theta)(\beta - \alpha\cos\theta) > 0$. Define $\theta^\ast \in (0,\frac{\pi}{2})$ by $\tan\theta^\ast = \alpha/(\beta\sin\theta) - 1/\tan\theta$. The measure $\mathcal{W}_D(\alpha,\beta,\theta)$ of the spherical wedge $\mathcal{W}_{\mathbf{v},\alpha,\mathbf{w},\beta}$ is
    \begin{align*}
        \mathcal{W}_D(\alpha,\beta,\theta) := \frac{\mu(\mathcal{W}_{\mathbf{v},\alpha,\mathbf{w},\beta})}{\mu(\mathcal{S}^{D-1})} = J_D(\theta^\ast, \arccos\beta) + J_D(\theta - \theta^\ast, \arccos\alpha),
    \end{align*}
    where
    \begin{align*}
        J_D(\theta_1, \theta_2) := \frac{1}{\sqrt{\pi}}\frac{\Gamma(\frac{D}{2})}{\Gamma(\frac{D-1}{2})} \int_{\theta_1}^{\theta_2} \mathcal{C}_{D-1}\left(\arccos\left(\frac{\tan{\theta_1}}{\tan\phi}\right)\right) \sin^{D-2}\phi ~{\rm d}\phi.
    \end{align*}
\end{fact}

\subsection{Quantum computing}

We assume the reader is somewhat familiar with quantum computing. The quantum state of a quantum system is described by a vector from a Hilbert space $\mathscr{H}$, i.e., a complex vector space with inner product structure. A qubit, the quantum equivalent of a bit, is a quantum system described by a vector in $\mathscr{H} \cong\mathbb{C}^2$, while an $n$-qubit system is described by a vector $|\psi\rangle$ in $\mathscr{H} \cong\mathbb{C}^{2^n}$. Equivalently, an $n$-qubit quantum system can also be described by a density matrix $\rho\in\mathbb{C}^{2^n\times 2^n}$, i.e., a semi-definite positive matrix with unit trace. The evolution of a quantum state $|\psi\rangle\in\mathbb{C}^{2^n}$ is described by a unitary operator $\mathsf{U}\in\mathbb{C}^{2^n\times 2^n}$, $\mathsf{U}\mathsf{U}^\dagger = \mathsf{I}$ where $\mathsf{U}^\dagger$ is the Hermitian conjugate of $\mathsf{U}$. A unitary operator is also referred to as a quantum gate. In order to extract classical information from a quantum system, a quantum measurement is usually performed. A quantum measurement is expressed as a positive operator-valued measure (POVM), i.e., a set $\{\mathsf{E}_m\}_m$ of positive operators $\mathsf{E}_m \succ \mathsf{0}$ that sum to identity, $\sum_m \mathsf{E}_m = \mathsf{I}$. The probability of measuring $\mathsf{E}_m$ on $|\psi\rangle$ is $p_m = \langle\psi|\mathsf{E}_m|\psi\rangle$. A quantum circuit is a sequence of quantum gates acting on a set of qubits. At the end of the circuit, a measurement is performed and a classical outcome is observed. We refer the reader to~\cite{nielsen2010quantum,de2019quantum} for more information.

There are a few sets of universal gates that can serve as building blocks for any quantum circuit. One of the most common is the Clifford+T gate set comprising the one and two-qubit gates
\begin{align*}
    \mathsf{H} = \frac{1}{\sqrt{2}}\begin{pmatrix}
        1 & 1 \\
        1 & -1
    \end{pmatrix}, ~
    \mathsf{S} = \begin{pmatrix}
        1 & 0 \\
        0 & i
    \end{pmatrix},~
    \mathsf{T} = \begin{pmatrix}
        1 & 0 \\
        0 & e^{i\pi/4} 
    \end{pmatrix},  ~
    \mathsf{CNOT} = \begin{pmatrix}
         1 & 0 & 0 & 0 \\ 
         0 & 1 & 0 & 0 \\ 
         0 & 0 & 0 & 1 \\
         0 & 0 & 1 & 0 
    \end{pmatrix}.
\end{align*}
Here, $\mathsf{H}, \mathsf{CNOT}, \mathsf{S}$ are Clifford gates, while the $\mathsf{T}$ gate is a non-Clifford gate (it does not normalise the Pauli group). The Clifford+T gate set $\{\mathsf{H}, \mathsf{S}, \mathsf{T}, \mathsf{CNOT}\}$ is universal~\cite{divincenzo1995two,boykin1999universal}, meaning that any quantum circuit can be written in terms of its elements as accurately as required. Another universal gate set is $\{\mathsf{H},\mathsf{S},\mathsf{CNOT},\mathsf{Toffoli}\}$~\cite{Shor1996fault}, where
\begin{align*}
    \mathsf{Toffoli} = \begin{pmatrix}
         1 & 0 & 0 & 0 & 0 & 0 & 0 & 0 \\ 
         0 & 1 & 0 & 0 & 0 & 0 & 0 & 0 \\ 
         0 & 0 & 1 & 0 & 0 & 0 & 0 & 0 \\ 
         0 & 0 & 0 & 1 & 0 & 0 & 0 & 0 \\ 
         0 & 0 & 0 & 0 & 1 & 0 & 0 & 0 \\ 
         0 & 0 & 0 & 0 & 0 & 1 & 0 & 0 \\ 
         0 & 0 & 0 & 0 & 0 & 0 & 0 & 1 \\ 
         0 & 0 & 0 & 0 & 0 & 0 & 1 & 0 \\ 
    \end{pmatrix}.
\end{align*}
Here, $\mathsf{Toffoli}$ is a non-Clifford gate. Define also the $\mathsf{CZ}$ and $\mathsf{CCZ}$ gates as  $\mathsf{CZ} = (\mathsf{I}_2\otimes \mathsf{H})\mathsf{CNOT}(\mathsf{I}_2\otimes \mathsf{H})$ and $\mathsf{CCZ} = (\mathsf{I}_4\otimes \mathsf{H})\mathsf{Toffoli}(\mathsf{I}_4\otimes \mathsf{H})$, respectively. In this work, we shall focus on the $\{\mathsf{H},\mathsf{S},\mathsf{CNOT},\mathsf{Toffoli}\}$ universal gate set, as all of our circuits can be naturally decomposed using these gates. We shall also consider the $\mathsf{CCZ}$ gate to have the same cost as the $\mathsf{Toffoli}$ gate and will count them as a single resource.

By \emph{ancillary qubits} (or simply \emph{ancillae}) we mean qubits that can be re-used across computation, so that a gate $\mathsf{U}_2$ can use ancillae from some previous gate $\mathsf{U}_1$. This means that if two gates $\mathsf{U}_1$ and $\mathsf{U}_2$ use $c_1$ and $c_2$ ancillae, respectively, then the joint gate $\mathsf{U}_1\mathsf{U}_2$ requires $\max(c_1, c_2)$ ancillae. By \emph{dirty ancillae} we mean auxiliary qubits employed in a quantum gate that are left entangled with other qubits and therefore cannot be reused in later computations afresh. This means that if two gates $\mathsf{U}_1$ and $\mathsf{U}_2$ use $c_1$ and $c_2$ dirty ancillae, respectively, then the joint gate $\mathsf{U}_1\mathsf{U}_2$ requires $c_1 + c_2$ dirty ancillae. We will routinely keep dirty ancillae after some computation to facilitate its uncomputation at a later time.  

By $\mathsf{C}^{(k)}\text{-}\mathsf{X}$ we mean an $\mathsf{X}$ gate controlled on $k$ qubits, i.e., an $\mathsf{X}$ gate is applied conditioned on all $k$ qubits being on the $|1\rangle$ state. This means that $\mathsf{C}^{(1)}\text{-}\mathsf{X} = \mathsf{CNOT}$ and $\mathsf{C}^{(2)}\text{-}\mathsf{X} = \mathsf{Toffoli}$. It is possible to decompose $\mathsf{C}^{(k)}\text{-}\mathsf{X}$ into $\mathsf{Toffoli}$ gates  as summarised in the next well-known result.

\begin{fact}[Multi-controlled $\mathsf{Toffoli}$]
    \label{fact:mctoff}
    The multi-controlled $\mathsf{Toffoli}$ gate $\mathsf{C}^{(k)}\text{-}\mathsf{X}$ with $k>1$ controls can be implemented using $k-1$ $\mathsf{Toffoli}$ gates and $k-2$ ancillae.
\end{fact}

\subsection{Locality-sensitive hashing and locality-sensitive filtering}
\label{sec:lsh}

In this work, we consider lattice sieving algorithms. These are algorithms that start with (an exponentially) large list of lattice vectors consisting of long vectors and use it to find shorter lattice vectors. If the length of the vectors in the initial list is roughly the same, then this can be done by finding nearby lattice vectors in the list, since their difference would be a shorter lattice vector. More precisely, we would like to 
\begin{align*}
    \text{find vectors}~ \mathbf{v},\mathbf{w} ~\text{from a list}~ L ~\text{such that}~ \|\mathbf{v}\pm \mathbf{w}\| \leq \max\{\|\mathbf{v}\|, \|\mathbf{w}\|\},
\end{align*}
which is equivalent to
\begin{align*}
    \text{find vectors}~ \mathbf{v},\mathbf{w} ~\text{from a list}~ L ~\text{such that}~ \theta(\mathbf{v}, \pm \mathbf{w}) \leq \pi/3
\end{align*}
if $\|\mathbf{v}\|\approx \|\mathbf{w}\|$. The above problem can naturally be framed as a \emph{nearest neighbour search}. In the nearest neighbour search, a list of $D$-dimensional vectors $L = \{\mathbf{w}_1,\dots,\mathbf{w}_N\}\subset\mathbb{R}^D$ is given and the task is to preprocess $L$ in such a way that given a new vector $\mathbf{v}\notin L$, it is possible to efficiently find an element $\mathbf{w}\in L$ close(st) to $\mathbf{v}$. \emph{Locality-sensitive hashing} (LSH) is a well-known technique to speed up nearest neighbour search and it makes use of locality-sensitive hash functions~\cite{Indyk1998approximate}. A locality-sensitive hash function $h(\cdot)$ projects a $D$-dimensional vector into a low-dimension sketch and has the property that nearby vectors have a higher probability of collision than far away vectors. This sketch can then be used to bucket vectors in $L$ such that the vectors in the same bucket are close and hence speed up the search. A family of hash functions $\mathcal{H}=\{ h : \mathbb{R}^D \to U \subset \mathbb{N}\}$ is characterised by the collision probability
\begin{align*}
    p(\theta) := \Pr_{h\sim \mathcal{H}}[h(\mathbf{v}) = h(\mathbf{w})~|~\mathbf{v},\mathbf{w}\in\mathcal{S}^{D-1}, \langle\mathbf{v}, \mathbf{w}\rangle = \cos\theta],
\end{align*}
where $h\sim \mathcal{H}$ means a hash function $h$ uniformly picked over $\mathcal{H}$.

Another well-known technique is \emph{locality-sensitive filtering} (LSF)~\cite{becker2016new}, which employs a filter that maps a vector to a binary value: a vector either passes a filter or not. A filter that a vector $\mathbf{v}$ passes through is called a \emph{relevant filter} for $\mathbf{v}$. Applied to a list $L$, a filter $f$ maps $L$ to an output filtered list $L_f \subset L$ of points that survive the filter. The idea is to choose a filter that yields an output list $L_f$ of only nearby vectors. A family of filter functions $\mathcal{F}=\{f : \mathbb{R}^D \to \{0,1\}\}$ is characterised by the collision probability
\begin{align*}
    p(\theta) := \Pr_{f\sim \mathcal{F}}[\mathbf{v},\mathbf{v}\in L_f~|~\mathbf{v},\mathbf{w}\in\mathcal{S}^{D-1}, \langle\mathbf{v}, \mathbf{w}\rangle = \cos\theta],
\end{align*}
where $f\sim \mathcal{F}$ means a filter function $f$ uniformly picked over $\mathcal{F}$. We note that while $p(0) = 1$ for hash families, the same is not true for most filter families, since in general the collision probability of $\mathbf{v}$ with itself is $p(0) < 1$.

A hash/filter family with $p(\theta_1) \gg p(\theta_2)$ can efficiently distinguish nearby vectors at angle $\theta_1$ from distant vectors at angle $\theta_2$ by looking at their hash/filter values. The existence of hash/filter families with $p(\theta_1) \approx 1$ and $p(\theta_2) \approx 0$ is, however, not straightforward. A common technique is to first construct a hash/filter family with $p(\theta_1) \approx p(\theta_2)$ and use a series of AND- and OR-compositions to amplify the gap between $p(\theta_1)$ and $p(\theta_2)$ and obtain a new hash/filter family with $p'(\theta_1) > p(\theta_1)$ and $p'(\theta_2) < p(\theta_2)$.

\paragraph*{AND-composition.} Given a hash family $\mathcal{H}$ with collision probability $p(\theta)$, it is possible to construct a hash family $\mathcal{H}' = \mathcal{H}^k$ with collision probability $p(\theta)^k$ by taking $k$ different and pairwise independent hash functions $h_1,\dots,h_k\in\mathcal{H}$ and defining $h\in\mathcal{H}'$ such that $h(\mathbf{v}) = (h_1(\mathbf{v}),\dots,h_k(\mathbf{v}))$. Clearly $h(\mathbf{v}) = h(\mathbf{w})$ if and only if $h_i(\mathbf{v}) = h_i(\mathbf{w})$ for all $i\in[k]$, and thus $p'(\theta) = p(\theta)^k$. Similarly for a filter family $\mathcal{F}$.

\paragraph*{OR-composition.} Given a hash family $\mathcal{H}$ with collision probability $p(\theta)$, it is possible to construct a hash family $\mathcal{H}'$ with collision probability $1 - (1-p(\theta))^t$ by taking $t$ different and pairwise independent hash functions $h_1,\dots,h_t\in\mathcal{H}$ and defining $h\in\mathcal{H}'$ by the relation $h(\mathbf{v}) = h(\mathbf{w})$ if and only if $h_i(\mathbf{v}) = h_i(\mathbf{w})$ for some $i\in[t]$. Clearly $h(\mathbf{v}) \neq h(\mathbf{w})$ if and only if $h_i(\mathbf{v}) \neq h_i(\mathbf{w})$ for all $i\in[t]$, and thus $1 - p'(\theta) = (1-p(\theta))^t$. Similarly for a filter family $\mathcal{F}$. \\

Suitable hash/filter families together with AND and OR-compositions can be used to find nearest neighbors as first described by Indyk and Motwani~\cite{Indyk1998approximate}. The idea is to choose $t\cdot k$ hash functions $h_{i,j}\in\mathcal{H}$ from some hash family $\mathcal{H}$ and use the AND-composition to combine $k$ of them at a time to build $t$ new hash functions $h_1,\dots,h_t$, where $h_i(\cdot) = (h_{i,1}(\cdot),\dots,h_{i,k}(\cdot))$ for $i\in[t]$. Then, given the list $L$, we build $t$ different hash tables $\mathcal{T}_1,\dots,\mathcal{T}_t$ and for each hash table $\mathcal{T}_i$ we insert a vector $\mathbf{w}\in L$ from the list into the bucket labelled by $h_i(\mathbf{w})$. This means that all the vectors from $L$ are inserted into an appropriate bucket in each hash table. Finally, given a target vector $\mathbf{v}$, we compute its $t$ hash images $h_1(\mathbf{v}),\dots,h_t(\mathbf{v})$ and look only for candidate vectors in the bucket labelled $h_i(\mathbf{v})$ in hash table $\mathcal{T}_i$, for all $i\in[t]$ (OR-composition). In other words, we consider only the vectors that collide with $\mathbf{v}$ in at least one of the hash tables. A similar idea applies to filter families. A vector $\mathbf{v}$ is inserted into a filtered bucket $\mathcal{B}_i$ if and only if it survives the concatenated filter $f_i$ made out of filters $f_{i,1},\dots,f_{i,k}$, for $i\in[t]$.

\subsubsection{Angular LSH}

A famous hash family is the angular (or hyperplane) locality-sensitive hash method of Charikar~\cite{charikar2002similarity}, which, as we will see in \cref{sec:sieving_algorithms}, can be used to improve sieving algorithms~\cite{laarhoven15angular}. Charikar proposed the following hash family $\mathcal{H}_{\rm ang}$,
\begin{align*}
    \mathcal{H}_{\rm ang} = \{h_{\mathbf{a}}:\mathbb{R}^D\to\{0,1\} ~|~ \mathbf{a}\in\mathcal{S}^{D-1}\}, \qquad h_{\mathbf{a}}(\mathbf{v}) = \begin{cases}
        1 &\text{if}~\langle\mathbf{a}, \mathbf{v}\rangle \geq 0,\\
        0 &\text{if}~\langle\mathbf{a}, \mathbf{v}\rangle < 0.
    \end{cases}
\end{align*}
The vector $\mathbf{a}$ defining the hash function $h_{\mathbf{a}}$ also defines a hyperplane (for which $\mathbf{a}$ is a normal vector), and $h_{\mathbf{a}}$ maps the two regions separated by the hyperplane onto different bits. Charikar proved~\cite{charikar2002similarity} that the probability of collision is $p(\theta) = 1 - \theta/\pi$, which can be seen from the fact that two vectors $\mathbf{v},\mathbf{w}$ define a two-dimensional plane and these two vectors are mapped onto different hashes if a random line (the intersection between this plane and the hyperplane defined by $\mathbf{a}$) separates $\mathbf{v}$ and $\mathbf{w}$.

Under the angular hash family $\mathcal{H}_{\rm ang}$, consider $t$ hash tables, each with $2^k$ hash buckets, constructed via AND and OR-compositions with randomly sampled hash functions $h_{i,j}\in\mathcal{H}_{\rm ang}$ as previously described. It is possible to calculate the average probability $p_1^\ast$ that two vectors $\mathbf{v},\mathbf{w}\in\mathbb{R}^D$ with $\theta(\mathbf{v},\mathbf{w}) \leq \pi/3$ collide in at least one of the $t$ hash tables:
\begin{align*}
    p_1^\ast = \operatorname*{Pr}_{h_{i,j}\sim\mathcal{H}_{\rm ang}}[\exists i\in[t], h_i(\mathbf{v}) = h_i(\mathbf{w})~|~\theta(\mathbf{v},\mathbf{w}) \leq \pi/3] = \int_0^{\frac{\pi}{3}} \Theta_{[0,\frac{\pi}{3}]}(\theta) \big(1 - (1 - (1 - \theta/\pi)^k)^t\big) {\rm d}\theta.
\end{align*}
It can be shown (see~\cite[Lemma~10.5]{Laarhoven2016search}) that $p_1^\ast \geq 1-\varepsilon$ if $k = \log_{3/2}t - \log_{3/2}\ln(1/\varepsilon)$. On the other hand, the average probability $p_2^\ast$ that two vectors $\mathbf{v},\mathbf{w}\in\mathbb{R}^D$ with $\theta(\mathbf{v},\mathbf{w}) > \pi/3$ collide in at least one of the $t$ hash tables is
\begin{align}\label{eq:probability_collision_angular_hash}
    p_2^\ast = \operatorname*{Pr}_{h_{i,j}\sim\mathcal{H}_{\rm ang}}[\exists i\in[t], h_i(\mathbf{v}) = h_i(\mathbf{w})~|~\theta(\mathbf{v},\mathbf{w}) > \pi/3] = \int_{\frac{\pi}{3}}^{\frac{\pi}{2}} \Theta_{[\frac{\pi}{3},\frac{\pi}{2}]}(\theta) \big(1 - (1 - (1 - \theta/\pi)^k)^t\big) {\rm d}\theta.
\end{align}
It can be shown (see~\cite[Lemma~10.8]{Laarhoven2016search}) that $p_2^\ast \leq t\cdot 2^{-\beta D + o(D)}$ if $k = \log_{3/2}t + O(1)$, where 
\begin{align}\label{eq:beta_parameter_angular_hash}
    \beta = - \max_{\theta\in(\frac{\pi}{3},\frac{\pi}{2})}\left\{\log_2\sin\theta + \frac{\log_2{t}}{D\log_2(3/2)}\log_2(1-\theta/\pi) \right\} > 0.
\end{align}
Ultimately, the choice for $t$ will depend on the balance between the time hashing and the time searching, as we shall see in \cref{sec:sieving_algorithms}.

\subsubsection{Spherical LSH}

Another important hash family that can be used to improve sieving algorithms~\cite{laarhoven2015faster} is the spherical LSH proposed by Andoni et al.~\cite{Andoni2014beyond,Andoni2015optimal}. The spherical LSH partitions the unit sphere $\mathcal{S}^{D-1}$ by first sampling $u = 2^{\Theta(\sqrt{D})}$ vectors $\mathbf{g}_1,\dots,\mathbf{g}_u\in\mathbb{R}^D$ from a standard $D$-dimensional Gaussian distribution $\mathcal{N}(0,1)^D$. A hash region $\mathcal{R}_i$ is then associated to each $\mathbf{g}_i$ as
\begin{align*}
    \mathcal{R}_i = \{\mathbf{x}\in \mathcal{S}^{D-1}: \langle \mathbf{x},\mathbf{g}_i\rangle \geq D^{1/4}\}\setminus \bigcup_{j=1}^{i-1}\mathcal{R}_j, \qquad \forall i\in[u].
\end{align*}
This procedure sequentially ``carves'' spherical caps of radius $\sqrt{2} - o(1)$. The hash of a vector $\mathbf{v}$ is given by the index of the region $\mathcal{R}_i$ it lies in. Moreover, the choice of $u=2^{\Theta(\sqrt{D})}$ guarantees that the unit sphere is entirely covered by the hash regions with high probability since each hash region covers a fraction $2^{-\Theta(\sqrt{D})}$ of the sphere. Indeed, $\operatorname{Pr}_{\mathbf{g}\sim \mathcal{N}(0,1)^D}[\langle \mathbf{x},\mathbf{g}\rangle \geq D^{1/4}] \geq (2\pi)^{-1/2}(D^{-1/4} - D^{-3/4})e^{-\sqrt{D}/2}$ for any fixed point $\mathbf{x}\in\mathcal{S}^{D-1}$~\cite{Karger1998approximate}, and by following the argument in~\cite[Appendix~A.3]{andoni2015optimalarxiv}, $u=2^{\sqrt{D}}$ hash regions is enough to cover the unit sphere with failure probability super-exponentially small in $D$. Andoni et al.~\cite{Andoni2014beyond,Andoni2015optimal} proved that the collision probability for the spherical hash family $\mathcal{H}_{\rm sph}$ is
\begin{align*}
    p(\theta) = \exp\left(-\frac{\sqrt{D}}{2}\tan^2\left(\frac{\theta}{2}\right)(1+o(1))\right).
\end{align*}

Under the spherical hash family $\mathcal{H}_{\rm sph}$ with randomly sampled hash functions $h_{i,j}\in\mathcal{H}_{\rm sph}$, the average probability $p_1^\ast$ that two vectors $\mathbf{v},\mathbf{w}\in\mathbb{R}^D$ with $\theta(\mathbf{v},\mathbf{w}) \leq \pi/3$ collide in at least one of $t$ hash tables is
\begin{align*}
    p_1^\ast = \operatorname*{Pr}_{h_{i,j}\sim\mathcal{H}_{\rm sph}}[\mathbf{v},\mathbf{w} ~\text{collide}~|~\theta(\mathbf{v},\mathbf{w}) \leq \pi/3] 
    = \int_0^{\frac{\pi}{3}} \Theta_{[0,\frac{\pi}{3}]}(\theta) \bigg(1 - \Big(1 - e^{-\frac{k\sqrt{D}}{2}\tan^2\left(\frac{\theta}{2}\right)(1+o(1))}\Big)^t\bigg) {\rm d}\theta.
\end{align*}
It can be shown (see~\cite[Lemma~11.5]{Laarhoven2016search}) that $p_1^\ast \geq 1 - \varepsilon$ if $k=6(\ln{t} - \ln\ln(1/\varepsilon))/\sqrt{D}$. On the other hand, the average probability $p_2^\ast$ that two vectors $\mathbf{v},\mathbf{w}\in\mathbb{R}^D$ with $\theta(\mathbf{v},\mathbf{w}) > \pi/3$ collide in at least one of $t$ hash tables is
\begin{align}\label{eq:probability_collision_spherical_hash}
    p_2^\ast = \!\operatorname*{Pr}_{h_{i,j}\sim\mathcal{H}_{\rm sph}}[\mathbf{v},\mathbf{w} ~\text{collide}~|~\theta(\mathbf{v},\mathbf{w}) > \pi/3] 
    = \!\int_{\frac{\pi}{3}}^{\frac{\pi}{2}} \Theta_{[\frac{\pi}{3},\frac{\pi}{2}]}(\theta) \bigg( 1 - \Big(1 - e^{-\frac{k\sqrt{D}}{2}\tan^2\left(\frac{\theta}{2}\right)(1+o(1))}\Big)^t\bigg) {\rm d}\theta.
\end{align}
It can be shown (see~\cite[Lemma~11.6]{Laarhoven2016search}) that $p_2^\ast \leq 2^{-\beta D + o(D)}$ if $k = 6\ln(t)/\sqrt{D} + o(1)$, where
\begin{align}\label{eq:beta_parameter_spherical_hash}
    \beta = - \max_{\theta\in(\frac{\pi}{3},\frac{\pi}{2})}\left\{\log_2\sin\theta - \left(3\tan^2\left(\frac{\theta}{2}\right) - 1 \right)\frac{\log_2{t}}{D}\right\} > 0.
\end{align}

\subsubsection{Spherical LSF}
\label{sec:spherical_lsf}

Becker et al.~\cite{becker2016new} proposed the spherical LSF family akin to spherical LSH. In spherical LSF, a filter is constructed by drawing a random $\mathbf{a}\in\mathcal{S}^{D-1}$ and a vector $\mathbf{v}$ passes the filter if $\langle\mathbf{a}, \mathbf{v}\rangle \geq \alpha$ for some parameter $\alpha > 0$. In other words,
\begin{align*}
    \mathcal{F}_{\rm sph} = \{f_{\mathbf{a}}:\mathbb{R}^D \to \{0,1\} ~|~ \mathbf{a}\in\mathcal{S}^{D-1}\}, \qquad f_{\mathbf{a}}(\mathbf{v}) = \begin{cases}
        1 &\text{if}~\langle\mathbf{a}, \mathbf{v}\rangle \geq \alpha,\\
        0 &\text{if}~\langle\mathbf{a}, \mathbf{v}\rangle < \alpha.
    \end{cases}
\end{align*}
As shown by Becker et al.~\cite{becker2016new}, the collision probability for the spherical filter family $\mathcal{F}_{\rm sph}$ is
\begin{align}\label{eq:collision_probability_filter}
    p(\theta) = \mathcal{W}_D(\alpha,\alpha,\theta) = \exp\left(\frac{D}{2}\ln\left(1 - \frac{2\alpha^2}{1 + \cos\theta} \right)(1+o(1))\right),
\end{align}
while the collision probability of a vector with itself is
\begin{align*}
    p(0) = \mathcal{C}_D(\alpha) = \exp\left(\frac{D}{2}\ln(1 - \alpha^2 )(1+o(1))\right).
\end{align*}

Under the spherical filter family $\mathcal{F}_{\rm sph}$ with randomly sampled filters $f_{i,j}\in\mathcal{F}_{\rm sph}$, the average probability $p_1^\ast$ that two vectors $\mathbf{v},\mathbf{w}\in\mathbb{R}^D$ with $\theta(\mathbf{v},\mathbf{w}) \leq \pi/3$ collide in at least one of $t$ filters is
\begin{align}\label{eq:probability_collision_spherical_filter}
    p_1^\ast = \operatorname*{Pr}_{f_{i,j}\sim \mathcal{F}_{\rm sph}}[\exists i\in[t],\mathbf{v},\mathbf{w}\in L_{f_i}~|~\theta(\mathbf{v},\mathbf{w}) \leq \pi/3] = \int_0^{\frac{\pi}{3}} \Theta_{[0,\frac{\pi}{3}]}(\theta) (1 - (1 - \mathcal{W}_D(\alpha,\alpha,\theta)^k)^t){\rm d}\theta.
\end{align}
Since $\mathcal{W}_D(\alpha,\alpha,\theta)$ is decreasing in $\theta$, it is not hard to see that $p_1^\ast \geq 1 - (1 - \mathcal{W}_D(\alpha,\alpha,\pi/3)^k)^t$. Therefore, $p_1^\ast \geq 1-\varepsilon$ if $t\geq \ln(1/\varepsilon)/\ln(1/(1-\mathcal{W}_D(\alpha,\alpha,\pi/3)^k))$. Regarding the choice for $k$, the trivial lower bound $k\geq 1$ leads to an upper bound on $\alpha$, which is normally the optimal choice, see~\cite{becker2016new} for more information. This means that we shall take $k=1$ in the above expressions.

LSF methods usually yield better asymptotic complexities when it comes to sieving algorithms, as shown in \cref{sec:sieving_algorithms}. However, a crucial assumption for the use of filter families over hash families is the existence of an efficient oracle that identifies any of the concatenated filters a vector passes through in time proportional to the number of relevant filters out of all concatenated filters. Becker et al.~\cite{becker2016new} developed such an oracle, called $\mathtt{EfficientListDecoding}$, by employing random product codes to efficiently obtain the set of relevant filters, which only mildly affects the overall complexities. The complexity of their oracle is summarised below.
\begin{fact}[{\cite[Lemma~5.1]{becker2016new}}]\label{fact:EfficientListDecoding}
    Let $t = 2^{\Omega(D)}$ be the number of filter buckets. There is an algorithm that returns the set of filters that a given vectors passes in average time $O(\log_2{D}\cdot t\cdot \mathcal{C}_D(\alpha))$ by mainly visiting at most $2\log_2{D}\cdot t\cdot \mathcal{C}_D(\alpha)$ nodes for a pruned enumeration.
\end{fact}

\section{Quantum error correction}
\label{sec:error_correction}

Quantum circuits are usually described on a logical level by applying logical gates onto logical qubits. If one wants to implement a quantum circuit in actual physical devices, then noise should be taken into consideration. This is not only valid for classical devices, but especially true for quantum computers, where exquisite control of quantum systems is severely affected by noise. One of the greatest breakthroughs of the 90s was the realisation that redundancy could also be introduced into quantum systems to protect them against several types of noise, and therefore quantum error-correction codes exist. Starting with Shor's nine-qubit code~\cite{shor1995scheme}, several simple quantum error-correction codes were soon discovered, e.g., Steane's seven-qubit code~\cite{steane1996error}, the five-qubit code~\cite{bennett1996mixed,laflamme1996perfect}, and the CSS (Calderbank-Shor-Steane) codes~\cite{calderbank1996good,steane1996multiple}. All these codes are examples of stabiliser codes, i.e., quantum error-correction codes based on the stabiliser formalism invented by Gottessman~\cite{gottesman1996class,gottesman1997stabilizer}. In any quantum error-correction code, a set of \emph{physical} qubits are entangled in particular states and these joint states are interpreted as \emph{logical} qubits. As an example, in Shor's code~\cite{shor1995scheme} $|0_L\rangle$ is encoded as $(|000\rangle + |111\rangle)^{\otimes 3}/2\sqrt{2}$ and $|1_L\rangle$ is encoded as $(|000\rangle - |111\rangle)^{\otimes 3}/2\sqrt{2}$, which protects against an arbitrary error on a single qubit.

\subsection{Physical error model}

Several properties of quantum error-correction codes are functions of the underlying physical error model. In this work, we assume incoherent circuit-level noise for the physical qubits, meaning that each physical gate, state initialisation, and measurement outcome is affected by a random Pauli error with probability $p_{\rm phy}$. More precisely, at any point of a quantum circuit, the quantum state of a physical qubit is mapped according to
\begin{align*}
    \rho \mapsto (1- p_{\rm phy}) \rho + \frac{p_{\rm phy}}{3}\cdot \mathsf{X}\rho \mathsf{X} + \frac{p_{\rm phy}}{3}\cdot \mathsf{Y}\rho \mathsf{Y} + \frac{p_{\rm phy}}{3}\cdot \mathsf{Z}\rho \mathsf{Z}.
\end{align*}
Even though two-qubit gates are more prone to errors than single-qubit gates, we consider a single characteristic error rate $p_{\rm phys}$ for both types of gate in circuit-level noise. We will assume that $p_{\rm phy} = 10^{-5}$ throughout, which is an optimistic but not unrealistic assumption~\cite{ballance2016high,gaebler2016high,madjarov2020high,clark2021high,moskalenko2022high,ruskuc2022nuclear}.

\subsection{Surface codes}

The codes mentioned above are only resilient to very small physical errors. This was later greatly improved with the introduction of surface codes by Kitaev~\cite{kitaev1997quantum,kitaev2003fault}. In surface codes, the physical qubits are arranged in a two-dimensional array on a surface of non-trivial topology, e.g., a plane or a torus, and quantum operations are associated with non-trivial homology cycles of the surface. Surface codes have several appealing properties, e.g., very high error tolerance~\cite{wang2003confinement,ohno2004phase} and local check (stabiliser) measurements. We will not review the surface code in detail here, but we will quote important properties that will be used in our analysis. For more details on the surface code, see~\cite{dennis2002topological,fowler2012surface,Litinski2019gameofsurfacecodes,Cleland2022introduction}.

There are a few different encoding schemes for surface codes, e.g., defect-based~\cite{fowler2012surface}, twist-based~\cite{bombin2010topological}, and patch-based~\cite{horsman2012surface} encodings. Here we shall work exclusively with the latter, since surface-code patches offer lower space overhead and low-overhead Clifford gates~\cite{brown2017poking,litinski2018latticesurgery}. A (rotated) surface-code patch of distance $d$ employs $d^2$ physical qubits to encode one logical qubit and is able to correct arbitrary errors on any $\lfloor (d-1)/2\rfloor$ qubits. In order to extract information from a surface code and check for errors, its check operators are measured, which requires $d^2$ extra measurement qubits for a total of $2d^2$ physical qubits. Moreover, the subroutine of measuring check operators naturally sets a time scale in any experiment. By a \emph{code cycle} we mean the time required to measure all surface-code check operators. It is also common to define a \emph{logical cycle} as $d$ code cycles~\cite{Litinski2019gameofsurfacecodes}, since $\Omega(d)$ check-operator measurements are required to successfully discern measurement errors from physical errors. We will assume that a quantum computer can perform one code cycle every $100$~ns, which is quite an optimistic but not unrealistic assumption~\cite{Chen2021exponential,Anderson2021realization,Krinner2022realizing,Battistel2023realtime}.

One of the main results of the theory of quantum fault tolerance is the \emph{threshold theorem}~\cite{knill1996concatenated,Aharonov1997fault,kitaev1997quantum,knill1998resilient,Preskill1998reliable,gottesman1997stabilizer,nielsen2010quantum}. On a high level, it states that, under some reasonable assumptions about the noise of the underlying hardware, an arbitrary long quantum computation can be carried out with arbitrarily high reliability, provided the error rate $p_{\rm phy}$ per quantum gate is below a certain critical \emph{threshold} value $p_{\rm th}$. Applied to the surface code specifically, the threshold theorem states that the probability of a \emph{logical error} occurring on a distance-$d$ surface code after measuring the check operators and correcting for the observed physical errors vanishes exponentially with the distance $d$ as long as $p_{\rm phy}$ is below the threshold $p_{\rm th}$. This means that a quantum computation can be made arbitrarily reliant by increasing the distance of the surface-code patch. The surface code exhibits a very high threshold~\cite{wang2003confinement,ohno2004phase,Stace2010error} for most error models, and has a threshold of approximately $1\%$ for circuit-level noise~\cite{Wang2011surface,Stephens2014fault}. Therefore, under a circuit-level noise model with error $p_{\rm phy}$, the logical error rate per logical qubit per code cycle can be approximated as~\cite{fowler2018low}
\begin{align*}
    p_L(p_{\rm phy}, d) = 0.1(100 p_{\rm phy})^{(d+1)/2}.
\end{align*}
If we wish that $n$ logical qubits survive for $T$ code cycles with high probability, say $99.9\%$, then the probability that a logical error affects any logical qubit during all code cycles must be smaller than $0.1\%$. This determines the required code distance $d$ when encoding the logical qubits as
\begin{align*}
    T\cdot n \cdot p_L(p_{\rm phy},d) < 0.001.
\end{align*}

\subsection{Baseline architecture vs active-volume architecture}

It is necessary to specify a physical architecture for a general-purpose fault-tolerant quantum computer which, together with a compilation scheme, converts quantum computations into instructions for that architecture. There are mainly two types of architectures that will be taken into consideration in this work: baseline architectures with nearest-neighbor logical two-qubit interactions on a 2D grid~\cite{Litinski2019gameofsurfacecodes,fowler2018low,Chamberland2022universal,Chamberland2022building,bombin2021interleaving}, and the active-volume architecture~\cite{litinski2022active} that employs a logarithmic number of non-local connections between logical qubits. 

In baseline architectures, the most relevant parameters are the number of data qubits $n_Q$ (i.e., the number of logical qubits on a circuit-level quantum computation) and the number of non-Clifford gates, which in our case is the number of $\mathsf{Toffoli}$ gates $n_{\rm Toff}$. This is because all Clifford gates can be commuted to the end of the computation and be absorbed by final measurements~\cite{Litinski2019gameofsurfacecodes}. 
Both quantities $n_Q$ and $n_{\rm Toff}$ define the \emph{circuit volume} $n_Q \cdot n_{\rm Toff}$, which is proportional to the \emph{spacetime volume cost} of the quantum computation, i.e., the total number of logical qubits taking into consideration space overheads multiplied by the total number of logical cycles. In baseline architectures, a $n_Q$-qubit quantum computation consists roughly of $2n_Q$ logical qubits. To be more precise, using Litinksi's fast data blocks~\cite{Litinski2019gameofsurfacecodes}, an $n_Q$-qubit quantum computation requires $2n_Q + \sqrt{8n_Q} + 1$ logical qubits in total (in order to efficiently consume magic states). On the other hand, one $\mathsf{Toffoli}$ gate is executed in $6$ logical cycles, or in $4$ logical cycles if the target qubit is in the $|0\rangle$ state, which will be the case of almost all $\mathsf{Toffoli}$ gates in our circuits.

The figure of merit in baseline architectures is the circuit volume. Most of the time, however, a large portion of the circuit volume is \emph{idle volume}, i.e., volume attributed to qubits that are not part of an operation at a certain time and are thus idling. Since idling qubits have the same cost as active qubits when using surface codes, the cost of logical operations scales with the number of logical qubits $n_Q$. In active-volume architectures, on the other hand, only active qubits contribute to the spacetime volume cost of a quantum computation. More specifically, in active-volume architectures, a quantum computer is made up of modules with $d^2$ physical qubits. Each module can operate as memory or a workspace module. A memory module increases the memory capacity by one logical qubit, and a workspace module increases the computational speed by one \emph{logical block} per logical cycle. An operation is measured in terms of logical blocks and its cost is basically the amount of workspace modules it requires per logical cycle. We assume that $n_Q$ logical qubits result in $n_Q/2$ memory qubits and a speed of $n_Q/2$ logical blocks per logical cycle. The figure of merit in an active-volume architecture is the number of logical blocks, called \emph{active volume}. In order to obtain the total active volume of a quantum computation, we must simply sum up all the active volume of its constituent operations, several of which were given in~\cite{litinski2022active}, e.g., a $\mathsf{Toffoli}$ has an active volume of $12$ plus the active volume of distilling a magic state (see \Cref{sec:magic_state_distillation}). Litinski and Nickerson~\cite{litinski2022active} proposed a general-purpose active-volume architecture that executes quantum computations with spacetime volume cost of roughly twice the active volume. Contrary to baseline architectures, active-volume ones rely on non-local connections between components, which allows for several fast operations. As an example, Bell measurements can be performed in one code cycle, while in baseline architectures it requires $2$ logical cycles via lattice surgery~\cite{horsman2012surface,litinski2018latticesurgery,fowler2018low}. We point the reader to~\cite{litinski2022active} for a detailed list of assumptions.

Common to both architectures is the time required to perform a layer of single-qubit measurements (or Bell measurements in active-volume architecture), feed the measurement outcomes into a classical decoder, perform a classical decoding algorithm like minimum-weight perfect matching~\cite{Edmonds1965paths,dennis2002topological} or union-find~\cite{delfosse2020lineartime,Delfosse2021almostlineartime}, and use the result to send new instructions to the quantum computer, which is called \emph{reaction time} $\tau_r$. In this work, we shall assume a reaction time of $1$~$\mu$s, which is an optimistic assumption, as most previous works assume a reaction time of $10$~$\mu$s~\cite{Gidney_2021,litinski2023compute}. Related to the reaction time is the \emph{reaction depth} of a quantum computation, which is the number of reaction layers, i.e., layers of reactive measurements that must be classically decoded and fed back into the circuit. We thus distinguish between the time required to execute all gates in a circuit when the reaction time is zero (\emph{circuit time}) and the reaction depth times the reaction time (\emph{circuit reaction (time) limit}).

\subsection{Magic state distillation}
\label{sec:magic_state_distillation}

It is known that no quantum error-correction code can transversally implement a universal gate set~\cite{eastinKnill09}, i.e., be physically implemented on a logical qubit by independent actions of single-qubit physical gates on a subset of the physical qubits. For surface codes, this means $\mathsf{T}$ and $\mathsf{CCZ}$ gates, among others. In order to overcome this problem, a resource state is first prepared separately and subsequentially consumed to execute a non-transversal gate like a $\mathsf{T}$ or a $\mathsf{CCZ}$ gate~\cite{Bravyi2005universal}. For $\mathsf{T}$ gates, the resource state is a magic state $|T\rangle = (|0\rangle + e^{i\pi/4}|1\rangle)/\sqrt{2}$, while for $\mathsf{CCZ}$ gates the resource state is $|CCZ\rangle = \mathsf{CCZ}|+\rangle^{\otimes 3}$. A magic state $|T\rangle$ can be used to perform a $\mathsf{T}$ gate by measuring the logical Pauli product $\mathsf{Z}\otimes \mathsf{Z}$ acting on an input state and the magic state~\cite{Litinski2018quantum,Litinski2019gameofsurfacecodes,Litinski2019magicstate} akin to teleportation protocols (cf.~\cite[Figure~10.25]{nielsen2010quantum}). A similar procedure can be used to perform a $\mathsf{CCZ}$ gate by consuming one $|CCZ\rangle$ state (see~\cite[Figure~14(a)]{litinski2022active}). However, applying a physical $\mathsf{T}$ or $\mathsf{CCZ}$ gate onto a few physical qubits yields a resource state with physical error $p_{\rm phy}$. If this resource state is then used to perform a logical gate, the error rate of the logical gate will be proportional to $p_{\rm phy}$, which can be too high for a long computation and will spoil the final outcome. One common procedure to generate low-error magic states is to employ magic state distillation protocols~\cite{Bravyi2005universal}.

Magic state distillation is a short error-detecting quantum procedure to generate a low-error magic state from several high-error magic state copies. First introduced by Bravyi and Kitaev~\cite{Bravyi2005universal} and Reichardt~\cite{Reichardt2005quantum}, several different protocols have since been developed~\cite{Bravyi2012magic,Fowler2013surface,Meier2013magic,Jones2013multilevel,duclos_cianci2013distillation,duclos_cianci2015reducing,Campbell2017unified,OGorman2017quantum,Haah2018codesprotocols,Campbell2018magicstateparity,Litinski2019gameofsurfacecodes,Litinski2019magicstate}. There are a few different but equivalent ways to understand magic state distillation. One is to create the logical magic state using an error-correction code with transversal $\mathsf{T}$ gates, e.g., punctured Reed-Muller codes~\cite{Bravyi2005universal,Haah2018codesprotocols} or code-blocks~\cite{Bravyi2012magic,Jones2013multilevel,Fowler2013surface}. As an example, the 15-to-1 distillation procedure~\cite{Bravyi2005universal,Reichardt2005quantum,fowler2018low} employs a punctured Reed-Muller code to first encode a logical $|+_L\rangle$ state within $15$ physical qubits. The transversallity of the code allows to perform a logical $\mathsf{T}_L$ gate onto $|+_L\rangle$ from individual physical $\mathsf{T}$ gates, which yields $\mathsf{T}_L|+_L\rangle$. The encoding procedure is then uncomputed and the logical information is shifted to one of the physical qubits. Measuring the remaining physical qubits gives information on possible errors and on whether the procedure was successful or not. If the error probability of the $15$ $\mathsf{T}$ gates is $p_{\rm phy}$, then the error probability of the output state is $35p_{\rm phy}^3$, where the factor $35$ comes from different error configurations that are not detectable by the protocol. As a result, $15$ magic states with error $p_{\rm phy}$ are distilled down to one magic state with error $35p_{\rm phy}^3$. A similar distillation procedure exists for creating a low-error $|CCZ\rangle$ state, e.g., Gidney and Fowler~\cite{Gidney2019efficientmagicstate} proposed an 8-to-CCZ distillation protocol to output a $|CCZ\rangle$ state with error $28p_{\rm phy}^2$ from $8$ $|T\rangle$ states with error $p_{\rm phy}$.
In order to achieve lower error rates than $35p_{\rm phy}^3$ or $28p_{\rm phy}^2$, it is possible to concatenate different distillation protocols, meaning that the output states of a level-$1$ distillation protocol can serve as input magic states for a level-$2$ distillation protocol, etc.

For baseline architectures, we shall employ the magic state distillation protocols from Litinski~\cite{Litinski2019magicstate} which are, as far as we are aware, one of the best to this day. Litinski's protocols are characterised by three code distances $d_X$, $d_Z$, $d_m$ from several internal patches. As shown in~\cite{Litinski2019magicstate}, a $(15\text{-to-}1)_{d_X,d_Z,d_m}$ distillation protocol outputs a low-error magic state every $6d_m$ code cycles using $2(d_X + 4d_Z)\cdot 3d_X + 4d_m$ physical qubits. 
Similarly, a two-level protocol is described by three additional code distances $d_{X2}$, $d_{Z2}$, and $d_{m2}$, plus the number $n_{L1}$ of level-1 distillation blocks, where $n_{L1}$ is an even integer. 
As an example quoted from Litinski's paper~\cite[Table~1]{Litinski2019magicstate}, if $p_{\rm phy} = 10^{-4}$, then the $(15\text{-to-}1)^4_{7,3,3}\times (8\text{-to-CCZ})_{15,7,9}$ protocol outputs a $|CCZ\rangle$ state with error $p_{\rm out} = 7.2\cdot 10^{-14}$ in $36.1$ code cycles using $12,400$ physical qubits. 
As will be clear in \Cref{sec:results}, we shall require higher-than-two-level protocols to achieve error rates below $10^{-40}$. Even though Litinski~\cite{Litinski2019magicstate} only focuses on one and two-level distillation protocols, it is not hard to continue with their analysis and derive the resources required for a three-level distillation protocol: we simply input level-$2$ magic states into a level-$3$ protocol with code parameters $d_{X3}$, $d_{Z3}$, and $d_{m3}$, plus the number $n_{L2}$ of level-2 distillation blocks. When optimising the code distances, one usually finds that $d_X = d$, $d_Z \approx d/2$, $d_m \approx d/2$~\cite{Litinski2019magicstate,litinski2022active}. We shall then consider concatenated protocols of the form $(15\text{-to-}1)_{d/4,d/8,d/8}^{n_{L1}}\times(15\text{-to-}1)_{d/2,d/4,d/4}^{n_{L2}}\times (8\text{-to-CCZ})_{d,d/2,d/2}$. 

Regarding active-volume architectures, on the other hand, we employ the distillation protocols from~\cite{litinski2022active} of the form $(15\text{-to-}1)_{d,d/2,d/2}$ and $(8\text{-to-CCZ})_{d,d,d/2}$. Given a quantum computation with logical blocks of distance $d$, then a $(15\text{-to-}1)_{ad,ad/2,ad/2}$ protocol has an active volume of $35a^2/2$, while a $(8\text{-to-CCZ})_{ad,ad,ad/2}$ protocol has an active volume of $25a^2/2$. Therefore, a $(15\text{-to-}1)_{d/4,d/8,d/8}^{n_{L1}}\times(15\text{-to-}1)_{d/2,d/4,d/4}^{n_{L2}}\times (8\text{-to-CCZ})_{d,d,d/2}$ protocol has an active volume of $\frac{35}{32}n_{L_1}n_{L_2} + \frac{35}{8}n_{L_2} + \frac{25}{2}$. By using $n_{L_1} = 8$ level-1 protocols and $n_{L_2} = 4$ level-2 protocols, $16$ level-1 distilled $|T\rangle$ states are produced every $d/4$ code cycles, and $8$ level-2 distilled $|T\rangle$ states are produced every $d/2$ code cycles, meaning that one level-3 distilled $|CCZ\rangle$ can be produced every $d$ code cycles. Therefore, the $(15\text{-to-}1)_{d/4,d/8,d/8}^{8}\times(15\text{-to-}1)_{d/2,d/4,d/4}^{4}\times (8\text{-to-CCZ})_{d,d,d/2}$ protocol has an active volume of $65$ and produces a $|CCZ\rangle$ state every logical cycle. The output error can be calculated using the approximate expressions in~\cite{litinski2022active}, or using the Python file for the baseline-architecture distillation protocols from~\cite{Litinski2019magicstate}.

\section{Arithmetic on a quantum computer}
\label{sec:arithmetic}

In this section, we turn our attention to the resources needed to perform some simple arithmetic operations on a quantum computer that will be the building blocks for the analysis of quantum sieving. But first, we need a way to store $D$-dimensional vectors with integer entries in a quantum computer. In order to do that, we store the \emph{two's-complement} representation $x_{\kappa-1}\dots x_0$ of a $\kappa$-bit integer $x$ in a $\kappa$-qubit quantum register $\ket{x_{\kappa-1},x_{\kappa-2},\dots, x_0}$, where $x_0, \dots x_{\kappa-1}\in\{0,1\}$. We do not use a sign-magnitude representation for an integer $x = (-1)^{x_{\kappa-1}}(x_{\kappa-2}\cdot 2^{\kappa-2} + \dots + x_0 \cdot 2^0)$, as done by other works~\cite{Yang2023quantumalphatron,doriguello2022quantum}, since addition is non trivial in such representation and several known quantum adders would have to be modified to take negative numbers into consideration. The value of $\kappa$ is chosen in advance and remains the same throughout the whole computation. Increasing the value of $\kappa$ of course requires more physical resources for the algorithm execution but at the same time reduces the chance of an overflow occurring. Throughout this work, we assume $\kappa=32$, which translates to a capacity of working with integers in the range $[-2^{31}, 2^{31}-1]$. To store an entire $D$-dimensional vector, we store each of its entries separately using the above encoding, so that $D \kappa$ qubits are required in total.

We now start with reviewing fundamental arithmetic operations on a quantum computer: addition, comparison, and multiplication.

\subsection{Quantum adders}

An \emph{out-of-place quantum adder} (modulo $2^\kappa$) is a unitary that adds two $\kappa$-bit integers $x = x_{\kappa-1}\dots x_0$ and $y = y_{\kappa-1}\dots y_0$ together onto a third register,
\begin{align*}
    |x_{\kappa-1},\dots,x_0\rangle|y_{\kappa-1},\dots,y_0\rangle|0\rangle^{\otimes \kappa} \mapsto |x_{\kappa-1},\dots,x_0\rangle|y_{\kappa-1},\dots,y_0\rangle |(x+y)_{\kappa-1}, \dots, (x+y)_0\rangle.
\end{align*}
It is possible to define an \emph{in-place quantum adder} which replaces one of the inputs with the outcome, but in this work we shall focus on out-of-place adders since they have a lower $\mathsf{Toffoli}$-count~\cite{gidney2018halving}.

Several quantum adders or related circuits have been proposed in the past few decades~\cite{beckman1996efficient,gossett1998quantum,vedral1996quantum,zalka1998fast,draper2000addition,cuccaro2004new,draper2006logarithmic,lin2014qlib,jayashree2016ancilla,amy2013meet,jones2013low,gidney2018halving,munoz2019quantum,li2020efficient,li2022circuit}, see~\cite{orts2020review} for a review. As far as we are aware, the state-of-the-art quantum adder in terms of $\mathsf{Toffoli}$-count is due to Gidney~\cite{gidney2018halving}, which is an improved version of Cuccaro's adder~\cite{cuccaro2004new}. Gidney's adder (\cref{fig:adder}) concatenates several copies of the adder building-block, each of which is made of one $\mathsf{Toffoli}$ computation and its uncomputation requiring no $\mathsf{Toffoli}$ gates. In order to add two $\kappa$-bit integers, Gidney's adder requires $\kappa-1$ $\mathsf{Toffoli}$ gates in total. Even though Gidney's results are phrased in terms of $\mathsf{T}$ gates, we translate them into $\mathsf{Toffoli}$ gates. The $\mathsf{Toffoli}$-count, together with several other quantities like $\mathsf{Toffoli}$-width (maximum number of $\mathsf{Toffoli}$ gates in a single layer), reaction depth, number of logical qubits are shown in \cref{table:arithmetics}. Its active volume, on the other hand, was computed by Litinski and Nickerson~\cite[Table~1]{litinski2022active} and equals to $(\kappa-1)(39+C_{|CCZ\rangle}) + 7$, where $C_{|CCZ\rangle}$ is the active volume of distilling one $|CCZ\rangle$ state.

\begin{figure}[t]
    \centering
    \includegraphics[scale=0.8, trim={1.5cm, 21.5cm, 6.1cm, 0.85cm}, clip=true]{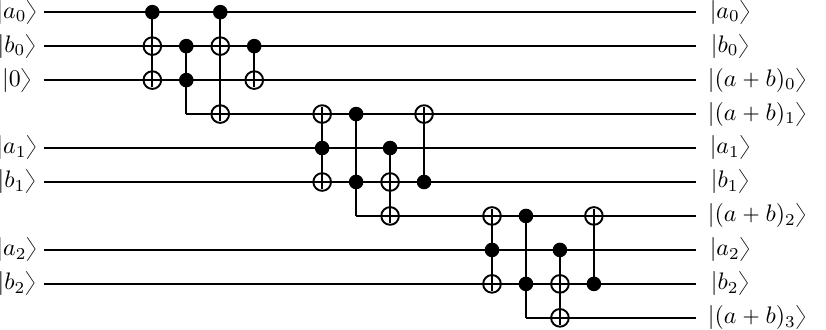}
    \caption{Gidney's out-of-place quantum adder (modulo $2^\kappa$) that adds two $\kappa$-bit numbers $a$ and $b$ stored in quantum registers.}
    \label{fig:adder}
\end{figure}

Using Gidney's quantum (out-of-place) adder, it is easy to develop a quantum controlled (out-of-place) adder (modulo $2^\kappa$): first apply the vanilla adder to get $|c\rangle |x\rangle |y\rangle |0\rangle^{\otimes 2\kappa} \mapsto |c\rangle |x\rangle |y\rangle |x+y\rangle |0\rangle^{\otimes \kappa}$, followed by $\kappa$ $\mathsf{Toffoli}$ gates to copy each bit $(a+b)_i$ onto another register controlled on $c\in\{0,1\}$. This yields $|c\rangle |a\rangle |x\rangle |x+y\rangle |c(x+y)\rangle$. It is possible to uncompute the ancillary register $|x+y\rangle$ by performing the inverse of the first adder, which uses no $\mathsf{Toffoli}$ gates. However, if we keep such ancillary register, uncomputing the whole controlled adder requires no $\mathsf{Toffoli}$ gates, as opposed to $2\kappa+O(1)$ if you call the inverse of the entire controlled adder. Therefore, we shall keep the ancillary register $|x+y\rangle$ until the uncomputation of the whole circuit. Finally, we note that the controlled copying of the ancillary register $|x+y\rangle$ can be done while the out-of-place adder is being performed. The active volume of the whole computation, while not considered by~\cite{litinski2022active}, can be easily calculated from the its separated parts. The results are described in \cref{table:arithmetics}.

\subsection{Quantum comparator}

A quantum comparator is a unitary that compares whether a $\kappa$-bit integer $x = x_{\kappa-1}\dots x_0$ is bigger than another $\kappa$-bit integer $y = y_{\kappa-1}\dots y_0$,
\begin{align*}
    |x\rangle|y\rangle|0\rangle \mapsto |x\rangle|y\rangle|\mathbf{1}[x > y]\rangle.
\end{align*}
A comparator can be obtained from the highest-order bit of the difference $x-y$. Whether we use one’s-complement or two’s-complement arithmetic, the identity $x - y = \overline{\overline{x} + y}$ holds. Therefore, it is possible to use an out-of-place adder as a comparator: complement $x$, employ a quantum adder and keep the highest-order bit, and complement the obtained highest-order bit. All the adders described in the previous section are modulo $2^\kappa$, meaning that the highest-order bit is not calculated. Nonetheless, we shall assume that there is no overflow and therefore the highest-order bit of the summation modulo $2^\kappa$ yields the correct answer. Moreover, if one of the inputs is classical, say $y$, then there is no need to complement the quantum register holding $x$, except maybe for the highest-order bit of $x-y$ depending on whether we are checking for $x>y$ or $x<y$. 

\subsection{Quantum multipliers}

Similarly to addition, we can define an \emph{out-of-place quantum multiplier} (modulo~$2^\kappa$) as the unitary that multiplies two $\kappa$-bit integers $x=x_{\kappa - 1}\dots x_0$ and $y=y_{\kappa - 1}\dots y_0$ together and places the outcome on a third register,
\begin{align*}
    |x_{\kappa-1},\dots,x_0\rangle|y_{\kappa-1},\dots,y_0\rangle|0\rangle^{\otimes \kappa} \mapsto |x_{\kappa-1},\dots,x_0\rangle|y_{\kappa-1},\dots,y_0\rangle |(x\cdot y)_{\kappa-1}, \dots, (x\cdot y)_0\rangle.
\end{align*}
Several quantum multipliers have been proposed in the past decade~\cite{lin2014qlib,jayashree2016ancilla,babu2016cost,kotiyal2014circuit,ruizperez2017quantum,munoz2019quantum,PourAliAkbar2019efficient,li2020efficient,Gayathri2021Tcount,li2022circuit,orts2023improving,Manochehri2024regular}. In terms of $\mathsf{T}$-count, the works of Li et al.~\cite{li2022circuit} and Orts et al.~\cite{orts2023improving} are the best as far as we are aware. Li et al.~\cite{li2022circuit} proposed a quantum multiplier with $16\kappa^2 -14\kappa$ $\mathsf{T}$ gates, $\kappa + 1$ ancillae, and $\mathsf{T}$-depth of $4\kappa^2 + 4\kappa + 4$. Orts et al.~\cite{orts2023improving}, on the other hand, proposed a quantum multiplier with $18\kappa^2 - 24\kappa$ $\mathsf{T}$ gates, $2\kappa^2 - 2\kappa + 2$ ancillae, and $\mathsf{T}$-depth of $14\kappa - 14$. Both $\mathsf{T}$-counts are comparable, while the trade-off is between ancillae and $\mathsf{T}$-depth.

\begin{table}[t]
    \footnotesize
    \centering
    \caption{State-of-the-art constructions for several quantum arithmetic circuits on $\kappa$-bit integers. All operations are out-of-place, modulo $2^\kappa$, and already include their inverses. The resources are broken down into $\mathsf{Toffoli}$-count, $\mathsf{Toffoli}$-width, reaction depth, qubit-width (ancillae plus input/output qubits), and active volume. Here $C_{|CCZ\rangle}$ is the active volume of distilling one $|CCZ\rangle$ state.}
    \label{table:arithmetics}
    \def\arraystretch{2}
    \resizebox{\linewidth}{!}{
    \begin{tabular}{ | c || c | c | c | c | c | }
        \hline

        Circuit / Resource & $\mathsf{Toffoli}$-count & $\mathsf{Toffoli}$-width & Reaction depth & Qubit-width & Active volume \\ \hline\hline

        Adder/Comparator & $\kappa-1$ & $1$ & $2(\kappa - 1)$ & $3\kappa$ & $(\kappa-1)(39+C_{|CCZ\rangle}) + 7$ \\ \hline

        Controlled adder & $2\kappa - 1$ & $\kappa$ & $2\kappa$ & $4\kappa+1$ & $(\kappa-1)(51 + C_{|CCZ\rangle}) + 19$ \\ \hline

        Multiplier & $\kappa^2 - \kappa + 1$ & $0.5\kappa^2 + 0.5\kappa$ & $2\kappa\log_2\kappa - 2\kappa - 2\log_2\kappa + 4$ & $2\kappa^2+\kappa$ & \makecell{$28\kappa^2 - 42\kappa + 28$ \\ $+(\kappa^2 - \kappa + 1)C_{|CCZ\rangle}$} \\ \hline

        Multiplier (hybrid) & $0.5\kappa^2 - 1.5\kappa + 1$ & $0.5\kappa$ & $2\kappa\log_2\kappa - 2\kappa - 2\log_2\kappa + 2$ & $1.5\kappa^2+0.5\kappa$ & \makecell{$20.25\kappa^2 - 48.75\kappa + 32$ \\ $+(0.5\kappa^2 - 1.5\kappa + 1)C_{|CCZ\rangle}$ }\\ \hline
 
    \end{tabular}
    }
\end{table}

Since we are mostly concerned with $\mathsf{Toffoli}$-count and are willing to use extra ancillae (including keeping dirty ones for subsequent uncomputation), we employ a quantum multiplier based on schoolbook multiplication with $\kappa-1$ out-of-place additions from \cref{table:arithmetics}. We note that a similar idea appeared before in~\cite{Sanders2020compilation}, although not modulo $2^\kappa$, and only very recently, by the time this manuscript was finalised, a similar construction with a similar $\mathsf{Toffoli}$-count was proposed by Litinski~\cite{litinski2024quantum}.

The multiplier works as follows. The input registers $|x_{\kappa-1},\dots,x_0\rangle$ and $|y_{\kappa-1},\dots,y_0\rangle$ are first copied $\kappa-1$ times: the bits $x_i$ and $y_i$ are copied $\kappa - 1 - i$ times, $i=0,\dots,\kappa-2$. This can be done with $\kappa^2 - \kappa$ $\mathsf{CNOT}$s in depth $\lceil\log_2{\kappa} \rceil$. We do not need to copy $x_i$ and $y_i$ a number of $\kappa-1$ times since the multiplication is done modulo $2^\kappa$ and high-order bits are ignored. We then perform $\kappa$ steps in parallel: in the $i$-th step, $i=0,\dots,\kappa-1$, the qubits $|x_i,\dots,x_0\rangle$ are copied onto fresh ancillae $|0\rangle^{\otimes (i+1)}$ controlled on one copy of $y_{\kappa-1-i}$ using $\mathsf{Toffoli}$ gates. At the end of this process, we have $\kappa$ registers holding all partial sums: the $i$-th one made up of $i+1$ bits, $i=0,\dots,\kappa-1$. Then, the $\kappa$ partial sums are added up using out-of-places adders until the final sum is computed. This can be done in any particular order, the amount of resources is left unchanged except for the reaction depth. The optimal reaction depth combination is tree-wise in $\lceil\log_2\kappa\rceil$ layers. For simplicity of analysis, let us assume the combination is done sequentially. More precisely, at layer $i=1,\dots,\kappa-1$, the sum of the previous layer, which has $i$ bits, is added onto the partial sum with $i+1$ bits, which requires an $i$-bit out-of-place adder (the least significant digit of the second register is just attached to the result register to form the $(i+1)$-bit answer). This means that the total $\mathsf{Toffoli}$-count (already taking into account the $\sum_{i=0}^{\kappa-1} (i+1) = (\kappa^2 + \kappa)/2$ $\mathsf{Toffoli}$ gates from the controlled copying) is 
\begin{align*}
    \frac{\kappa^2 + \kappa}{2} + \sum_{i=1}^{\kappa-1} (i-1) = \kappa^2 - \kappa + 1.   
\end{align*}
By keeping all dirty ancillae from the computation, the inverse circuit can be implemented with no $\mathsf{Toffoli}$ gates! Regarding ancillae, the initial copying requires $\kappa^2 - \kappa$ ancillae, while the controlled copying requires another $(\kappa^2 + \kappa)/2$ ancillae. The $\kappa-1$ out-of-place adders require $\sum_{i=1}^{\kappa-1} i = (\kappa^2 - \kappa)/2$ ancillae, $\kappa$ of which will be the output. There are thus $2\kappa^2 - 2\kappa$ dirty ancillae, and the total width is $2\kappa^2 + \kappa$ (ancillae plus $3\kappa$ input and output qubits). 

The active-volume calculation is similar to the $\mathsf{Toffoli}$-count. The $2\kappa^2 - 2\kappa$ $\mathsf{CNOT}$s (taking into consideration the inverse) have an active volume of $2\sum_{i=0}^{\kappa-2}(\frac{3}{2}(\kappa - 1 - i) + 2) = 1.5\kappa^2 + 2.5\kappa - 4$, while the $(\kappa^2 + \kappa)/2$ $\mathsf{Toffoli}$ gates from the controlled copying (plus inverse) have an active volume of $(14+ C_{|CCZ\rangle})(\kappa^2 + \kappa)/2$. Finally, the active volume of all adders is $\sum_{i=1}^{\kappa - 1}((i-1)(39+C_{|CCZ\rangle}) + 7) = 19.5\kappa^2 - 51.5\kappa + 32 + (0.5\kappa^2 - 1.5\kappa + 1)C_{|CCZ\rangle}$. Summing everything up yields the active volume of $28\kappa^2 - 42\kappa + 28 + (\kappa^2 - \kappa + 1)C_{|CCZ\rangle}$.

Concerning the reaction depth, assume for simplicity that $\kappa$ is a power of $2$. The controlled copying has reaction depth of $2$. On the other hand, the $\kappa - 1$ out-of-place adders are distributed in $\log_2\kappa$ layers, and at the $j$-th layer, $j=0,\dots,\log_2\kappa-1$, we sort the partial sums in increasing number of bits and add up the $i$-th partial sum with the $(\kappa/2^j - i + 1)$-th partial sum, $i=1,\dots, \kappa/2^{j+1}$. For example, at the $0$-th layer, the partial sum with $i$ bits is added to the partial sum with $\kappa-i+1$ bits, which requires an $i$-bit quantum adder, $i=1,\dots,\kappa/2$ (the $\kappa-2i+1$ least significant bits of the larger partial sum are simply attached to the result register to form the $\kappa$-bit answer). At the $j$-th layer, there are $\kappa/2^j$ partial sums, ranging from $1 + (1-2^{-j})\kappa$ bits to $\kappa$ bits. The longest addition at the $j$-th layer is the summation between the partial sums with $(1-2^{-j-1})\kappa$ and $1 + (1-2^{-j-1})\kappa$ bits, which requires an $(1-2^{-j-1})\kappa$-bit quantum adder with reaction depth of $2(1-2^{-j-1})\kappa - 2$. Therefore, the total reaction depth is
\begin{align*}
    2 + 2\sum_{j=0}^{\log_2\kappa - 1} ((1-2^{-j-1})\kappa - 1) = 2\kappa\log_2\kappa - 2\kappa - 2\log_2\kappa + 4.
\end{align*}

\paragraph*{Hybrid classical-quantum inputs.} In the case when one of the inputs is classical, say $y$, then the amount of resources decrease a bit. This is because there is no need to copy the registers $|x\rangle$ and $|y\rangle$ a number of $\kappa - 1$ times at the beginning, since $|y\rangle$ becomes classical and there is no need to parallelise $\mathsf{CNOT}$ gates. Moreover, the $\mathsf{Toffoli}$ gates used to controlled copy the register $|x_i,\dots,x_0\rangle$ using $y_{\kappa - 1 - i}$ become classically controlled $\mathsf{CNOT}$ gates. This means that the $\mathsf{Toffoli}$-count decreases by $(\kappa^2 + \kappa)/2$, while the $\mathsf{CNOT}$-count decreases to $\kappa^2 + \kappa$ (already taking the inverse into consideration), which have an active volume of $\sum_{i=0}^{\kappa-1}(\frac{3}{2}(\kappa - i) + 2) = 0.75\kappa^2 + 2.75\kappa$.

\section{Grover's quantum search algorithm} 
\label{sec:grover_search}

Unstructured search can be defined as follows. Given a function $f:\{0,1\}^n\to\{0,1\}$, find a marked element $x \in \{0,1\}^n$ such that $f(x) = 1$, or determine that with high probability, no such input exists. Classically this requires $\Theta(2^n)$ evaluations of $f$ to find such an input or determine it does not exist with high probability. By contrast, Grover~\cite{grover1996fast,grover97quantum} designed a quantum algorithm that finds a marked element with high probability and requires only $O(2^{n/2})$ calls to an binary oracle $\mathcal{U}_f$ that evaluates $f$ in superposition, $\mathcal{U}_f : |x\rangle|y\rangle \to |x\rangle|y\oplus f(x)\rangle$ for all $x\in\{0,1\}^n$ and $y\in\{0,1\}$. It was later shown~\cite{bennett97strengths,boyer1998tight,zalka99grover,beals01quantum} that Grover's algorithm is optimal for unstructured search.

Grover's algorithm is depicted in \cref{fig:grover}. By starting with the state $2^{-n/2}\sum_{x\in\{0,1\}^n} |x\rangle$ (which can be obtained from $|0\rangle^{\otimes n}$ by applying one layer $\mathsf{H}^{\otimes n}$ of $n$ Hadamard gates), the algorithm repeatedly applies the so-called \emph{Grover operator}
\begin{align*}
    \mathsf{G} = \mathsf{H}^{\otimes n} (2|0^n\rangle\langle 0^n| - \mathsf{I}_{2^n}) \mathsf{H}^{\otimes n} \mathcal{O}_f
\end{align*}
and then measures the state on the computational basis. The operator $\mathsf{D} := \mathsf{H}^{\otimes n} (2|0^n\rangle\langle 0^n| - \mathsf{I}_{2^n}) \mathsf{H}^{\otimes n}$ is called \emph{diffusion operator} and performs a conditional phase shift such that $|x\rangle \mapsto (-1)^{\mathbf{1}[x\neq 0^n]}|x\rangle$ for all $x\in\{0,1\}^n$. The oracle $\mathcal{O}_f$, on the other hand, is defined as $\mathcal{O}_f : |x\rangle \mapsto (-1)^{f(x)}|x\rangle$ for all $x\in\{0,1\}^n$. We note that it is possible to implement the phase oracle $\mathcal{O}_f$ from the binary operator $\mathcal{U}_f$ by simply applying $\mathcal{U}_f$ onto $|x\rangle|-\rangle$, where $|-\rangle := (|0\rangle - |1\rangle)/\sqrt{2}$.

Let $N = 2^n$ be the number of elements and $M$ the number of marked elements such that $f(x) = 1$. It can be shown~\cite{grover1996fast,grover97quantum,brassard2002quantum} that after $m$ iterations of $\mathsf{G}$, the probability of measuring a marked element is $\sin^2((2m+1)\theta)$, where $\sin^2\theta = \sqrt{M/N}$. Therefore, by using $m = \lfloor \frac{\pi}{4}\sqrt{N/M}\rfloor$ iterations, the measurement outcome will be a marked state with probability at least $1 - \frac{M}{N}$, which is sufficiently close to $1$ for $N \gg 1$. Each iteration requires one query to the oracle $\mathcal{O}_f$ (or $\mathcal{U}_f$) and one application of the diffusion operator. The diffusion operator, in turn, requires $2n$ Hadamard gates and one conditional phase $|x\rangle \mapsto (-1)^{\mathbf{1}[x\neq 0^n]}|x\rangle$, which is basically a slightly modified multi-controlled $\mathsf{Toffoli}$. More precisely, $2|0^n\rangle\langle 0^n| - \mathsf{I}_{2^n}$ equals $(\mathsf{X}^{\otimes n}\otimes \mathsf{I})(\mathsf{C}^{(n)}\text{-}\mathsf{X})(\mathsf{X}^{\otimes (n+1)})$ applied onto $|x\rangle|-\rangle$. The multi-controlled gate $\mathsf{C}^{(n)}\text{-}\mathsf{X}$ can be implemented using $n-1$ $\mathsf{Toffoli}$ gates and $n-2$ ancillae according to \cref{fact:mctoff} (among other resources). 

\begin{figure}[t]
    \centering
    \includegraphics[scale=0.7, trim={1.4cm, 19.2cm, 3.7cm, 0.8cm}, clip=true]{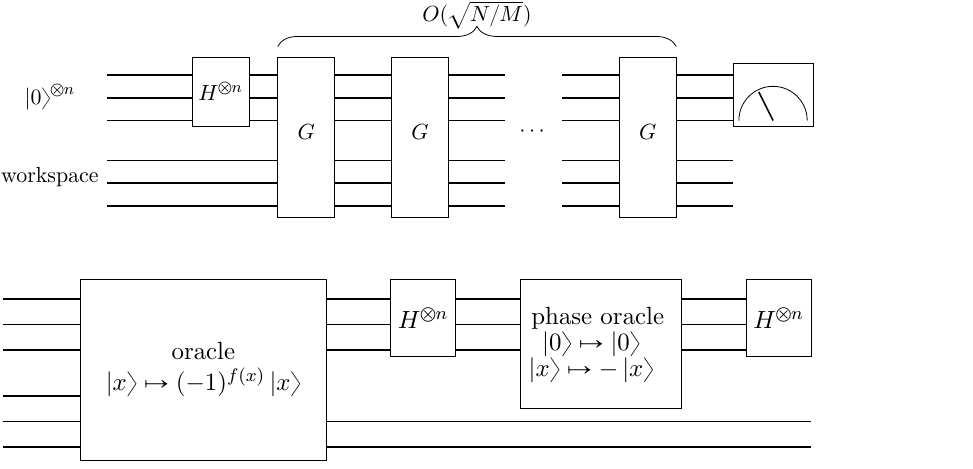}
    \caption{Circuit for Grover's search algorithm (top) and the Grover oracle $\mathsf{G}$ (bottom).}
    \label{fig:grover}
\end{figure}

We summarise the above discussion in the following result.
\begin{fact}[Grover's algorithm]\label{fact:grover_search}
    Let the positive integers $N = 2^n$ and $M$. Consider a Boolean function $f:\{0,1\}^n \to \{0,1\}$. Assume $M = |\{x\in\{0,1\}^n:f(x) = 1\}|$ is known and we have access to a quantum oracle $\mathcal{O}_f : |x\rangle \mapsto (-1)^{f(x)}|x\rangle$. Then it is possible to find one marked element of $f$ with probability at least $1 - \frac{M}{N}$ by using $\lfloor \frac{\pi}{4}\sqrt{N/M}\rfloor$ queries to $\mathcal{O}_f$ and the diffusion operator $\mathsf{D}$.
\end{fact}

\cref{fact:grover_search} assumes that the number of solutions is known beforehand, which is usually not the case. Nonetheless, there is a variant of Grover's algorithm due to Boyer, Brassard, H\o{}yer, and Tapp~\cite{boyer1998tight} (see also~\cite{zalka1999grover}) that applies to the case when the number of solutions is not known ahead of time. The main idea of their algorithm is to start with some parameter $m$, choose an integer $j$ uniformly at random such that $0\leq j < m$, and perform $j$ iterations of Grover's search. If it does not return a solution, the value $m$ is increased to $\lambda m$ for any constant $1 < \lambda < 4/3$ and the procedure is repeated. Boyer et al.~\cite{boyer1998tight} showed that this algorithm finds a solution (or determines that no solution exists) with high probability in expected time $O(\sqrt{N/M})$. A very thoroughly analysis of Grover's algorithm has been done by Cade, Folkertsma, Niesen, and Weggemans~\cite{Cade2023quantifyinggrover}, which we quote next and expand with the necessary resources for the diffusion operator.

\begin{fact}[{\cite[Lemma~4]{Cade2023quantifyinggrover}}]\label{fact:grover_search_unknown}
    Let $\delta\in(0,1)$ and the positive integer $N = 2^n$. Consider a Boolean function $f:\{0,1\}^n \to \{0,1\}$ with $|\{x\in\{0,1\}^n:f(x) = 1\}| = M$, where $0\leq M < N/4$ and the value $M$ is not necessarily known. Assume we have access to a quantum oracle $\mathcal{O}_f : |x\rangle \mapsto (-1)^{f(x)}|x\rangle$ and let $\mathsf{D} := \mathsf{H}^{\otimes n} (2|0^n\rangle\langle 0^n| - \mathsf{I}_{2^n}) \mathsf{H}^{\otimes n}$ be the diffusion operator. Then there is a quantum algorithm that, with probability at least $1 - \delta$,
    \begin{itemize}
        \item returns $x\in\{0,1\}^n$ such that $f(x) = 1$ if such a solution exists by using an expected number $\lceil 3.1\sqrt{N/M}\rceil$ of queries to $\mathcal{O}_f$ and $\mathsf{D}$,
        \item or concludes that no such solution exists by using $\lceil 9.2\sqrt{N}\log_{3}(1/\delta) \rceil$ queries to $\mathcal{O}_f$ and~$\mathsf{D}$.
    \end{itemize}
    Moreover, one call to the diffusion operator requires $n-1$ $\mathsf{Toffoli}$ gates and $n-1$ ancillae, and has reaction depth of $2\lceil\log_2{n}\rceil$, $\mathsf{Toffoli}$-width of $\lfloor n/2\rfloor$, and active volume of $(n-1)(18 + C_{|CCZ\rangle})$, where $C_{|CCZ\rangle}$ is the active volume of distilling one $|CCZ\rangle$ state.
\end{fact}

We mention that there exists an exact version of Grover's algorithm that succeeds with probability~$1$ (see~\cite[Section~2.1]{brassard2002quantum} and~\cite[Exercise~7.5]{de2019quantum}). However, even though this version is query efficient, it is not necessarily gate efficient, therefore we will not use it.

\section{Quantum random access memory (QRAM)}
\label{sec:qram}

A quantum random access memory ($\mathsf{QRAM}$) is the quantum analogue of the classical random access memory (RAM) device which allows access to classical or quantum data in superposition. From a architecture perspective, a $\mathsf{QRAM}$ of size $2^n$ and precision $\kappa$ is composed of a $\kappa 2^n$-(qu)bit memory register that stores either $\kappa$ bits or qubits in each of $2^n$ different cells, an $n$-qubit address register which points to the memory cell to be addressed, a $\kappa$-qubit target address into which the content of the addressed memory cell is copied, and an $O(2^n)$-qubit auxiliary register that intermediates the copying of the memory register into the target register controlled on the address register. For more details on the architecture of $\mathsf{QRAM}$s, we point the reader to~\cite{hann2021practicality,phalak2023quantum,jaques2023qram,allcock2023constant}.

In general, a $\mathsf{QRAM}$ allows access to either classical or quantum data stored in some register. Throughout this paper, we shall work exclusively with $\mathsf{QRAM}$s that access \emph{classical} data, which are sometimes referred to as quantum random access classical memory ($\mathsf{C\text{-}QRAM}$ or $\mathsf{QRACM}$). For simplicity, we stick to the usual $\mathsf{QRAM}$ nomenclature. Moreover, we shall consider $\mathsf{QRAM}$ calls that keep a garbage register (dirty ancillae) in order to aid their uncomputation at latter stages.
\begin{definition}[Quantum random access memory ($\mathsf{QRAM}$)]
    A $\mathsf{QRAM}$ of size $2^n$ and precision $\kappa$ is a device that stores classical, indexed data $\{x_i \in \{0,1\}^\kappa: i\in[2^n]\}$ and allows oracle queries
    \begin{align*}
        \mathsf{QRAM}:\ket{i}\ket{0}^{\otimes \kappa}|\bar{0}\rangle \mapsto \ket{i}\ket{x_i}|{\rm garbage}_i\rangle, \quad \forall i\in[2^n].
    \end{align*}
\end{definition}

The first architectures for $\mathsf{QRAM}$ were proposed and formalised in~\cite{giovannetti2008quantum,giovannetti2008architectures}, namely the Fan-Out and bucket-brigade architectures. In these architectures, the memory register is accessed by a binary tree of size $O(2^n)$ and depth $n$. Each qubit of the address register controls the direction from the root down to the correct memory cell within the binary tree, i.e., the $k$-th address qubit specifies whether to go left or right at a router on the $k$-th level of the binary tree. The target is sent down the tree and is routed controlled on the address qubits at each level until the memory register, at which point the information is copied into the target and the target is sent back up the tree. The difference between the Fan-Out and bucket-brigade architectures is in how the target qubits are routed down the binary tree. We point out the reader to~\cite{giovannetti2008quantum,giovannetti2008architectures,arunachalam2015robustness,hann2021resilience} for more information.

\begin{figure}[t]
    \centering
    \includegraphics[scale=0.65, trim={1.4cm, 7.2cm, 3.7cm, 0.8cm}, clip=true]{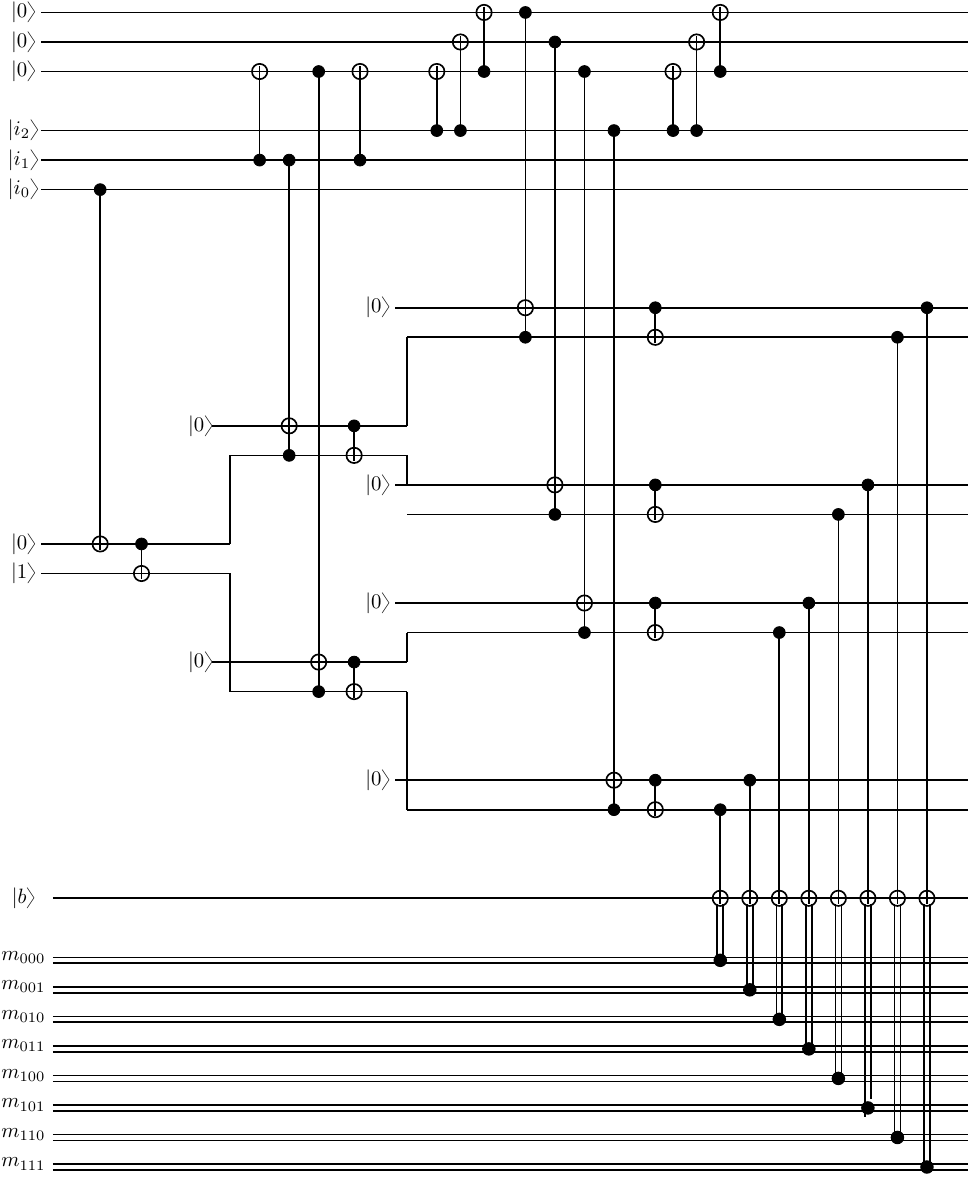}
    \caption{The bucket-brigade $\mathsf{QRAM}$ circuit from Arunachalam et al.~\cite{arunachalam2015robustness}. In every layer, before the parallel layer of $\mathsf{Toffoli}$ gates, a log-depth linear-size gadget copy the index register so the $\mathsf{Toffoli}$ gates can be executed in parallel.}
    \label{fig:bucket-brigade}
\end{figure}

Here we shall be agnostic regarding the underlying architecture of a $\mathsf{QRAM}$ and shall work with the circuit model instead. We assume nonetheless that the contents of the memory are stored statically, meaning that the classical data is stored in an external physical hardware, e.g., a tape, which is quantumly queried. This is accomplished by applying classically controlled $\mathsf{CNOT}$ gates onto the target qubit with one classical control (a bit from the memory) and one quantum control (a qubit from the last layer of the binary tree). We show a circuit implementation of $\mathsf{QRAM}$ in \Cref{fig:bucket-brigade}. Moreover, we also assume that the classical memory can be updated independently from the $\mathsf{QRAM}$ device itself. In other words, $m$ different cells from the classical memory can be rewritten in time $O(\kappa m)$ without the need to update the remaining registers from the $\mathsf{QRAM}$. This differs from Quantum Read-Only Memory ($\mathsf{QROM}$) or table lookups~\cite{babbush2018encoding} which usually encode the memory content into the circuit~layout.

The fault-tolerance resources required to implement a $\mathsf{QRAM}$ have been studied by a few works~\cite{di2020fault,Low2024tradingtgatesdirty,litinski2022active,litinski2023compute}. Di Matteo, Gheorghiu, and Mosca~\cite{di2020fault} studied the amount of $\mathsf{T}$ gates in bucket-brigade style $\mathsf{QRAM}$s, while Low, Kliuchnikov, and Schaeffer~\cite{Low2024tradingtgatesdirty} proposed a $\mathsf{T}$-efficient sequential $\mathsf{QROM}$ circuit. Litinski and Nickerson~\cite{litinski2022active} worked out the active volume of Low et al.\ proposal. Here we employ a bucket-brigade $\mathsf{QRAM}$ due to its exponentially smaller reaction depth compared to Low, Kliuchnikov, and Schaeffer's $\mathsf{QROM}$.
\begin{lemma}[Bucket-brigade $\mathsf{QRAM}$]\label{lem:qram_resources}
    One bucket-brigade $\mathsf{QRAM}$ call of size $2^n$ and precision $\kappa$ requires (already including its uncomputation) $2^n - 2$ $\mathsf{Toffoli}$ gates, $2^{n+1} - n - 1$ dirty ancillae (plus $n+\kappa$ input/output qubits), and has $\mathsf{Toffoli}$-width of $2^{n-1}$, reaction depth of $2(n-1)$, and active volume of $(25 + 1.5\kappa + C_{|CCZ\rangle})2^n$.
\end{lemma}
\begin{proof}
    All the resources apart from the active volume are straightforward. In the following, we already take the uncomputation into consideration. A bucket-brigade $\mathsf{QRAM}$ can be divided into $2^n - 2$ blocks made up of one $\mathsf{Toffoli}$ and one $\mathsf{CNOT}$ gate and having active volume $14 + 2\cdot 4 + C_{|CCZ\rangle}$; $\kappa\cdot 2^n$ classically controlled $\mathsf{CNOT}$s with average active volume of $(\frac{3}{2}2^n + 1)\kappa$ (since on average half the $\mathsf{CNOT}$s is actually performed); and $2^n - n - 1$ $\mathsf{CNOT}$s to copy the address register with active volume of $2\sum_{i=0}^{n-1}(\frac{3}{2}(2^i - 1) + 2) = 3 \cdot 2^n + n - 3$. Summing all active volumes yields the result after some simple approximations.
\end{proof}


\section{The shortest vector problem and sieving algorithms}
\label{sec:sieving_algorithms}

The most important problem on lattices and that underlies many lattice-based cryptography functions~\cite{ajtai1997public,regev2004new,regev2006lattice,Micciancio2009lattice} is the \emph{shortest vector problem} (SVP). Given a set $\mathbf{B} = \{ \mathbf{b}_1, \ldots, \mathbf{b}_N \} \subset \R^D$ of $N$ linearly independent vectors, the set 
\begin{align*}
    \L(\mathbf{B}) := \left\{\sum_{j=1}^{N} \lambda_j \mathbf{b}_j : \lambda_1,\dots,\lambda_N \in \Z\right\}
\end{align*}
of all integer linear combinations of $\mathbf{B}$ is called the lattice associated with $\mathbf{B}$. The set $\mathbf{B}$ is called the basis of the lattice, while the integers $N$ and $D$ are its rank and dimension, respectively. In this work, we consider \emph{full rank} lattices, which is the case when $N = D$. The minimum distance $\lambda(\mathcal{L})$ of a lattice $\mathcal{L}$ is the length of its shortest non-zero lattice vector, $\lambda(\mathcal{L}) := \min\{\|\mathbf{x}\|:\mathbf{x}\in\mathcal{L}\setminus\{\mathbf{0}\}\}$. We shall abuse notation and write $\lambda(\mathbf{B})$ instead of $\lambda(\mathcal{L}(\mathbf{B}))$.
\begin{definition}[Shortest vector problem]
    Given a lattice basis $\mathbf{B}\in\mathbb{R}^{D\times D}$, find $\mathbf{x}\in\mathcal{L}(\mathbf{B})$ such that $\|\mathbf{x}\| = \lambda(\mathbf{B})$.
\end{definition}

SVP is known to be NP-hard under randomised reductions~\cite{Ajtai1998shortest,Micciancio1998shortest,Micciancio2001shortest} given an arbitrary basis of an arbitrary lattice. Even the approximate version of SVP, wherein one is tasked to find a lattice vector with norm at most $(1+\epsilon)\lambda(\mathbf{B})$ for $\epsilon>0$, is known to be NP-hard~\cite{Khot2004hardness,Khot2005hardness}. Nonetheless, several exponential-time algorithms have been proposed in the past few decades to tackle SVP. There are currently three main methodologies: enumeration~\cite{fincke1985improved,Kannan1983improved,Pohst1981computation}, sieving~\cite{ajtai2001sieve,Ajtai2001overview,micciancio2010faster,aggarwal2015solving}, and constructing the Voronoi cell of the lattice~\cite{Agrell2002closest,micciancio2010deterministic}. Whereas enumeration has a polynomial space complexity but a superexponential time complexity $O(2^{D\log{D}})$ on the dimension $D$ of the lattice~\cite{Pohst1981computation,Schnorr1991lattice,Schnorr1994Lattice,Kannan1983improved}, the remaining methods all have both exponential space and time complexities.

In this section, we focus on and review the major sieving algorithms since their introduction by Ajtai, Kumar, and Sivakumar~\cite{ajtai2001sieve,Ajtai2001overview}. Sieving algorithms work by sampling a long list $L = \{\mathbf{v}_1,\dots,\mathbf{v}_m\}$ of lattice vectors (either initially or during the algorithm) and considering all pair-wise differences $\mathbf{v}_i \pm \mathbf{v}_j\in\mathcal{L}(\mathbf{B})$ from the list. Most of these combinations result into longer vectors than the initial vectors $\mathbf{v}_i$ and $\mathbf{v}_j$, but some lead to shorter vectors. By keeping the resulting shorter vectors into a new list, progress is made into finding the shortest vector. The step of combining lattice vectors from a list in order to form a new list with shorter lattice vectors is called \emph{sieving}. We hope that, if a substantially large number of lattice vectors is sampled, then several sieving steps will result into a small list that contains the shortest vector.

Whereas Ajtai, Kumar, and Sivakumar~\cite{ajtai2001sieve,Ajtai2001overview} originally proved that sieving can solve SVP in time and space $2^{\Theta(D)}$,  later works improved their results and showed that sieving can provably solve SVP in time $2^{2.465D + o(D)}$ and space $2^{1.233D + o(D)}$~\cite{nguyen2008sieve,Pujol2009solving,Hanrot2011algorithms}. At first glance, these provable bounds suggest that sieving algorithms would perform poorly in practice, and that solving SVP on dimension beyond $50$ would be  impractical. Experimental works suggest otherwise and that sieving algorithms perform well in practice. This has led to new sieving proposals that can tackle SVP under \emph{heuristic} assumptions. The first proposal for a heuristic sieving algorithm was given by Nguyen and Vidick~\cite{nguyen2008sieve}, which we now review. In what follows, we assume that all the vectors have coordinates described using $\kappa$-bits. 

\subsection{The Nguyen-Vidick sieve}

Nguyen and Vidick~\cite{nguyen2008sieve} proposed the first sieving algorithm that relies on heuristic assumptions. A version of the Nguyen-Vidick sieve ($\mathtt{NVSieve}$) that already incorporates Grover's algorithm is depicted in \cref{alg:nv_sieve} (cf.~\cite[Algorithm~2]{Laarhoven2016search}). The first step is to sample a list $L$ of lattice vectors using, e.g., Klein's algorithm~\cite{klein00, GPV08trapdoors}, which samples lattice vectors from a distribution that is statistically close to a discrete Gaussian on a lattice with a reasonably small variance. 
A sieving process is then applied onto $L$ to reduce pairs by considering the differences $\mathbf{v} - \mathbf{w}$ of pairs of lattice vectors $\mathbf{v},\mathbf{w}\in L$. If $\mathbf{v} - \mathbf{w}$ yields a shorter vector than $\mathbf{v},\mathbf{w}$, it is stored in a new list $L'$. Instead of considering all pair-wise combinations of vectors from the list $L$, the $\mathtt{NVSieve}$ keeps a list of centers $S\subset L$, each covering a part of the space. Each vector $\mathbf{v}\in L$ from the list is thus combined with vectors $\mathbf{w}\in S$ from the list of centers. If the result is a shorter vector, $\mathbf{v} - \mathbf{w}$ is added to new list $L'$, otherwise the initial list vector $\mathbf{v}\in L$ is added to the list of centers $S$ to cover a part of the space which was previously left uncovered. At the end of the sieving step, $L\gets L'$. After many sieving steps as necessary, the list $L$ contains the shortest lattice vector or is left empty, in which case the whole algorithm is repeated. 

Under the heuristic assumption that the angle between two list vectors $\mathbf{v},\mathbf{w}\in L$ follows the same distribution as the angle between two uniformly random vectors over the unit sphere, Nguyen and Vidick~\cite{nguyen2008sieve} proved that an initial list of size $(4/3)^{D/2 + o(D)} = 2^{0.208D + o(D)}$ suffices to find the shortest vector. This bounds the space complexity of the $\mathtt{NVSieve}$. On the other hand, the time complexity is dominated by comparing every list vector $\mathbf{v}\in L$ to a center vector $\mathbf{w}\in S$, and since the number of center vectors is asymptotically equivalent to the number of list vectors, this means that the $\mathtt{NVSieve}$ solves SVP in time $2^{0.415D + o(D)}$. 

\begin{algorithm}[t]
    \caption{The Nguyen-Vidick sieve}\label{alg:nv_sieve}
    \SetKwFor{ForEach}{for each}{do}{endfch}
    \SetKwComment{Comment}{/* }{ */}
    
    \KwIn{Basis $\mathbf{B}$ for a $D$-dimensional lattice and parameter $\gamma \in (0,1)$}
    \KwOut{Shortest vector $\mathbf{v}^\ast$ of the lattice.}

    Sample $L\subset \mathbb{R}^D$

    \While{$L \neq \emptyset$}
    {
        $L_0 \gets L$, $L' \gets \emptyset$, $S \gets \{\mathbf{0}\}$, $R\gets \max_{\mathbf{v}\in L}\|\mathbf{v}\|$

        \ForEach{$\mathbf{v}\in L$}
        {
            $\mathbf{w} \gets \mathtt{GroverSearch}(\mathbf{w}\in S : \|\mathbf{v} - \mathbf{w}\| \leq \gamma R)$ \label{line:Grover_nv_sieve}
    
            \If(\tcp*[f]{$\exists\mathbf{w}\in S  : \|\mathbf{v} - \mathbf{w}\| \leq \gamma R$}\label{line:nvsieve_line6}){$\mathbf{w} \neq \mathtt{NULL}$} 
            {
                $L' \gets L'\cup\{\mathbf{v} - \mathbf{w}\}$ \label{line:nvsieve_line7}
            }
            \Else(\tcp*[f]{$\nexists\mathbf{w}\in S  : \|\mathbf{v} - \mathbf{w}\| \leq \gamma R$}\label{line:nvsieve_line8}) 
            {
                $S \gets S\cup \{\mathbf{v}\}$ \label{line:nvsieve_line9}
            }
        }
        $L \gets L'$
    }    
    \Return shortest vector $\textbf{v}^\ast$ in $L_0$
\end{algorithm}

\subsubsection{Numerical experiments and heuristic assumptions}
\label{sec:heuristics_nvsieve}

The asymptotic complexity hides a lot of details, especially in case of a quantum algorithm. We want a more refined analysis of the runtime of this algorithm. Since the analysis of the $\mathtt{NVSieve}$ relies on heuristic assumptions, i.e., that vectors in $L\cap B_D(\gamma, R)$ are uniformly distributed in $B_D(\gamma, R) := \{\mathbf{v}\in\mathbb{R}^D: \gamma R\leq \|\mathbf{v}\| \leq R\}$ (where $R = \max_{\mathbf{v}\in L}\|\mathbf{v}\|$), quantities like the number of sieving steps or the evolution of the list size $|L|$ can behave as random variables and are thus not determined beforehand. Nonetheless, it is possible to assert average trends and worst-case bounds through plausible assumptions and numerical experiments. In the following, we list several observations that shall be useful in forming assumptions.
\begin{enumerate}
    \item Nguyen and Vidick~\cite{nguyen2008sieve} proved that, given $\gamma\in (2/3,1)$, the maximum size of the list of centers $S$ is upper-bounded by $N_S = \lceil 3\sqrt{2\pi}(D+1)^{3/2}(\gamma\sqrt{1-\gamma^2/4})^{-D}\rceil$. By letting $\gamma \to 1$, then $N_S \to \lceil 3\sqrt{2\pi}(D+1)^{3/2}(4/3)^{D/2}\rceil$. Experimentally, Nguyen and Vidick~\cite{nguyen2008sieve} observed that the size of $S$ is upper-bounded by $\ln|S| \leq aD + b\ln{D} + c$ where $a=0.163(\pm 0.017)$, $b=0.102(\pm 0.65)$, and $c=1.73(\pm 1.72)$ if $\gamma = 0.97$. In other words, $|S| \leq 2^{0.2352D + 0.102\log_2{D} + 2.45}$.
    
    \item In practice, one samples an initial list $L$ of considerable size and runs the $\mathtt{NVSieve}$. If the shortest vector is not found and the $\mathtt{NVSieve}$ thus fails, the whole procedure is restarted but with a larger initial list. Given numerical experiments from~\cite{nguyen2008sieve} and also conducted by us, an initial list $L$ of size $D$ times that of $|S| \leq 2^{0.2352D + 0.102\log_2{D} + 2.45}$ suffices. Alternatively, $|L| = \lceil 3\sqrt{2\pi}(D+1)^{3/2}(4/3)^{D/2}\rceil$ also works.
    
    \item As pointed out by Nguyen and Vidick~\cite{nguyen2008sieve}, the list size $|L|$ decreases roughly by $(\gamma\sqrt{1-\gamma^2/4})^{-D}$ at each sieving step, provided the vectors in $L$ are well distributed in $B_D(\gamma, R)$. Indeed, in \Cref{line:nvsieve_line6,line:nvsieve_line7,line:nvsieve_line8,line:nvsieve_line9} from \Cref{alg:nv_sieve}, each $\mathbf{v}\in L$ is either selected to $L'$ (reduced by $\mathbf{w}$) or to $S$, and thus $|L'| \approx |L| - |S|$ if there are few collisions ($\mathbf{v} - \mathbf{w} = \mathbf{0}$). Numerical experiments from~\cite[Figure~2]{nguyen2008sieve} show that the number of collisions is negligible until $R/\lambda(\mathbf{B}) \approx 4/3$. Therefore, for most sieving steps, the size of $L$ is reduced by the size of $S$, which is at most $2^{0.2352D + 0.102\log_2{D} + 2.45}$.
    
    \item As $\gamma \to 1$, the expected size of $S$ decreases, while the number of sieving steps clearly increases. Nguyen and Vidick~\cite{nguyen2008sieve} used a contraction parameter $\gamma = 0.97$ in their simulations. By keeping $\gamma \approx 0.97$, we expect the upper-bound $|S| \leq 2^{0.2352D + 0.102\log_2{D} + 2.45}$ to hold. Moreover, if $\gamma$ is not too close to $1$, we can abstract away the number of sieving steps by assuming that $|L|$ roughly decreases by $|S|$.
    
    \item While the size of $S$ fluctuates within a sieving step in \Cref{alg:nv_sieve}, there are other implementations of the $\mathtt{NVSieve}$~\cite{laarhoven15angular,Laarhoven2016search} in which the list of centers $S$ is sampled from $L$ beforehand in every sieving step and, therefore, $|S|$ is kept constant.
\end{enumerate}
 
\subsubsection{Quantum oracle for Grover search} 

The $\mathtt{NVSieve}$ employs one search subroutine per sieve step, per list vector, which can be done using Grover's algorithm (\cref{line:Grover_nv_sieve} in \cref{alg:nv_sieve}). For fixed $\mathbf{v}\in L$, the search is done over the centers $S$ in order to find an element $\mathbf{w}$ such that $\|\mathbf{v} - \mathbf{w}\| \leq \gamma R$, where $R = \max_{\mathbf{v}\in L}\|\mathbf{v}\|$. 
Define the Boolean function $f_{\rm NV}:[|S|]\to\{0,1\}$ such that $f_{\rm NV}(i) = 1$ if and only if $\|\mathbf{v} - \mathbf{w}_i\| \leq \gamma R$. In order to use Grover search, we must implement the phase oracle $\mathcal{O}_{\rm NV}:|i\rangle \mapsto (-1)^{f_{\rm NV}(i)}|i\rangle$, as explained next. 

\begin{table}[t]
    \small
    \centering
    \caption{Amount of subroutines required to implement a phase oracle in each Grover search per sieve step in the following sieving algorithms: $\mathtt{NVSieve}$, $\mathtt{NVSieve}$ with LSH/LSF, $\mathtt{GaussSieve}$, and $\mathtt{GaussSieve}$ with LSH/LSF. The $\mathtt{NVSieve}$ requires only one type of Grover search, while the $\mathtt{GaussSieve}$ requires two types. All quantum adders, comparators, and multipliers are $\kappa$-bit out-of-place operations. Operations marked by ($\ast$) have hybrid classical-quantum inputs and are thus cheaper. All subroutines include their inverse.}
    \label{table:oracle_modulos}
    \begin{tabular}{ | c | c | c | c | c | c | c |}
        \hline
        Sieve/Operations & $\mathsf{QRAM}$ & Adders & Multipliers & Extra $\mathsf{CNOT}$s \\ \hline
        
        $\mathtt{NVSieve}$ & $1$ & $2D$ & $D$ & $0$ \\ \hline

        $\mathtt{NVSieve}$ + LSH/LSF & $1$ & $2D$ & $D$ & $0$ \\ \hline

        \multirow{2}{*}{$\mathtt{GaussSieve}$} & $1$ & $4D-2$ & $2D$ &  $2D\kappa + 4$\\  
        & $1$ & $D+1$ & $D^\ast$ & $4$ \\ \hline 

        \multirow{2}{*}{$\mathtt{GaussSieve}$ + LSH/LSF} & $1$ & $4D-2$ &  $2D$ & $2D\kappa + 4$\\  
        & $1$ & $D+1$ & $D^\ast$ & $4$ \\ \hline 
    \end{tabular}
\end{table}

Given any index $|i\rangle$ where $i\in[|S|]$, we start with one $\mathsf{QRAM}_{S}$ call to load $\mathbf{w}_i$ onto a $(\kappa D)$-qubit ancillary register. The list of centers $S$ is already loaded onto the $\mathsf{QRAM}$ at the beginning of every sieve step and the resources required for one call are given in \cref{lem:qram_resources}. Next we must compute the value of $f_{\rm NV}$. Rewrite the inequality defining the Boolean function $f_{\rm NV}$ as
\begin{align*}
    \sum_{j=1}^D (\mathbf{w}_i)_j (\mathbf{w}_i - 2\mathbf{v})_j =  \mathbf{w}_i \cdot (\mathbf{w}_i - 2\mathbf{v}) \leq \gamma^2R^2 - \|\mathbf{v}\|^2.
\end{align*}
In order to compute $\sum_{j=1}^D (\mathbf{w}_i)_j (\mathbf{w}_i - 2\mathbf{v})_j$, we first compute $\mathbf{w}_i - 2\mathbf{v}$ using $D$ parallel $\kappa$-bit out-of-place adders with the classical input $2\mathbf{v}$. At this point the quantum registers hold $|i\rangle|\mathbf{w}_i\rangle|\mathbf{w}_i - 2\mathbf{v}\rangle$. Next, all the terms $(\mathbf{w}_i)_j (\mathbf{w}_i - 2\mathbf{v})_j$, $j\in[D]$, are computed using $D$ parallel $\kappa$-bit out-of-place multipliers. This yields the quantum registers $|i\rangle|\mathbf{w}_i\rangle|\mathbf{w}_i - 2\mathbf{v}\rangle \bigotimes_{j=1}^D |(\mathbf{w}_i)_j (\mathbf{w}_i - 2\mathbf{v})_j\rangle$. For the next step, all $D$ terms $(\mathbf{w}_i)_j (\mathbf{w}_i - 2\mathbf{v})_j$ are summed in depth $\lceil \log_2{D}\rceil$ by using $D-1$ $\kappa$-bit out-of-place adders. Finally, we employ a $\kappa$-bit comparator (which counts as an $\kappa$-bit adder) between the quantum register holding $\sum_{j=1}^D (\mathbf{w}_i)_j (\mathbf{w}_i - 2\mathbf{v})_j$ and the classical input $\lambda := \gamma^2R^2 - \|\mathbf{v}\|^2$, but the output register of the comparator is initialised in the $|-\rangle$ state instead of the $|0\rangle$ state. This procedure introduces the phase $(-1)^{f_{\rm NV}(i)}$ as wanted. We summarise the whole chain of operations as follows:
\begin{align*}
    &|-\rangle|i\rangle|0\rangle^{\otimes \kappa(2+ 3D)} \\ \xrightarrow{\mathsf{QRAM}_{S}} &|-\rangle|i\rangle|\mathbf{w}_i\rangle|0\rangle^{\otimes \kappa(2+ 2D)} \\
    \xrightarrow{D~\text{adders}} &|-\rangle|i\rangle|\mathbf{w}_i\rangle|\mathbf{w}_i - 2\mathbf{v}\rangle|0\rangle^{\otimes \kappa(2+ D)} \\
    \xrightarrow{D~\text{multipliers}} &|-\rangle|i\rangle|\mathbf{w}_i\rangle|\mathbf{w}_i - 2\mathbf{v}\rangle\left(\bigotimes_{j=1}^D|(\mathbf{w}_i)_j (\mathbf{w}_i - 2\mathbf{v})_j\rangle\right)|0\rangle^{\otimes 2\kappa}\\
    \xrightarrow{D-1~\text{adders}} &|-\rangle|i\rangle|\mathbf{w}_i\rangle|\mathbf{w}_i - 2\mathbf{v}\rangle\left(\bigotimes_{j=1}^D|(\mathbf{w}_i)_j (\mathbf{w}_i - 2\mathbf{v})_j\rangle\right)| \mathbf{w}_i^\top(\mathbf{w}_i-2\mathbf{v})\rangle|0\rangle^{\otimes \kappa}\\
    \xrightarrow{1~\text{adder}} &(-1)^{f_{\rm NV}(i)}|-\rangle|i\rangle|\mathbf{w}_i\rangle|\mathbf{w}_i - 2\mathbf{v}\rangle\left(\bigotimes_{j=1}^D|(\mathbf{w}_i)_j (\mathbf{w}_i - 2\mathbf{v})_j\rangle\right) |\mathbf{w}_i^\top(\mathbf{w}_i-2\mathbf{v})\rangle|\mathbf{w}_i^\top(\mathbf{w}_i-2\mathbf{v})-\lambda\rangle.
\end{align*}
At the end of this chain of operations, after the phase $(-1)^{f_{\rm NV}(i)}$ has been applied, we uncompute all the $2D$ adders, $D$ multipliers, and one $\mathsf{QRAM}_{S}$ call using their suitable inverses. Even though all the extra ancillae required for adders, multipliers, and $\mathsf{QRAM}$ are not made explicit in the arguments above, dirty ancillae are kept throughout the computations in order to ease their inverses. 
In \cref{table:oracle_modulos}, we summarise all the arithmetic and memory modules required in the phase oracle $\mathcal{O}_{\rm NV}$.

\subsubsection{Using LSH/LSF in the Nguyen-Vidick sieve} 

Locality-sensitive hashing and filtering can be used to speed up sieving, particularly the $\mathtt{NVSieve}$~\cite{nguyen2008sieve}. The main idea of employing nearest-neighbour-search techniques in the $\mathtt{NVSieve}$ is to preprocess the list of centers $S$ and thus replace the brute-force list search over $\mathbf{w}\in S$ with a smaller list of probable reducible candidates. As described in \cref{sec:lsh}, we sample $k\cdot t$ random hash function $h_{i,j}\in\mathcal{H}$ from a suitable hash family $\mathcal{H}$, and using the AND and OR-compositions, introduce $t$ hash tables, each with an exponential number of buckets in the parameter $k$ ($2^k$ buckets in angular LSH and $2^{k\sqrt{D}}$ buckets in spherical LSH). Each vector $\mathbf{w}\in S$ is then added to its corresponding bucket $\mathcal{T}_i[h_i(\mathbf{v})]$ labelled by the hash $h_i(\mathbf{v})$ for each hash table $\mathcal{T}_i$. Afterwards, given $\mathbf{v}\in L$, only vectors in buckets $\mathcal{T}_1[h_1(\mathbf{v})],\dots, \mathcal{T}_t[h_t(\mathbf{v})]$ are considered as possible candidates for reduction. The search space is thus considerable reduced and many of far-away vectors in $S$ to $\mathbf{v}$ are ignored. The $\mathtt{NVSieve}$ with LSH is described in \cref{alg:nv_sieve_LSH}. A similar procedure applies to LSF: $k\cdot t$ random filter functions $f_{i,j}\in \mathcal{F}$ from a suitable filter family $\mathcal{F}$ are sampled and employed to add vectors  onto $t$ different filtered buckets $\mathcal{B}_1,\dots,\mathcal{B}_t$. A vector $\mathbf{w}\in S$ is added onto the bucket $\mathcal{B}_i$ if and only if it passes through all filters $f_{i,1},\dots,f_{i,k}$. Afterwards, a query $\mathbf{v}\in L$ is answered by recovering all the filters that $\mathbf{v}$ passes through. We note that insertions into buckets and searching over relevant filters might use filters with different internal parameters (an example being the parameter $\alpha$ in spherical LSF). The different between LSH and LSF is that, in the second case, each hash table is reduced to a single bucket of vectors that survived the filters. Another difference is the use of random product codes to efficiently find all buckets that contain a given vector.

\begin{algorithm}[t]
    \caption{The Nguyen-Vidick sieve with LSH}\label{alg:nv_sieve_LSH}
    \SetKwFor{ForEach}{for each}{do}{endfch}
    \SetKwComment{Comment}{/* }{ */}
    
    \KwIn{Basis $\mathbf{B}$ for a $D$-dimensional lattice, parameters $\gamma\in(0,1)$, $k$, and $t$, hash family $\mathcal{H}$.}
    \KwOut{Shortest vector $\mathbf{v}^\ast$ of the lattice.}

    Sample $L\subset \mathbb{R}^D$

    \While{$L \neq \emptyset$}
    {
        $L_0\gets L$, $L' \gets \emptyset$, $S\gets \{\mathbf{0}\}$, $R\gets \max_{\mathbf{v}\in L}\|\mathbf{v}\|$

        Initialise $t$ empty hash tables $\mathcal{T}_1,\dots,\mathcal{T}_t$ and sample $k\cdot t$ random hash functions $h_{i,j}\in\mathcal{H}$

        For each $i\in[t]$, add $\mathbf{0}$ to the bucket $\mathcal{T}_i[h_i(\mathbf{0})]$ in the hash table $\mathcal{T}_i$

        \ForEach{$\mathbf{v}\in L$}
        {
            $C\gets \bigcup_{i=1}^t \mathcal{T}_i[h_i(\mathbf{v})]$ is the list of candidate vectors \label{line:candidates_nvsieve}

            Construct a $\mathsf{QRAM}$ for $C$
        
            $\mathbf{w} \gets \mathtt{GroverSearch}(\mathbf{w}\in C : \|\mathbf{v} - \mathbf{w}\| \leq \gamma R)$ \label{line:Grover_nv_sieve_LSH}
    
            \If(\tcp*[f]{$\exists\mathbf{w}\in C  : \|\mathbf{v} - \mathbf{w}\| \leq \gamma R$}){$\mathbf{w} \neq \mathtt{NULL}$}
            {
                $L' \gets L'\cup\{\mathbf{v} - \mathbf{w}\}$
            }
            \Else(\tcp*[f]{$\nexists\mathbf{w}\in C  : \|\mathbf{v} - \mathbf{w}\| \leq \gamma R$})
            {
                $S\gets S\cup\{\mathbf{v}\}$
                
                For each $i\in[t]$, add $\mathbf{v}$ to the bucket $\mathcal{T}_i[h_i(\mathbf{v})]$ in the hash table $\mathcal{T}_i$ \label{line:add_vector_hash_tables}
            }
        }
        $L \gets L'$
    }    
    \Return shortest vector $\textbf{v}^\ast$ in $L_0$
\end{algorithm}

Apart from searching over the buckets $\mathcal{T}_1[h_1(\mathbf{v})],\dots, \mathcal{T}_t[h_t(\mathbf{v})]$ for a reducible vector using Grover's algorithm, another big difference between the classical and quantum versions of the $\mathtt{NVSieve}$ with LSH is obtaining the list of candidates vectors $C = \bigcup_{i=1}^t \mathcal{T}_i[h_i(\mathbf{v})]$ in the first place. Whereas classically the search over $C$ can be done while sequentially visiting all the $t$ buckets, to retain the quadratic quantum advantage, we must first classically gather all the indices (or hashes) of the vectors in the buckets $\mathcal{T}_1[h_1(\mathbf{v})],\dots,\mathcal{T}_t[h_t(\mathbf{v})]$, after which we can perform Grover's search over these vectors. If $C = \{\mathbf{w}_{j_1}, \mathbf{w}_{j_2},\dots, \mathbf{w}_{j_{|C|}}\}$, then we start with the state $|C|^{-1/2}\sum_{i=1}^{|C|}|i\rangle$ within Grover's search. To proceed, we would like to use $\mathsf{QRAM}$ to map $|i\rangle|0\rangle^{\otimes \kappa D} \mapsto |i\rangle|\mathbf{w}_{j_i}\rangle$. This can be done using the $\mathsf{QRAM}$ of \cref{fig:bucket-brigade} by classically ordering the hashes $\{j_1,j_2,\dots,j_{|C|}\}$ so that a classically controlled $\mathsf{CNOT}$ is applied based on the content $\mathbf{w}_{j_i}$, which is accessed via a $\mathsf{RAM}$ call. The phase oracle $\mathcal{O}_{\rm NV}:|i\rangle \mapsto (-1)^{f_{\rm NV}(i)}|i\rangle$, where $f_{\rm NV}:[|C|]\to\{0,1\}$ is defined by $f_{\rm NV}(i) = 1$ if and only if $\|\mathbf{v} - \mathbf{w}_{j_i}\| \leq \gamma R$, is hence implemented initially as
\begin{align}\label{eq:nv_sieve_qram}
    |i\rangle|0\rangle^{\otimes \kappa D} \xrightarrow{\mathsf{QRAM}_{C}} |i\rangle|\mathbf{w}_{j_i}\rangle,
\end{align}
after which the remaining addition and multiplication operations explained in the previous section are performed (plus overall uncomputation). The phase oracle $\mathcal{O}_{\rm NV}$ thus requires one $\mathsf{QRAM}$ call one of size $|C|$ and $|C|$ $\mathsf{RAM}$ calls of size $|L|$. All the required subroutines to implement the phase oracle are summarised in \cref{table:oracle_modulos}.

The need to take $\mathsf{QRAM}$ into consideration means that the use of Grover's search in the $\mathtt{NVSieve}$ does not improve its asymptotic scaling, since gathering the list of candidates $C$ for each $\mathbf{v}\in L$ takes $O(|L|\cdot |C|)$ time. The only improvement is moving the more expensive norm computation into the Grover's search, so that the classical cost of $O(D\cdot |L|\cdot |C|)$ per sieving step becomes a classical-quantum cost of $O(|L| \cdot |C| + D\cdot|L|\cdot\sqrt{|C|})$.

\subsubsection{NVSieve with angular LSH} 

When employing angular LSH, we hash each vector into a $k$-bit string for each of the $t$ hash tables. Each hash table has thus $2^k$ buckets. We choose $k = \log_{3/2}{t} - \log_{3/2}\ln(1/\varepsilon)$ so that nearby vectors collide with high probability. Hashing one vector requires computing $\mathbf{a}_i\cdot \mathbf{v}$ for all $i\in[k]$, so $kD$ multiplications and $k(D-1)$ additions. As pointed out by Laarhoven~\cite{laarhoven15angular}, it is possible to employ sparse vectors $\mathbf{a}_i$ for the angular hash functions while still maintaining the same performance~\cite{achlioptas2001database,achlioptas2003database,li2006very}. Laarhoven~\cite{laarhoven15angular} employed vectors $\mathbf{a}_i$ with just two non-zero entries, therefore we require $2k$ multiplications and $k$ additions to hash $\mathbf{v}$. The time spent hashing all vectors in the list $L$ into the $t$ hash tables is thus $O(k\cdot |L|\cdot t)$, or more precisely, it requires $2k\cdot |L|\cdot t$ multiplications and $k\cdot |L|\cdot t$ additions. On the other hand, the list of candidates on \cref{line:candidates_nvsieve} has size $|C| \approx |S|\cdot p_2^\ast$, where $p_2^\ast$ is the average probability that a non-reducing vector collides with another vector in at least one of the $t$ hash tables (\Cref{eq:probability_collision_angular_hash}).

\paragraph{Classical complexity.} The classical time spent searching over $C$ is $O(D\cdot |S|\cdot |L|\cdot p_2^\ast)$ per sieving step. More precisely, according to \cref{table:oracle_modulos} (the number of classical arithmetic operations is the same as in the quantum case), searching over $C$ requires $D\cdot |L|\cdot |C| = D\cdot |L| \cdot |S|\cdot p_2^\ast$ multiplications and $2D\cdot |L| \cdot |C| = 2D\cdot |L|\cdot |S|\cdot p_2^\ast$ additions per sieving step. The number of hash tables $t$ is determined by balancing the time hashing $O(k\cdot |L|\cdot t)$ with the time searching $O(D\cdot |L|\cdot|S| \cdot p_2^\ast)$. Asymptotically over all sieving steps, $t$ is determined by taking $|L| = |S| = (4/3)^{D/2 + o(D)}$ and approximating $p_2^\ast \approx t\cdot 2^{-\beta D + o(D)}$, where $\beta$ is given by \cref{eq:beta_parameter_angular_hash}. We must then have that $(4/3)^{D/2 + o(D)}\cdot 2^{-\beta D + o(D)} = 2^{o(D)} \implies \beta = \frac{1}{2}\log_2(4/3)$, which yields $t \approx 2^{0.129043D}$. Therefore, by using $t \approx 2^{0.129043D}$ hash tables and a hash length of $k\approx 0.220600D$, the overall time and space complexities are $\widetilde{O}(|L|\cdot t) = 2^{0.336562D + o(D)}$~\cite{laarhoven15angular,Laarhoven2016search}.

\paragraph{Quantum complexity.} The quantum time searching over $C$ is $O(D\cdot |L|\sqrt{|S|\cdot p_2^\ast})$ per sieving step. The number of hash tables $t$ is determined by balancing the time hashing $O(k\cdot |L|\cdot t)$ with the time searching $O(D\cdot |L|\sqrt{|S| \cdot p_2^\ast})$. Asymptotically, $t$ is determined by taking $|L| = |S| = (4/3)^{D/2 + o(D)}$ and approximating $p_2^\ast \approx t\cdot 2^{-\beta D + o(D)}$, so that $(4/3)^{D/4 + o(D)}\cdot 2^{-\beta D/2 + o(D)} = \sqrt{t} \implies \beta = \frac{1}{2}\log_2(4/3) - \frac{1}{D}\log_2{t}$. This yields $t \approx 2^{0.078430D}$. Therefore, by using $t \approx 2^{0.078430D}$ hash tables and a hash length of $k \approx 0.134077D$, the overall time and space complexities are $\widetilde{O}(|L|\cdot t) = 2^{0.285949D + o(D)}$~\cite{Laarhoven2015Quantum,Laarhoven2016search}.

\subsubsection{NVSieve with spherical LSH} 

The complexity analysis of the $\mathtt{NVSieve}$ with spherical LSH (also known as $\mathtt{SphereSieve}$~\cite{Laarhoven2015Quantum}) is similar to $\mathtt{NVSieve}$ with angular LSH. When employing spherical LSH, we hash each vector into a string in $[2^{\sqrt{D}}]^k$ for each of the $t$ hash tables. Each hash table has thus $2^{k\sqrt{D}}$ buckets. We choose $k = 6(\ln{t} - \ln\ln(1/\varepsilon))/\sqrt{D}$ so that nearby vectors collide with high probability. Hashing one vector requires comparing $\langle \mathbf{v}, \mathbf{g}_{ij}\rangle \geq D^{1/4}$ for $i\in[2^{\sqrt{D}}]$ and $j\in[k]$. The time spent hashing all vectors in the list $L$ into the $t$ hash tables is thus $O(D \cdot 2^{\sqrt{D}}\cdot k \cdot t \cdot |L|)$, or more precisely, it requires $D\cdot 2^{\sqrt{D}}\cdot k\cdot t\cdot |L|$ multiplications and $D\cdot 2^{\sqrt{D}}\cdot k\cdot t\cdot |L|$ additions. On the other hand, the list of candidates on \cref{line:candidates_nvsieve} has size $|C| \approx |S|\cdot p_2^\ast$, where $p_2^\ast$ is the average probability that a non-reducing vector collides with another vector in at least one of the $t$ hash tables (\Cref{eq:probability_collision_spherical_hash}).

\paragraph*{Classical complexity.}  Classically searching over $C$ requires $D\cdot |L|\cdot |S|\cdot p_2^\ast$ multiplications and $D\cdot |L|\cdot |S|\cdot p_2^\ast$ additions per sieving step. The number of hash tables is determined by balancing the time hashing $O(D\cdot 2^{\sqrt{D}}\cdot k\cdot t\cdot |L|)$ and the time searching $O(D\cdot |L|\cdot |S|\cdot p_2^\ast)$. Asymptotically, $t$ is determined by taking $|L| = |S| = (4/3)^{D/2 + o(D)}$ and approximating $p_2^\ast \approx 2^{-\beta D + o(D)}$, where $\beta$ is given by \cref{eq:beta_parameter_spherical_hash}. Hence $(4/3)^{D/2 + o(D)}\cdot 2^{-\beta D + o(D)} = t \implies \beta = \frac{1}{2}\log_2(4/3) - \frac{1}{D}\log_2{t}$, which yields $t \approx 2^{0.089624D}$. Therefore, by using $t \approx 2^{0.089624D}$ hash tables and $k \approx 0.372737\sqrt{D}$, the time and space complexities are $2^{0.297143D + o(D)}$~\cite{laarhoven2015faster,Laarhoven2016search}.

\paragraph*{Quantum complexity.} Quantumly searching over $C$ requires $O(D\cdot |L|\sqrt{|S|\cdot p_2^\ast})$ time. The number of hash tables is determined by balancing the time hashing $O(D\cdot 2^{\sqrt{D}}\cdot k\cdot t\cdot |L|)$ and the time searching $O(D\cdot |L|\sqrt{|S|\cdot p_2^\ast})$. Asymptotically, $t$ is determined by taking $|L| = |S| = (4/3)^{D/2 + o(D)}$ and approximating $p_2^\ast \approx 2^{-\beta D + o(D)}$, so that $(4/3)^{D/4 + o(D)}\cdot 2^{-\beta D/2 + o(D)} = t \implies \beta = \frac{1}{2}\log_2(4/3) - \frac{2}{D}\log_2{t}$. This yields $t \approx 2^{0.059581D}$. Therefore, by using $t\approx 2^{0.059581D}$ hash tables and $k\approx 0.247792\sqrt{D}$, the overall time and space complexities are $2^{0.267100D+o(D)}$~\cite{Laarhoven2015Quantum,Laarhoven2016search}.

\subsubsection{NVSieve with spherical LSF} 

The complexity analysis of the $\mathtt{NVSieve}$ with spherical LSF is a bit different than with LSH, the main reason being that each filter bucket covers an equally large region of $\mathbb{R}^D$, which simplifies the analysis. As shown in~\cite{becker2016new}, fixing $k=1$ concatenated filters per bucket is usually optimal, as it allows a larger $\alpha$ parameter (see \cref{eq:collision_probability_filter}). On the other hand, the number of filter buckets $t$ is chosen so that nearby vectors collide with high probability $p_1^\ast \geq 1-\varepsilon$, where $p_1^\ast$ is given by \cref{eq:probability_collision_spherical_filter}, meaning that $t = \ln(1/\varepsilon)/\mathcal{W}_D(\alpha,\alpha,\pi/3)$. Each vector is on average contained in $t\cdot \mathcal{C}_D(\alpha)$ buckets, meaning there are $|S|\cdot t\cdot \mathcal{C}_D(\alpha)$ vectors in total in all buckets and $|S|\cdot \mathcal{C}_D(\alpha)$ vectors in each bucket on average. The list of candidates on \cref{line:candidates_nvsieve} has size $|C|\approx |S|\cdot t\cdot \mathcal{C}_D(\alpha)^2$. Inserting vectors into relevant filters requires an efficient oracle (as mentioned in \cref{sec:spherical_lsf}). Becker et al.~\cite{becker2016new} proposed such an efficient oracle, called $\mathtt{EfficientListDecoding}$, using random product codes. According to~\cite[Lemma~5.1]{becker2016new} (see \cref{fact:EfficientListDecoding}), the time complexity of such an oracle is mainly due to visiting at most $2\log_2{D}\cdot t\cdot \mathcal{C}_D(\alpha)$ nodes for a pruned enumeration, which requires mostly addition-like operations. We thus approximate the time to insert all the vectors in $L$ into relevant filters by $2\log_2{D}\cdot |L| \cdot t \cdot \mathcal{C}_D(\alpha)$ additions.

\paragraph{Classical complexity.} The classical time spent searching over $C$ requires $D\cdot |L|\cdot |S|\cdot t\cdot \mathcal{C}_D(\alpha)^2$ additions and multiplications per sieving step. Moreover, we also need to retrieve the relevant filters of $\mathbf{v}$ before performing the search over $C$. Retrieving the relevant filters of all vectors in $L$ requires $2\log_2{D}\cdot |L| \cdot t \cdot \mathcal{C}_D(\alpha)$ additions per sieving step. 
The parameter $\alpha$ is chosen in order to minimise the time complexity coming from filtering and from searching, $O(\log{D}\cdot |L|\cdot t \cdot \mathcal{C}_D(\alpha) + D\cdot |L|\cdot |S|\cdot t\cdot \mathcal{C}_D(\alpha)^2)$. 
Asymptotically, we approximate $t = O(1/\mathcal{W}_D(\alpha,\alpha,\pi/3))$ (see \cref{sec:spherical_lsf}). On the other hand, $\mathcal{C}_D(\alpha) = \poly(D) (1-\alpha^2)^{D/2}$~\cite[Lemma~4.1]{micciancio2010faster} and $\mathcal{W}_D(\alpha,\alpha,\theta) = \poly(D)(1 - 2\alpha^2/(1+\cos\theta))^{D/2}$~\cite[Lemma~2.2]{becker2016new}. Together with $|L| = (4/3)^{D/2 + o(D)}$, this means that the total time complexity over all sieving steps is
\begin{align*}
    \widetilde{O}\left(\frac{|L|\cdot \mathcal{C}_D(\alpha)(1 + |L|\cdot \mathcal{C}_D(\alpha))}{\mathcal{W}_D(\alpha,\alpha,\pi/3)}\right) = \widetilde{O}\left(\left(\frac{4(1-\alpha^2)}{3 - 4\alpha^2} \right)^{D/2}\left(1 + \left(\frac{4(1-\alpha^2)}{3} \right)^{D/2} \right) \right).
\end{align*}
The high-order term is minimised by taking $\alpha = 1/2$. Therefore, the time complexity is $(3/2)^{D/2 + o(D)} \approx 2^{0.292481D + o(D)}$ by choosing $\alpha = 1/2$, $k=1$, and $t=(3/2)^{D/2 + o(D)}$~\cite{becker2016new,Laarhoven2016search}. The space complexity is also $(3/2)^{D/2 + o(D)}$. The list of candidates, i.e., the list of vectors that collide with a given vector, has average size $|C| = |L|\cdot t\cdot \mathcal{C}_D(\alpha)^2 = (9/8)^{D/2 + o(D)}$.

\paragraph{Quantum complexity.} The quantum time spent comparing a given vector to other vectors colliding in one of the filters is now $O(D \sqrt{|S|\cdot t\cdot \mathcal{C}_D(\alpha)^2})$. The total time complexity of one sieving step with list $L$ is thus $O(\log{D}\cdot |L|\cdot t \cdot \mathcal{C}_D(\alpha) + D\cdot |L|\cdot |S|^{1/2}\cdot t^{1/2}\cdot \mathcal{C}_D(\alpha))$. The $\alpha$ parameter is chosen in order to minimise the classical hashing time plus the quantum searching time. Asymptotically, the approximations $t = O(1/\mathcal{W}_D(\alpha,\alpha,\pi/3))$, $\mathcal{C}_D(\alpha) = \poly(D) (1-\alpha^2)^{D/2}$, and $\mathcal{W}_D(\alpha,\alpha,\theta) = \poly(D)(1 - 2\alpha^2/(1+\cos\theta))^{D/2}$, together with $|L| = (4/3)^{D/2 + o(D)}$, yield the total time complexity over all sieving steps of
\begin{align*}
    \widetilde{O}\left(\frac{|L|\cdot \mathcal{C}_D(\alpha)}{\mathcal{W}_D(\alpha,\alpha,\pi/3)} + \frac{|L|^{3/2}\cdot \mathcal{C}_D(\alpha)}{\sqrt{\mathcal{W}_D(\alpha,\alpha,\pi/3)}}\right) = \widetilde{O}\left(\left(\frac{4(1-\alpha^2)}{3 - 4\alpha^2}\right)^{D/2} + \left(\frac{8(1-\alpha^2)}{3\sqrt{3 - 4\alpha^2}} \right)^{D/2} \right).
\end{align*}
The high-order term is minimised by taking $\alpha = \sqrt{3}/4$. Hence the time complexity is $(13/9)^{D/2 + o(D)} \approx 2^{0.265257D + o(D)}$ by choosing $\alpha = \sqrt{3}/4$, $k=1$, and $t=(4/3)^{D/2 + o(D)}$~\cite{Laarhoven2016search}. The space complexity is also $(13/9)^{D/2 + o(D)}$. The list of candidate vectors has average size $|C| = |L|\cdot t\cdot \mathcal{C}_D(\alpha)^2 = (13/12)^{D + o(D)}$.

\subsection{The GaussSieve}

A few years after the work of Nguyen and Vidick, Micciancio and Voulgaris~\cite{micciancio2010faster} presented $\mathtt{ListSieve}$, a probabilistic algorithm that provably finds the shortest vector with a high probability in $2^{3.199D + o(D)}$ time and $2^{1.325D + o(D)}$ space, and a heuristic variant called $\mathtt{GaussSieve}$, which we now focus on and is described in \cref{alg:gauss_sieve}. The $\mathtt{GaussSieve}$ starts with an empty list $L$ and keeps adding shorter lattice vectors to it. At each sieve step, a new lattice vector $\mathbf{v}$ is reduced with all the points in the list $L$. By this we mean the rule:
\begin{align*}
    \text{Reduce}~\mathbf{v}~\text{with}~\mathbf{w}:\quad\text{if}~\|\mathbf{v}\pm\mathbf{w}\| < \|\mathbf{v}\| ~\text{then}~ \mathbf{v}\gets \mathbf{v}\pm\mathbf{w}.
\end{align*}
The difference between both sieves is that, in the $\mathtt{ListSieve}$, the reduced vector is then added to the list, meaning that vectors in $L$ never change, while in the $\mathtt{GaussSieve}$, we also attempt to reduce the vectors in $L$ with $\mathbf{v}$ before adding $\mathbf{v}$ to $L$. In other words, in the $\mathtt{GaussSieve}$, for all vectors $\mathbf{v}, \textbf{w} \in L$ such that $\min(\|\mathbf{v}\pm \mathbf{w}\|) < \max(\|\mathbf{v}\|, \|\mathbf{w}\|)$, the longer of $\mathbf{v}$ and $\mathbf{w}$ is replaced with the shorter of $\mathbf{v}\pm \mathbf{w}$. Consequently, all pairs of vectors in the list are always pairwise reduced: $\forall \textbf{v}, \textbf{w} \in L : \min(\|\mathbf{v} \pm \mathbf{w}\|) \geq \max(\|\mathbf{v}\|, \|\mathbf{w}\|)$. Thus any pair of vectors in the list always form a Gauss reduced basis for a two dimensional lattice, and thus the angle between any two list points is at least $\pi/3$ and the list forms a good spherical code. It follows that the size of the list never exceeds the kissing constant $\tau_D$ in $D$ dimensions. Therefore the list size (and thus the space complexity of the $\mathtt{GaussSieve}$) is bounded by $2^{0.401D}$ in theory and $2^{0.208D}$ in practice, corresponding to the asymptotic upper and lower bounds on $\tau_D$~\cite{kl78}. In contrast, there are no known bounds on the time complexity of the $\mathtt{GaussSieve}$, since the list $L$ can grow or shrink at any time. One might guess that the time complexity is quadratic in the list size since at each sieving step every pair of vectors $\mathbf{v},\mathbf{w}\in L$ is compared at least once to check for possible reductions. Furthermore, the asymptotic behaviour of the $\mathtt{GaussSieve}$ is similar to that of the $\mathtt{NVSieve}$~\cite{micciancio2010faster}. A natural conjecture is that the $\mathtt{GaussSieve}$ has time complexity $\widetilde{O}(|L|^2) = 2^{0.415D + o(D)}$.

\begin{algorithm}[t]
    \caption{The $\mathtt{GaussSieve}$}\label{alg:gauss_sieve}
    \SetKwFor{ForEach}{for each}{do}{endfch}
    \SetKw{And}{and}
    
    \KwIn{Basis $\mathbf{B}$ for a $D$-dimensional lattice and termination constant $\zeta$.}
    \KwOut{Shortest vector $\mathbf{v}^\ast$ of the lattice.}

    $L \gets \emptyset$, $S \gets \emptyset$, $K \gets 0$

    \While{$K < \zeta$}
    {
        \If{$S = \emptyset$}
        {
            $\mathbf{v}\gets \text{SampleKlein}(\mathbf{B})$
        }
        \Else
        {
            $\mathbf{v}\gets S.\text{pop}()$
        }

        \While{$\mathbf{w} \gets \mathtt{GroverSearch}(\mathbf{w}\in L : \|\mathbf{v} \pm \mathbf{w}\| < \|\mathbf{v}\|)$ \And $\mathbf{w} \neq \mathtt{NULL}$ \label{line:gauss_sieve_search1}} 
        {
            Reduce $\mathbf{v}$ with $\mathbf{w}$
        }
        \While{$\mathbf{w} \gets \mathtt{GroverSearch}(\mathbf{w}\in L : \|\mathbf{w} \pm \mathbf{v}\| < \|\mathbf{w}\|)$ \And $\mathbf{w} \neq \mathtt{NULL}$ \label{line:gauss_sieve_search2}} 
        {
            $L\gets L\setminus\{\mathbf{w}\}$
        
            Reduce $\mathbf{w}$ with $\mathbf{v}$

            \If{$\mathbf{w}\neq \mathbf{0}$}
            {
                $S \gets S\cup\{\mathbf{w}\}$
            }
        }
        \If{$\mathbf{v} ~ \mathrm{has~changed}$ \And $\mathbf{v}\neq \mathbf{0}$}
        {
            $S.\text{push}(\mathbf{v})$
        }
        \If{$\mathbf{v} = \mathbf{0}$}
        {
            $K \gets K+1$
        }
        \Else
        {
            $L \gets L\cup\{\mathbf{v}\}$
        }
    }   
    \Return shortest vector $\textbf{v}^\ast$ in $L$
\end{algorithm}

\subsubsection{Numerical experiments and heuristic assumptions}
\label{sec:heuristics_gausssieve}

Once again, the asymptotic complexity hides a lot of operations when doing a resource estimate. The overall analysis of $\mathtt{GaussSieve}$ is more complicate than the $\mathtt{NVSieve}$, since the size of the list $L$ can both increase and decrease, which hinders a bound on time complexity. Moreover, the search loops in \cref{line:gauss_sieve_search1,line:gauss_sieve_search2} are performed in an exhaustive manner, meaning that a search will be attempted while there are solutions. Nonetheless, it is still possible to gather average trends and bounds through heuristic assumptions and numerical experiments. In the following, we list several observations that shall be useful in forming assumptions.
\begin{enumerate}
    \item Schneider~\cite{schneider11} noticed that $\mathtt{GaussSieve}$’s performance in terms of runtime, iterations, list size, and collisions was not affected by the type of the underlying lattice (ideal, cyclic, and random). 
    
    \item Micciancio and Voulgaris~\cite{micciancio2010faster} proved that the list size $|L|$ never exceeds the kissing number $\tau_D$, which is defined as the highest number of points that can be placed on a $D$-dimensional sphere such that the angle between any two points is at least $\pi/3$. This theoretically bounds $|L|$ by $\tau_D \leq 2^{0.401D + o(D)}$. However, Micciancio and Voulgaris~\cite{micciancio2010faster} numerically observed that the maximum list size grows approximately as $2^{0.2D}$, which matches lower bounds on the kissing number $\tau_D \geq 2^{0.2075D + o(D)}$~\cite{conway2013sphere}. A plausible assumption is the maximum list size to be bounded by a lower bound on the kissing number, e.g., $\tau_D \geq (1+o(1))\frac{\sqrt{3\pi}}{4\sqrt{2}}\ln(3/2) D^{3/2} (4/3)^{D/2}$~\cite{fernandez2021new}. From a more numerical perspective, Schneider~\cite{schneider11} reported a maximum list size of $2^{0.2D + 2.8}$, while Mariano et al.~\cite{mariano14comprehensive} reported a maximum list size of $2^{0.199D + 2.149}$. We independently report a maximum list size of $2^{0.193D + 2.325}$.
    
    \item Schneider~\cite{schneider11} observed that the number of times a newly sampled vector from Klein's algorithm was reduced by the list vectors and the number of vectors removed from the list $L$ and pushed to the stack $S$ were approximately $10$ times the maximum list size. This means that on average the first search loop (\cref{line:gauss_sieve_search1}) is performed $10$ times. This observation was independently confirmed by us. The number of solutions to Grover's search in \cref{line:gauss_sieve_search1} varies greatly. A more pessimistic assumption is to take $M=1$ or $M=2$ solutions for each of the first $9$ calls, while the $10$-th call has $M=0$ solutions. On the other hand, the second search loop (\cref{line:gauss_sieve_search2}) is performed only once with $M=0$ number of solutions on the vast majority of cases.
    
    \item The number of sieving steps is roughly $10$ times the maximum list size (see~\cite[Figures~1 and~2]{schneider11}). Mariano et al.~\cite{mariano14comprehensive} numerically reported the number of iterations $I$ to grow as $2^{0.283D + 0.335}$, while we obtained a growth of $2^{0.283D + 0.491}$.
    
    \item A natural termination criteria proposed by Micciancio and Voulgaris~\cite{micciancio2010faster} is to stop after a certain number $\zeta$ of collisions. A conservative option for $\zeta$ adopted by~\cite{micciancio2010faster} is to set it as $10\%$ of the list size. The authors\footnote{See Appendix B of the \href{https://cseweb.ucsd.edu/~daniele/papers/Sieve.pdf}{unpublished version}.} also used an alternative criteria of $\zeta = 500$, which we independently checked to be enough to find the shortest vector. Under such criteria, the list size does not grow much beyond the point where a shortest vector is found.
    
    \item The list size $|L|$ starts from $0$ and quickly grows to an asymptote which, according to the previous point, roughly corresponds to the maximum list size. Meanwhile, collisions rarely occur before the shortest vector is found, after which the number of collisions quickly grows until the exit-condition is reached. The list size stays above $90\%$ of the maximum list size (i.e., the list size at the moment a shortest vector is found) for more than $95\%$ the number of iterations for large enough dimensions ($D>70$).
\end{enumerate}

\subsubsection{Quantum oracle for Grover search}

Similarly to the $\mathtt{NVSieve}$, the two search steps in the $\mathtt{GaussSieve}$ (\cref{line:gauss_sieve_search1,line:gauss_sieve_search2} in \cref{alg:gauss_sieve}) can be performed using Grover's algorithm. Namely, given an ordered list $L = \{\mathbf{w}_1,\mathbf{w}_2,\dots\}$ and a fixed element $\mathbf{v}\in L$, 
\begin{enumerate}
    \item Find an index $i\in[|L|]$ such that $\|\mathbf{v}\pm \mathbf{w}_i\| < \|\mathbf{v}\| \iff \mathbf{w}_i\cdot(\mathbf{w}_i \pm 2\mathbf{v}) < 0$;
    \item Find an index $i\in[|L|]$ such that $\|\mathbf{v}\pm \mathbf{w}_i\| < \|\mathbf{w}_i\| \iff |\mathbf{v}\cdot \mathbf{w}_i| \geq \|\mathbf{v}\|^2/2$.
\end{enumerate}
Define the Boolean function $f_{\rm gauss}:[|L|] \to \{0,1\}$ by $f_{\rm gauss}(i) = 1$ if and only if either $\mathbf{w}_i\cdot(\mathbf{w}_i + 2\mathbf{v}) < 0$ or $\mathbf{w}_i\cdot(\mathbf{w}_i - 2\mathbf{v}) < 0$. Similarly, let $g_{\rm gauss}:[|L|] \to \{0,1\}$ be such that $g_{\rm gauss}(i) = 1$ if and only if $|\mathbf{v}\cdot \mathbf{w}_i| \geq \|\mathbf{v}\|^2/2$. In order to use Grover search in \cref{line:gauss_sieve_search1}, we must construct the phase oracle $\mathcal{O}^{(1)}_{\rm gauss}:|i\rangle \mapsto (-1)^{f_{\rm gauss}(i)}|i\rangle$, while the Grover search in \cref{line:gauss_sieve_search2} requires the phase oracle $\mathcal{O}^{(2)}_{\rm gauss}:|i\rangle \mapsto (-1)^{g_{\rm gauss}(i)}|i\rangle$. We now describe how they can be constructed.

\paragraph{Phase oracle $\mathcal{O}^{(1)}_{\rm gauss}$.} The construction is similar to the one for the $\mathtt{NVSieve}$. We assume that the list $L$ is already stored in $\mathsf{QRAM}$. Given any index $|i\rangle$ where $i\in[|L|]$, the first step is to load $\mathbf{w}_i$ onto a $(\kappa D)$-qubit register using one $\mathsf{QRAM}_L$ call (\cref{lem:qram_resources}). Since we must check for two conditions, $\mathbf{w}_i\cdot(\mathbf{w}_i + 2\mathbf{v}) < 0$ or $\mathbf{w}_i\cdot(\mathbf{w}_i - 2\mathbf{v}) < 0$, we copy $|\mathbf{w}_i\rangle$ onto another $(\kappa D)$-qubit ancillary register using $\kappa D$ $\mathsf{CNOT}$s. We then use $2D$ $\kappa$-bit out-of-place adders in parallel to get $|i\rangle|\mathbf{w}_i\rangle^{\otimes 2}|\mathbf{w}_i + 2\mathbf{v}\rangle|\mathbf{w}_i - 2\mathbf{v}\rangle$. Next, all the terms $(\mathbf{w}_i)_j(\mathbf{w}_i\pm 2\mathbf{v})_j$, $j\in[D]$, are computed in parallel using $2D$ $\kappa$-bit out-of-place multipliers. Then, all $D$ terms $(\mathbf{w}_i)_j(\mathbf{w}_i+ 2\mathbf{v})_j$ are summed in depth $\lceil\log_2{D}\rceil$ by using $D-1$ $\kappa$-bit out-of-place adders, and similarly for the terms $(\mathbf{w}_i)_j(\mathbf{w}_i- 2\mathbf{v})_j$. In order to check whether $\mathbf{w}_i\cdot(\mathbf{w}_i \pm 2\mathbf{v})$ is smaller than $0$, it suffices to consider its highest-order bit. Since at most one of the conditions $\mathbf{w}_i\cdot(\mathbf{w}_i \pm 2\mathbf{v}) < 0$ can be true, we simply compute the parity of their high-bits instead of their logical $\mathsf{OR}$. Thus, by applying two $\mathsf{CNOT}$s controlled on the high-bits of $\mathbf{w}_i\cdot(\mathbf{w}_i \pm 2\mathbf{v})$ onto a qubit in the $|-\rangle$ state, we implement the phase $(-1)^{f_{\rm gauss}(i)}$. After that, we uncompute all the arithmetic operations, copying of $|\mathbf{w}_i\rangle$, and $\mathsf{QRAM}$ call. The amount of submodules is summarised in \cref{table:oracle_modulos}.

\paragraph{Phase oracle $\mathcal{O}_{\rm gauss}^{(2)}$.} Once again, one $\mathsf{QRAM}_L$ call is used to load $\mathbf{w}_i$, after which $D$ $\kappa$-bit hybrid multipliers are used to obtain all the $D$ terms $(\mathbf{w}_i)_j v_j$, $j\in[D]$. These $D$ terms are then summed up in depth $\lceil\log_2{D}\rceil$ using $D-1$ $\kappa$-bit out-of-place adders. At this point, one of the registers is $|\mathbf{v}\cdot \mathbf{w}_i\rangle$. In order to check for the condition $|\mathbf{v}\cdot \mathbf{w}_i| \geq \|\mathbf{v}\|^2/2$, we can first compute the sum $\mathbf{v}\cdot \mathbf{w}_i - \|\mathbf{v}\|^2/2$ by using a $\kappa$-bit adder and copy its highest-order bit onto a qubit in the $|-\rangle$ state for a phase kickback. The adder generates an ancillary register containing $|\mathbf{v}\cdot \mathbf{w}_i - \|\mathbf{v}\|^2/2\rangle$. In order to check whether $-\mathbf{v}\cdot\mathbf{w}_i \geq \|\mathbf{v}\|^2/2 \iff \mathbf{v}\cdot\mathbf{w}_i - \|\mathbf{v}\|^2/2 < -\|\mathbf{v}\|^2$, we can apply a second $\kappa$-bit adder between the ancillary register $|\mathbf{v}\cdot \mathbf{w}_i - \|\mathbf{v}\|^2/2\rangle$ and the classical input $\|\mathbf{v}\|^2$. The highest-order bit of the result $|\mathbf{v}\cdot \mathbf{w}_i + \|\mathbf{v}\|^2/2\rangle$ is then flipped, since we are checking for a negative number, and copied onto the $|-\rangle$ ancilla for a phase kickback. This implements the phase $(-1)^{g_{\rm gauss}(i)}$ as at most one condition $\mathbf{v}\cdot \mathbf{w}_i \geq \|\mathbf{v}\|^2/2$ or $-\mathbf{v}\cdot \mathbf{w}_i \geq \|\mathbf{v}\|^2/2$ can be true. After this, we uncompute all the arithmetic operations and $\mathsf{QRAM}$ call. The required submodules are summarised in \cref{table:oracle_modulos}.

\subsubsection{Using LSH/LSF in the GaussSieve} 

\begin{algorithm}[t]
    \caption{The $\mathtt{GaussSieve}$ with LSH}\label{alg:gauss_sieve_lsh}
    \SetKwFor{ForEach}{for each}{do}{endfch}
    \SetKw{And}{and}
    
    \KwIn{Basis $\mathbf{B}$ for a $D$-dimensional lattice, termination constant $\zeta$, parameters $k$ and $t$, hash family $\mathcal{H}$}
    \KwOut{Shortest vector $\mathbf{v}^\ast$ of the lattice.}

    $L \gets \emptyset$, $S \gets \emptyset$, $K \gets 0$

    Initialise $t$ empty hash tables $\mathcal{T}_1,\dots,\mathcal{T}_t$ and sample $k\cdot t$ random hash functions $h_{i,j}\in\mathcal{H}$

    \While{$K < \zeta$}
    {
        \If{$S = \emptyset$}
        {
            $\mathbf{v}\gets \text{SampleKlein}(\mathbf{B})$
        }
        \Else
        {
            $\mathbf{v}\gets S.\text{pop}()$
        }

        $C\gets \bigcup_{i=1}^t \mathcal{T}_i[h_i(\mathbf{v})]$ is the list of candidate vectors \label{line:candidates_gauss}

        Construct a $\mathsf{QRAM}$ for $C$

        \While{$\mathbf{w} \gets \mathtt{GroverSearch}(\mathbf{w}\in C : \|\mathbf{v} \pm \mathbf{w}\| < \|\mathbf{v}\|)$ \And $\mathbf{w} \neq \mathtt{NULL}$ \label{line:gauss_sieve_search1_lsh}} 
        {
            Reduce $\mathbf{v}$ with $\mathbf{w}$
        }
        \While{$\mathbf{w} \gets \mathtt{GroverSearch}(\mathbf{w}\in C : \|\mathbf{w} \pm \mathbf{v}\| < \|\mathbf{w}\|)$ \And $\mathbf{w} \neq \mathtt{NULL}$ \label{line:gauss_sieve_search2_lsh}} 
        {
            $L\gets L\setminus\{\mathbf{w}\}$

            Remove $\mathbf{w}$ from all hash tables $\mathcal{T}_1,\dots,\mathcal{T}_t$
        
            Reduce $\mathbf{w}$ with $\mathbf{v}$

            \If{$\mathbf{w}\neq \mathbf{0}$}
            {
                $S \gets S\cup\{\mathbf{w}\}$
            }
        }
        \If{$\mathbf{v} ~ \mathrm{has~changed}$ \And $\mathbf{v}\neq \mathbf{0}$}
        {
            $S.\text{push}(\mathbf{v})$
        }
        \If{$\mathbf{v} = \mathbf{0}$}
        {
            $K \gets K+1$
        }
        \Else
        {
            $L \gets L\cup\{\mathbf{v}\}$

            For each $i\in[t]$, add $\mathbf{v}$ to the bucket $\mathcal{T}_i[h_i(\mathbf{v})]$ in the hash table $\mathcal{T}_i$
        }
    }   
    \Return shortest vector $\textbf{v}^\ast$ in $L$
\end{algorithm}

Similarly to the $\mathtt{NVSieve}$, LSH/LSF can be used in the $\mathtt{GaussSieve}$ as a filter to get a preliminary set of vectors to search among: instead of using a brute-force list search, we can only search through the candidate vectors $C$ that hash to the same value (that is, they are close-by). The modified algorithm is given in \cref{alg:gauss_sieve_lsh}. The main idea is again to employ hash tables $\mathcal{T}_1,\dots,\mathcal{T}_t$ and replace the search over the entire list $L$ with a shorter list of candidates $C = \bigcup_{i=1}^t \mathcal{T}_i[h_i(\mathbf{v})]$ that collide with the target vector $\mathbf{v}$ in at least one of the buckets $\mathcal{T}_1[h_1(\mathbf{v})],\dots,\mathcal{T}_t[h_t(\mathbf{v})]$. Once again, in order to use Grover's search, we must first classically gather all the indices of the vectors that collide with~$\mathbf{v}$. If $C = \{\mathbf{w}_{j_1},\mathbf{w}_{j_2},\dots,\mathbf{w}_{j_{|C|}}\}$, then we use the indices $\{j_1,j_2,\dots,j_{|C|}\}$ to access the vectors in $C$ via $\mathsf{RAM}$ calls and thus perform the classically controlled $\mathsf{CNOT}$s in during a $\mathsf{QRAM}$ call. The phase $\mathcal{O}_{\rm gauss}^{(1)}:|i\rangle\mapsto(-1)^{f_{\rm gauss}(i)}|i\rangle$, where $f_{\rm gauss}:[|C|]\to\{0,1\}$ is defined by $f_{\rm gauss}(i) = 1$ if and only if either $\mathbf{w}_{j_i}\cdot(\mathbf{w}_{j_i}+2\mathbf{v}) < 0$ or $\mathbf{w}_{j_i}\cdot(\mathbf{w}_{j_i}-2\mathbf{v}) < 0$, is hence implemented first by one $\mathsf{QRAM}$ call and $|C|$ $\mathsf{RAM}$ calls, similarly to \cref{eq:nv_sieve_qram}, after which the remaining addition and multiplication operations explained in the previous section are performed (plus overall uncomputation). The exact same procedure is required in the phase oracle $\mathcal{O}_{\rm gauss}^{(2)}:|i\rangle\mapsto(-1)^{g_{\rm gauss}(i)}|i\rangle$, where $g_{\rm gauss}:[|C|]\to\{0,1\}$ is defined by $g_{\rm gauss}(i) = 1$ if and only if $|\mathbf{v}\cdot \mathbf{w}_{j_i}| \geq \|\mathbf{v}\|^2/2$. Both oracles $\mathcal{O}_{\rm gauss}^{(1)}$ and $\mathcal{O}_{\rm gauss}^{(2)}$ thus require one $\mathsf{QRAM}$ call of size $|C|$. \cref{table:oracle_modulos} summarises the subroutines needed to implement both phase~oracles.

\subsubsection{GaussSieve with angular LSH} 

Hashing all vectors in the list $L$ requires, similarly to $\mathtt{NVSieve}$, $2k\cdot |L|\cdot t$ multiplications and additions, where $k = \log_{3/2}{t} - \log_{3/2}\ln(1/\varepsilon)$. The list of candidates on \cref{line:candidates_gauss} has size $|C| \approx |L|\cdot p_2^\ast$, where $p_2^\ast$ is given by \Cref{eq:probability_collision_angular_hash}.

\paragraph{Classical complexity.} The classical time spent searching over $C$ is $O(D\cdot I \cdot |L|\cdot p^\ast_2)$, where $I$ is the number of iterations of $\mathtt{GaussSieve}$. To be more precise, the first search loop over $C$ (\cref{line:gauss_sieve_search1_lsh}) requires $4D - 2$ additions and $2D$ multiplications to check whether one vector can reduce another, while the second search loop over $C$ (\cref{line:gauss_sieve_search2_lsh}) requires $D+1$ additions and $D$ multiplications (see \cref{table:oracle_modulos}). The number of hash tables $t$ is determined by balancing the time hashing $O(k\cdot |L|\cdot t)$ with the time searching $O(D\cdot I \cdot |L| \cdot p_2^\ast)$. The asymptotic classical time and space complexities are $2^{0.336562D + o(D)}$~\cite{laarhoven15angular,Laarhoven2016search}. We stress that the time complexities are only conjectures, in contrast to the $\mathtt{NVSieve}$, where bounds can be proven under reasonable assumptions.

\paragraph{Quantum complexity.} The quantum time spent searching over $C$ is $O(D\cdot I \sqrt{|L|\cdot p^\ast_2})$. The number of hash tables $t$ is determined by balancing the time hashing $O(k\cdot |L|\cdot t)$ with the time searching $O(D\cdot I \sqrt{|L| \cdot p_2^\ast})$. The asymptotic quantum time and space complexities are $2^{0.285949D+o(D)}$~\cite{Laarhoven2015Quantum,Laarhoven2016search}.

\subsubsection{GaussSieve with spherical LSH} 

Hashing all vectors in the list $L$ requires, similarly to $\mathtt{NVSieve}$, $D\cdot 2^{\sqrt{D}} \cdot k \cdot t \cdot |L|$ multiplications and additions, where $k = 6(\ln{t} - \ln\ln(1/\varepsilon))/\sqrt{D}$. The list of candidates on \cref{line:candidates_gauss} has size $|C| \approx |L|\cdot p_2^\ast$, where $p_2^\ast$ is given by \Cref{eq:probability_collision_spherical_hash}.

\paragraph*{Classical complexity.} Similarly to angular LSH, the first search loop over $C$ (\cref{line:gauss_sieve_search1_lsh}) requires $4D - 2$ additions and $2D$ multiplications to check whether one vector can reduce another, while the second search loop over $C$ (\cref{line:gauss_sieve_search2_lsh}) requires $D+1$ additions and $D$ multiplications (see \cref{table:oracle_modulos}). The number of hash tables $t$ is determined by balancing the time hashing $O(D\cdot 2^{\sqrt{D}}\cdot k\cdot t \cdot |L|)$ with the time searching $O(D\cdot I \cdot |L| \cdot p_2^\ast)$. The asymptotic classical time and space complexities are $2^{0.297143D + o(D)}$~\cite{laarhoven15angular,Laarhoven2016search}. 

\paragraph*{Quantum complexity.} The quantum time spent searching over $C$ is $O(D\cdot I \sqrt{|L|\cdot p^\ast_2})$. The number of hash tables $t$ is determined by balancing the time hashing $O(k\cdot |L|\cdot t)$ with the time searching $O(D\cdot I \sqrt{|L| \cdot p_2^\ast})$. The asymptotic quantum time and space complexities are $2^{0.267100D+o(D)}$~\cite{Laarhoven2015Quantum,Laarhoven2016search}.

\subsubsection{GaussSieve with spherical LSF}

We fix $k=1$ concatenated filters per bucket and $t = \ln(1/\varepsilon)/\mathcal{W}_D(\alpha,\alpha,\pi/3)$ hash tables. Inserting all the vectors in $L$ into relevant filters requires approximately $2\log_2{D}\cdot |L| \cdot t\cdot \mathcal{C}_D(\alpha)$ additions.  The list of candidates on \cref{line:candidates_nvsieve} has size $|C|\approx |L|\cdot t\cdot \mathcal{C}_D(\alpha)^2$.

\paragraph*{Classical complexity.} Again, the first search loop over $C$  requires $4D - 2$ additions and $2D$ multiplications to check whether one vector can reduce another, while the second search loop over $C$ requires $D+1$ additions and $D$ multiplications. The parameter $\alpha$ is determined by minimising the sum of the time coming from filtering $O(\log{D}\cdot |L|\cdot t\cdot \mathcal{C}_D(\alpha))$ and the time coming from searching $O(D\cdot I\cdot |L|\cdot t\cdot \mathcal{C}_D(\alpha)^2)$. The asymptotic classical time and space complexities are $(3/2)^{D/2 + o(D)} \approx 2^{0.292481D + o(D)}$~\cite{becker2016new,Laarhoven2016search}.

\paragraph*{Quantum complexity.} The quantum time spent comparing vectors that collide on relevant filters is now $O(D\cdot I\sqrt{|L|\cdot t\cdot \mathcal{C}_D(\alpha)^2})$. The parameter $\alpha$ is determined by minimising the time required to filter plus the time required to search. The asymptotic quantum complexities are $(13/9)^{D/2 + o(D)} \approx 2^{0.265257D + o(D)}$~\cite{Laarhoven2016search}.

\section{Resource estimation analysis}
\label{sec:results}

In this section, we perform a thorough resource estimation required to implement Grover's search to speed-up the $\mathtt{NVSieve}$ and $\mathtt{GaussSieve}$ algorithms, both with and without LSH techniques. For such, we take into consideration the cost of arithmetic circuits from \Cref{sec:arithmetic} and $\mathsf{QRAM}$ from \Cref{sec:qram} in implementing the phase oracles from Grover's search (\Cref{sec:grover_search}), together with the overhead coming from quantum error correction and magic state distillation from \Cref{sec:error_correction}. Our analysis will cover several facets from the quantum computation part within the sieving algorithms: circuit size and depth, number of logical and physical qubits, and overall runtime. Moreover, we shall analyse the most expensive sieving step and the total cost of \emph{all} sieving steps (which includes smaller list sizes). We shall also gauge the impact of an error-corrected $\mathsf{QRAM}$ by suppressing its costs and comparing the end result with the full algorithmic cost. This shall be important from NIST's standpoint, because for the purpose of the standardization of post-quantum cryptographic technologies, it would be prudent to consider the possibility of a breakthrough quantum memory architecture making efficient queries possible. We start by describing how all the pieces from the previous sections fit together and the cost analysis is done in the case of lattice dimension $D = 400$, which is roughly the dimension in which SVP has to be solved to be able to break the minimally secure post-quantum cryptographic standards currently being standardised (see~\cite[Table~1]{Dilithium2021Algorithm} and~\cite[Table~4]{kyber2021}).

\subsection{Case study: $D=400$}
\label{sec:case_study}

\subsubsection{NVSieve without LSH/LSF}

Let us consider the case where the rank of the lattice is $D = 400$ and analyse the cost of employing Grover's search in the $\mathtt{NVSieve}$ without LSH/LSF from \cref{alg:nv_sieve}. For simplicity, we will focus on one Grover's search. Even though the sizes of $L$ and $S$ are random, we assume a worst-case list of centers of size $|S| = 2^{0.2352D + 0.102\log_2{D} + 2.45} \approx 2.15\cdot 10^{29}$ and a list of size $|L| = D|S| \approx 8.61\cdot 10^{31}$ as mentioned in \cref{sec:heuristics_nvsieve}. Moreover, we assume there is only one solution to each Grover's search.

\paragraph{Logical costs.} The first step is to gather all the logical costs like $\mathsf{Toffoli}$-count, number of logical qubits (circuit's width), and the circuit's active volume. According to \Cref{table:oracle_modulos}, the phase oracle $\mathcal{O}_{\rm NV}$ from Grover's search requires $1$ $\mathsf{QRAM}$ call, $2D$ $\kappa$-bit adders, and $D$ $\kappa$-bit multipliers. Since the expected number of Grover iterations is $\lceil 3.1\sqrt{|S|}\rceil$ per Grover's search, we require
\begin{align*}
    \mathsf{Toffoli}\text{-count}: \underbrace{\lceil 3.1\sqrt{|S|}\rceil}_{\rm Grover~iterations}\big(\underbrace{|S| - 2}_{\mathsf{QRAM}} + \underbrace{2D(\kappa - 1)}_{2D~\text{adders}} + \underbrace{D(\kappa^2 - \kappa + 1)}_{D~\text{multipliers}} + \underbrace{\lceil\log_2|S|\rceil - 1}_{\rm Diffusion~operator}\big) \approx 3.09\cdot 10^{44}.
\end{align*}

Regarding the number of logical qubits, ancillae can be reused from one iteration to the next, so the maximum width (dirty ancillae plus input/output qubits) of Grover's circuit comes from $\mathsf{QRAM}$ plus the arithmetic operations and diffusion operator. One $\mathsf{QRAM}$ call needs $2|S| - \lceil\log_2|S|\rceil - 1$ dirty ancillae, plus $\lceil\log_2|S|\rceil + D\kappa$ qubits from input/output registers. On the other hand, the first $D$ adders have a width of $3D\kappa$; the following $D$ multipliers have a width of $D(2\kappa^2 + \kappa)$; the subsequent $D-1$ adders have a width of $(2D-1)\kappa$; the final adder has a width of $3\kappa$. Taking into account the overlap between different widths, since the output of one step is the input of the subsequent one, the total amount of logical qubits required is
\begin{align*}
   \text{Logical qubits}: 2\big(\underbrace{2|S| + D\kappa - 1}_{\mathsf{QRAM}} + \underbrace{2D\kappa}_{D~\text{adders}} + \underbrace{2D\kappa^2}_{D~\text{multipliers}} + \underbrace{D\kappa + \kappa}_{D~\text{adders}}\big) \approx 8.61\cdot 10^{29},
\end{align*}
where the factor $2$ takes into account the space overhead coming from fast data blocks in baseline architectures, and from workspace qubits in active-volume architectures.

The active volume of the whole circuit is calculated by simply summing up the active volumes of $1$ $\mathsf{QRAM}$ call, $2D$ $\kappa$-bit adders, and $D$ $\kappa$-bit multipliers. Using the bucket-bridage $\mathsf{QRAM}$ from \Cref{lem:qram_resources} and $C_{|CCZ\rangle} = 65$, the active volume of one Grover's search is
\begin{align*}
\begin{multlined}[\textwidth]
    \text{Active volume} : \underbrace{\lceil 3.1\sqrt{|S|}\rceil}_{\rm Grover~iterations}\big(\underbrace{(25 + 1.5\kappa + C_{|CCZ\rangle})|S|}_{\mathsf{QRAM}} + \underbrace{2D((\kappa-1)(39+C_{|CCZ\rangle}) + 7)}_{\rm Adders} \\
    + \underbrace{D(28\kappa^2 - 42\kappa + 28 + (\kappa^2 - \kappa + 1)C_{|CCZ\rangle})}_{\rm Multipliers} + \underbrace{(\lceil\log_2 |S|\rceil - 1)(18+C_{|CCZ\rangle})}_{\rm Diffusion~operator}\big) \approx 4.27\cdot 10^{46}.
 \end{multlined}
\end{align*}

The reaction depth (which, in our case, is double the $\mathsf{Toffoli}$-depth) follows from a simple concatenation of all the individual operations. The reaction depth of the phase oracle is the sum of reaction depths of one $\mathsf{QRAM}$ call, one $\kappa$-bit multiplier, and $2+\lceil\log_2D\rceil$ $\kappa$-bit adders. By adding the reaction depth of the diffusion operator (\Cref{fact:grover_search_unknown}) and multiplying the result by the number of Grover iterations $\lceil 3.1\sqrt{|S|}\rceil$, the get
\begin{align*}
    \text{Reaction depth}: \underbrace{\lceil 3.1\sqrt{|S|}\rceil}_{\rm Grover~iterations}\big(&\underbrace{2\lceil\log_2|S|\rceil - 2}_{\mathsf{QRAM}} + \underbrace{2\kappa\log_2\kappa - 2\kappa - 2\log_2\kappa + 4}_{\rm Multipliers} \\
    &+ \underbrace{2(\lceil\log_2 D\rceil + 2)(\kappa - 1)}_{\rm Adders} + \underbrace{2\lceil\log_2 \lceil\log_2 |S|\rceil\rceil}_{\rm Diffusion~operator}\big) \approx 1.64\cdot 10^{18}.
\end{align*}

\paragraph{Code distance and time.} First consider a baseline architecture. Consider that there are enough distillation factories (see below) such that each $\mathsf{Toffoli}$ layer is performed every $4$ logical cycles. Then one Grover's search employs $8.61\cdot 10^{29}$ logical qubits and $3.28\cdot 10^{18}$ logical cycles, to a total spacetime volume of $2.82\cdot 10^{48}$ logical blocks of size $d^3$. In order to keep a logical error probability below $0.1\%$ per Grover's search, we must choose a code distance $d$ such that
\begin{align*}
    2.82\cdot 10^{48} \cdot d \cdot 0.1(100p_{\rm phy})^{(d+1)/2} \leq 0.001.
\end{align*}
With physical error $p_{\rm phy} = 10^{-5}$, the above is satisfied by $d=34$, which yields a logical error probability of $\approx 0.03\%$. Since each logical qubit requires $2d^2$ physical qubits (taking into account the ancillae required for the check operators measurements), one Grover's search employs $1.99\cdot 10^{33}$ physical qubits. With a code cycle of $100$~ns, the circuit time of one Grover's search is $3.53\cdot 10^{5}$ years.

Consider now an active-volume architecture. With $8.61\cdot 10^{29}$ logical qubits and an active volume of $4.27\cdot 10^{46}$, the total spacetime volume is $8.54\cdot 10^{46}$ logical blocks of size $d^3$, twice the active volume. The number of logical cycles is $2(4.27\cdot 10^{46})/(8.61\cdot 10^{29}) = 9.92\cdot 10^{16}$ per Grover's search, since only half the logical qubits, the workspace qubits, execute logical blocks in every logical cycle. In order to keep a logical error probability below $0.1\%$, we must choose a code distance $d$ such that
\begin{align*}
    2\cdot 4.27\cdot 10^{46} \cdot d \cdot 0.1(100p_{\rm phy})^{(d+1)/2} \leq 0.001.
\end{align*}
With physical error $p_{\rm phy} = 10^{-5}$, the above is satisfied by $d=34$, which yields a logical error probability of $\approx 0.001\%$. Since each logical qubit requires $d^2$ physical qubits, one Grover's search employs $9.95\cdot 10^{32}$ physical qubits. With a code cycle of $100$~ns, the circuit time of one Grover's search is $\approx 2,670$ years.

Due to the sequential natural of classical processing associated with surface-code-based quantum computation, the runtime of every circuit is limited by its reaction depth. Given the reaction depth of $1.64\cdot 10^{18}$ and a reaction time of $1$ $\mu$s, the Grover's search is thus reaction limited at $\approx 52,000$ years. This limits the active-volume architecture to a runtime of $52,000$ years, and not $\approx 2,670$ years.

\paragraph{Distillation protocol.} Finally, we determine the distillation protocol necessary for the computation, which is obtained from the $\mathsf{Toffoli}$-count. We require that the error probability of performing $3.09\cdot 10^{44}$ $\mathsf{Toffoli}$ gates be less than $0.1\%$, which means that each magic state $|CCZ\rangle$ must have an error rate below $3.23\cdot 10^{-48}$. For baseline architectures with the above code distance $d=34$, the distillation protocol $(15\text{-to-}1)_{\lceil d/4\rceil,\lceil d/8\rceil,\lceil d/8\rceil}^4\times (15\text{-to-}1)_{\lceil d/2\rceil,\lceil d/4\rceil,\lceil d/4\rceil }^4\times (8\text{-to-CCZ})_{d,\lceil d/2\rceil,\lceil d/2\rceil}$ outputs a magic state $|CCZ\rangle$ with error rate of $5.4\cdot 10^{-50}$ every $108$ code cycles by using $111,192$ physical qubits, which is enough for our needs. Since each $\mathsf{Toffoli}$ layer must be executed every $4d = 132$ code cycles, we require $\approx \frac{108}{132}|S|/2 = 8.54\cdot 10^{28}$ distillation factories running in parallel, which adds another $9.50\cdot 10^{33}$ physical qubits to a total of $1.15\cdot 10^{34}$ physical qubits. Regarding active-volume architectures, on the other hand, the same $(15\text{-to-}1)_{\lceil d/4\rceil,\lceil d/8\rceil,\lceil d/8\rceil}^4\times (15\text{-to-}1)_{\lceil d/2\rceil,\lceil d/4\rceil,\lceil d/4\rceil }^4\times (8\text{-to-CCZ})_{d,\lceil d/2\rceil,\lceil d/2\rceil}$ protocol with $d=34$ outputs a magic state $|CCZ\rangle$ with error rate of $5.4\cdot 10^{-50}$. The associated resources are already included in the active volume cost $C_{|CCZ\rangle}$.

\begin{table}[t]
    \small
    \centering
    \caption{Summary of required resources to perform one Grover's search with one solution in the $\mathtt{NVSieve}$ with and without LSH/LSF assuming baseline and active-volume physical architectures. Reaction limit and circuit time are measured in hours, and final time is the maximum between both. We assume a lattice dimension $D=400$, logical and magic distillation probability errors smaller than $10^{-3}$, and a Grover's search probability error of $10^{-3}$. $\mathtt{NVSieve}$ without LSH has list of centers of size $|S| = 2^{0.2352D + 0.102\log_2{D} + 2.45}$. $\mathtt{NVSieve}$ with LSH/LSF replaces $S$ with a list of candidates of size $|C| = |S|\cdot p_2^\ast$ for LSH and $|C| = |S|\cdot \mathcal{C}_D(\alpha)^2\cdot \ln(1/\varepsilon)/\mathcal{W}_D(\alpha,\alpha,\pi/3)$ for LSF, where $\varepsilon = 10^{-3}$.}
    \label{table:resources_D400_nvsieve}
    \def\arraystretch{1.4}
    \resizebox{0.8\linewidth}{!}{
    \begin{tabular}{ c | c | c | c | c | c |}
        \cline{2-6}
        & Resource / Sieve & $\mathtt{NVSieve}$ & \makecell{$\mathtt{NVSieve}$ +\\ angular LSH} & \makecell{$\mathtt{NVSieve}$ +\\ spherical LSH} & \makecell{$\mathtt{NVSieve}$ +\\ spherical LSF} \\ \cline{2-6}\cline{2-6}

        & List size & $2.15\cdot 10^{29}$ & $3.46\cdot 10^{21}$ & $2.71\cdot 10^{20}$ & $1.35\cdot 10^{15}$  \\ \cline{2-6}

        & Hashing parameter $k$ & - & $83$ & $5$ & $1$  \\ \cline{2-6}

        & Number hash tables $t$ & - & $2.28\cdot 10^{15}$ & $2.75\cdot 10^7$ & $2.84\cdot 10^{38}$  \\ \cline{2-6}

        & Filter angle $\alpha$ & - & - & - &  $\pi/3$ \\  \cline{2-6}
        
        & Logical qubits & $8.61\cdot 10^{29}$ & $1.39\cdot 10^{22}$ & $1.08\cdot 10^{21}$ & $5.42\cdot 10^{15}$ \\ \cline{2-6}

        & $\mathsf{Toffoli}$-count & $3.09\cdot 10^{44}$ & $6.32\cdot 10^{32}$ & $1.38\cdot 10^{31}$ & $1.55\cdot 10^{23}$ \\ \cline{2-6}

        & $\mathsf{Toffoli}$-width & $1.08\cdot 10^{29}$ & $1.73\cdot 10^{21}$ & $1.35\cdot 10^{20}$ & $6.77\cdot 10^{14}$ \\ \cline{2-6}

        & Active volume & $4.27\cdot 10^{46}$ & $8.73\cdot 10^{34}$ & $1.90\cdot 10^{33}$ & $2.13\cdot 10^{25}$ \\ \cline{2-6}

        & Reaction depth & $1.64\cdot 10^{18}$ & $1.99\cdot 10^{14}$ & $5.51\cdot 10^{13}$ & $1.19\cdot 10^{11}$ \\ \cline{2-6}

        & Reaction limit (hours) & $4.55\cdot 10^8$ & $5.51\cdot 10^4$ & $1.53\cdot 10^4$ & $3.31\cdot 10^1$  \\ \hhline{-=====}

        \multicolumn{1}{ | c | }{\multirow{5}{*}{\rotatebox{90}{Baseline} }} & Code distance & $34$ & $26$ & $25$ & $20$ \\ \cline{2-6}

        \multicolumn{1}{ |c | }{} & Distillation factories & $8.54\cdot 10^{28}$ & $1.40\cdot 10^{21}$ & $1.14\cdot 10^{20}$ & $5.08\cdot 10^{14}$ \\ \cline{2-6}

        \multicolumn{1}{ |c | }{} & Physical qubits & $1.15\cdot 10^{34}$ & $1.12\cdot 10^{26}$ & $8.78\cdot 10^{24}$ & $2.33\cdot 10^{19}$ \\ \cline{2-6}

        \multicolumn{1}{ |c | }{} & Circuit time (hours) & $3.10\cdot 10^9$ & $2.87\cdot 10^5$ & $7.65\cdot 10^4$ & $1.32\cdot 10^2$ \\ \cline{2-6}
        
        \multicolumn{1}{ |c | }{} & Final time (hours) & $3.10\cdot 10^9$ & $2.87\cdot 10^5$ & $7.65\cdot 10^4$ & $1.32\cdot 10^2$ \\ \hline \hline

        \multicolumn{1}{ | c | }{\multirow{4}{*}{\rotatebox{90}{Active-volume}}} & Code distance & $34$ & $26$ & $24$ & $20$ \\ \cline{2-6}

        \multicolumn{1}{ |c | }{} & Physical qubits & $9.95\cdot 10^{32}$ & $9.37\cdot 10^{24}$ & $6.24\cdot 10^{23}$ & $2.17\cdot 10^{18}$  \\ \cline{2-6}

        \multicolumn{1}{ |c | }{} & Circuit time (hours) & $2.34\cdot 10^7$ & $2.27\cdot 10^3$ & $5.86\cdot 10^2$ & $1.09\cdot 10^0$  \\ \cline{2-6}
        
        \multicolumn{1}{ |c | }{} & Final time (hours) & $4.55\cdot 10^8$ & $5.51\cdot 10^4$ & $1.53\cdot 10^4$ & $3.31\cdot 10^1$ \\ \hline 
    \end{tabular}}
\end{table}

\subsubsection{GaussSieve without LSH/LSF}
\label{sec:case_study_gausssieve}

We move on to analysing the cost of Grover's search in the $\mathtt{GaussSieve}$ without LSH/LSF (\Cref{alg:gauss_sieve}). The analysis of $\mathtt{GaussSieve}$ is harder since it is a heuristic algorithm with few proven properties. In each sieving step, there are two search loops that are called while a solution can be found. For simplicity, we consider one Grover's search with $M=1$ solution. Another heuristic parameter of the algorithm is the list size. Here we assume a sieving step with list size $|L| = 2^{0.193D + 2.325} \approx 8.70\cdot 10^{23}$ as reported by us (see \cref{sec:heuristics_gausssieve}).

\paragraph{Logical costs.} Once again we gather all the logical costs first. In each sieving step, there are two different search loops being performed. According to \Cref{table:oracle_modulos}, the phase oracle $\mathcal{O}_{\rm gauss}^{(1)}$ from the first loop requires $1$ $\mathsf{QRAM}$ call of size $|L|$, $4D-2$ $\kappa$-bit adders, and $2D$ $\kappa$-bit multipliers, while the phase oracle $\mathcal{O}_{\rm gauss}^{(2)}$ from the second loop requires $1$ $\mathsf{QRAM}$ call of size $|L|$, $D+1$ $\kappa$-bit adders, and $D$ $\kappa$-bit hybrid multipliers. The expected number of Grover iterations is $\lceil 3.1 \sqrt{|L|}\rceil$. The $\mathsf{Toffoli}$-count of the two search loops is
\begin{align*}
    \mathsf{Toffoli}\text{-count loop 1} &: \underbrace{\lceil 3.1 \sqrt{|L|}\rceil}_{\rm Iterations} \big(\underbrace{|L| - 2}_{\mathsf{QRAM}} + \underbrace{(4D-2)(\kappa-1)}_{4D-2~{\rm adders}} + \underbrace{2D(\kappa^2 - \kappa + 1)}_{2D ~{\rm multipliers}} + \underbrace{\lceil\log_2|L|\rceil - 1}_{\rm Difussion~operator}\big), \\
    \mathsf{Toffoli}\text{-count loop 2} &: \underbrace{\lceil 3.1 \sqrt{|L|}\rceil}_{\rm Iterations} \big(\underbrace{|L| - 2}_{\mathsf{QRAM}} + \underbrace{(D+1)(\kappa-1)}_{5D-1~{\rm adders}} + \underbrace{D(0.5\kappa^2 - 1.5\kappa + 1)}_{D ~{\rm multipliers}} + \underbrace{\lceil\log_2|L|\rceil - 1}_{\rm Difussion~operator}\big),
\end{align*}
both approximately equal to $2.51\cdot 10^{36}$.

Regarding the number of logical qubits, the first search loop requires (already taking overlaps into account) $2|L| + \kappa D - 1$ qubits for the $\mathsf{QRAM}$, $D\kappa$ qubits after copying $|\mathbf{w}_i\rangle$ once, $4D\kappa$ qubits for the parallel $2D$ $\kappa$-bit adders, $2D(2\kappa^2 - \kappa)$ qubits for the parallel $2D$ $\kappa$-bit multipliers, and finally $2\kappa(D-1)$ qubits for the final $2D-2$ $\kappa$-bit adders. The second search loop requires (already taking overlaps into account) $2|L| + \kappa D - 1$ qubits for the $\mathsf{QRAM}$, $D(1.5\kappa^2 - 0.5\kappa)$ qubits for the $D$ $\kappa$-bit hybrid multipliers, $(D-1)\kappa$ qubits for the $D-1$ $\kappa$-bit adders, and finally $3\kappa$ qubits for the last $2$ $\kappa$-bit adders. It is not hard to see that the first search loop employs the most logical qubits, which is the final count:
\begin{align*}
    \text{Logical qubits}&: 2\big(\underbrace{2|L| + D\kappa  - 1}_{\mathsf{QRAM}} + \underbrace{D\kappa}_{\rm Copying} + \underbrace{4D\kappa}_{2D~\text{adders}} + \underbrace{2D(2\kappa^2 - \kappa)}_{2D~\text{multipliers}} + \underbrace{2(D-1)\kappa}_{2D - 2 ~\text{adders}} \big) \approx 3.48\cdot 10^{24}.
\end{align*}

The active volume of both search loops is simply the sum of the individual active volumes,
\begin{align*}
    \text{Active volume loop 1} &: \underbrace{\lceil 3.1 \sqrt{|L|}\rceil}_{\rm Iterations}\big( \underbrace{(25 + 1.5\kappa + C_{|CCZ\rangle})|L|}_{\mathsf{QRAM}} + \underbrace{(\lceil\log_2|L|\rceil - 1)(18+C_{|CCZ\rangle})}_{\rm Diffusion~operator}\\
    &+ \underbrace{(4D-2)((\kappa-1)(39 + C_{|CCZ\rangle}) + 7)}_{\rm Adders} + \underbrace{4(2D\kappa + 4)}_{\rm Extra~\mathsf{CNOT}s} \\
    &+ \underbrace{2D(28\kappa^2 - 42\kappa + 28 + (\kappa^2 - \kappa + 1)C_{|CCZ\rangle})}_{\rm Multipliers} \big),\\
    \text{Active volume loop 2} &: \underbrace{\lceil 3.1 \sqrt{|L|}\rceil}_{\rm Iterations}\big( \underbrace{(25 + 1.5\kappa + C_{|CCZ\rangle})|L|}_{\mathsf{QRAM}} + \underbrace{(\lceil\log_2|L|\rceil - 1)(18+C_{|CCZ\rangle})}_{\rm Diffusion~operator}\\
    &+ \underbrace{(D+1)((\kappa-1)(39 + C_{|CCZ\rangle}) + 7)}_{\rm Adders}\\
    &+ \underbrace{D(20.25\kappa^2 - 48.75\kappa + 32 + (0.5\kappa^2 - 1.5\kappa + 1)C_{|CCZ\rangle})}_{\rm Multipliers} \big),
\end{align*}
both approximately equal to $3.47\cdot 10^{38}$, while the reaction depth of one Grover's search in each search loop is
\begin{align*}
    \text{Reaction depth} : \underbrace{\lceil 3.1 \sqrt{|L|}\rceil}_{\rm Iterations}&\big(\underbrace{2\lceil\log_2|L|\rceil - 2}_{\mathsf{QRAM}} + \underbrace{2(1+\lceil\log_2D\rceil)(\kappa-1)}_{\rm Adders} + \underbrace{2\kappa\log_2\kappa - 2\kappa - 2\log_2\kappa + 4}_{\rm Multipliers} \\
    &+ \underbrace{2\lceil\log_2\lceil\log_2|L|\rceil\rceil}_{\rm Diffusion~operator} \big) \approx 3.01\cdot 10^{15}.
\end{align*}

\paragraph{Code distance and time.} The analysis is the same to the $\mathtt{NVSieve}$ case. First consider a baseline architecture and the Grover's search with $Q = \lceil 3.1 \sqrt{|L|}\rceil$ iterations from the first search loop. Assuming enough distillation factories, each $\mathsf{Toffoli}$ layer is performed every $4$ logical cycles. The Grover's search employs a total of $3.48\cdot 10^{24}$ logical qubits and $6.03\cdot 10^{15}$ logical cycles. In order to keep a logical error probability below $0.1\%$, we choose a code distance $d$ such that
\begin{align*}
    3.48\cdot 10^{24} \cdot 6.03\cdot 10^{15} \cdot d \cdot 0.1(100p_{\rm phy})^{(d+1)/2} \leq 0.001.
\end{align*}
Give $p_{\rm phy} = 10^{-5}$, the above is satisfied by $d=29$, yielding a logical error probability of $\approx 0.006\%$. With each logical qubit requiring $2d^2$ physical qubits, $5.85\cdot 10^{27}$ physical qubits are used (excluding distillation qubits). With a code cycle of $100$~ns, the Grover's circuit time is $\approx 554$ years. The same steps can be repeated for the other search loop, which we omit here. 

Now consider an active-volume architecture. The Grover's search with $Q = \lceil 3.1 \sqrt{|L|}\rceil$ iterations from the first search loop requires $3.48\cdot 10^{24}$ logical qubits and active volume of $3.47\cdot 10^{38}$, and therefore $2(3.47\cdot 10^{38})/(3.48\cdot 10^{24}) = 6.03\cdot 10^{15}$ logical cycles. The code distance $d$ is chosen so that
\begin{align*}
    2\cdot 3.47\cdot 10^{38} \cdot d \cdot 0.1(100p_{\rm phy})^{(d+1)/2} \leq 0.001.
\end{align*}
Given $p_{\rm phy} = 10^{-5}$, the above is satisfied by $d = 28$, yielding a logical error probability of $\approx 0.006\%$. With each logical qubit requiring $d^2$ physical qubits, $2.73\cdot 10^{27}$ physical qubits are required. With a code cycle of $100$~ns, the Grover's circuit time is $\approx 4.4$ years. The same steps can be repeated for the other search loop, which we omit here. 

Finally, given a reaction depth of $3.01\cdot 10^{15}$ and a reaction time of $1$~$\mu$s, the Grover's search is thus reaction limited at $\approx 96$ years. This limits the active-volume execution time to $\approx 96$ years, and not $\approx 4.4$ years.

\paragraph{Distillation protocol.} Finally, we check whether the distillation protocol $(15\text{-to-}1)^4_{\lceil d/4\rceil,\lceil d/8\rceil,\lceil d/8\rceil}\times (15\text{-to-}1)^4_{\lceil d/2\rceil,\lceil d/4\rceil,\lceil d/4\rceil}\times (8\text{-to-CCZ})_{d,\lceil d/2\rceil,\lceil d/2\rceil}$ with code distance $d=29$ outputs magic states with error probability smaller than $0.001/(6.22\cdot 10^{36}) = 1.61\cdot 10^{-40}$. Indeed, the distillation protocol outputs magic states with error rate $1.0\cdot 10^{-43}$ every $96$ code cycles using $84,308$ physical qubits. Since each $\mathsf{Toffoli}$ layer must be executed every $4d = 116$ code cycles, we require $\frac{96}{116}|L|/2 = 3.60\cdot 10^{23}$ distillation factories, which adds another $3.03\cdot 10^{28}$ physical qubits to a total of $3.62\cdot 10^{28}$ physical qubits. For active-volume architectures, the distillation cost was already computed in $C_{|CCZ\rangle}$.

\begin{table}[t]
    \small
    \centering
    \caption{Summary of required resources to perform one Grover's search with one solution in $\mathtt{GaussSieve}$ with and without LSH/LSF assuming baseline and active-volume physical architectures. Reaction limit and circuit time are measured in hours, and final time is the maximum between both. We assume a lattice dimension $D=400$, logical and magic distillation probability errors smaller than $10^{-3}$, and a Grover's search probability error of $10^{-3}$. We focus on a Grover's search from the first loop search. $\mathtt{GaussSieve}$ has list size $|L| = 2^{0.193D + 2.325}$ and list of candidates of size $|C| = |L|\cdot p_2^\ast$ for LSH and $|C| = |L|\cdot \mathcal{C}_D(\alpha)^2\cdot \ln(1/\varepsilon)/\mathcal{W}_D(\alpha,\alpha,\pi/3)$ for LSF, where $\varepsilon = 10^{-3}$.}
    \label{table:resources_D400_gausssieve}
    \def\arraystretch{1.45}
    \resizebox{0.8\linewidth}{!}{
    \begin{tabular}{ c | c | c | c | c | c |}
        \cline{2-6}
        & Resource / Sieve & $\mathtt{GaussSieve}$ & \makecell{$\mathtt{GaussSieve}$ +\\ angular LSH} & \makecell{$\mathtt{GaussSieve}$ +\\ spherical LSH} & \makecell{$\mathtt{GaussSieve}$ +\\ spherical LSF} \\ \cline{2-6} \cline{2-6} 

        & List size & $8.70\cdot 10^{23}$ & $5.00\cdot 10^{14}$ & $3.90\cdot 10^{12}$ & $5.48\cdot 10^9$ \\  \cline{2-6}

        & Hashing parameter $k$ & - & $99$ & $7$ & $1$  \\ \cline{2-6}

        & Number hash tables $t$ & - & $1.57\cdot 10^{18}$ & $5.31\cdot 10^9$ & $2.84\cdot 10^{38}$  \\ \cline{2-6}

        & Filter angle $\alpha$ & - & - & - &  $\pi/3$ \\  \cline{2-6}

        & Logical qubits & $3.48\cdot 10^{24}$ & $2.00\cdot 10^{15}$ & $1.56\cdot 10^{13}$ & $2.19\cdot 10^{10}$ \\ \cline{2-6}

        & $\mathsf{Toffoli}$-count & $2.51\cdot 10^{36}$ & $3.47\cdot 10^{22}$ & $2.39\cdot 10^{19}$ & $1.26\cdot 10^{15}$ \\ \cline{2-6}

        & $\mathsf{Toffoli}$-width & $4.35\cdot 10^{23}$ & $2.50\cdot 10^{14}$ & $1.95\cdot 10^{12}$ & $2.74\cdot 10^9$ \\ \cline{2-6}

        & Active volume & $3.47\cdot 10^{38}$ & $4.79\cdot 10^{24}$ & $3.30\cdot 10^{21}$ & $1.73\cdot 10^{17}$ \\ \cline{2-6}

        & Reaction depth & $3.01\cdot 10^{15}$ & $6.78\cdot 10^{10}$ & $5.90\cdot 10^{9}$ & $2.17\cdot 10^8$ \\ \cline{2-6}

        & Reaction limit (hours) & $8.37\cdot 10^5$ & $1.88\cdot 10^1$ & $1.64\cdot 10^0$ & $6.03\cdot 10^{-2}$ \\ \hhline{-=====}

        \multicolumn{1}{ | c | }{\multirow{5}{*}{\rotatebox{90}{Baseline} }} & Code distance & $29$ & $19$ & $17$ & $14$ \\ \cline{2-6}

        \multicolumn{1}{ |c | }{} & Distillation factories & $3.60\cdot 10^{23}$ & $1.97\cdot 10^{14}$ & $1.72\cdot 10^{12}$ & $2.35\cdot 10^9$ \\ \cline{2-6}

        \multicolumn{1}{ |c | }{} & Physical qubits & $3.62\cdot 10^{28}$ & $8.65\cdot 10^{18}$ & $6.47\cdot 10^{16}$ & $5.51\cdot 10^{13}$ \\ \cline{2-6}

        \multicolumn{1}{ |c | }{} & Circuit time (hours) & $4.85\cdot 10^6$ & $7.16\cdot 10^1$ & $5.57\cdot 10^0$ & $1.69\cdot 10^{-1}$ \\ \cline{2-6}
        
        \multicolumn{1}{ |c | }{} & Final time (hours) & $4.85\cdot 10^6$ & $7.16\cdot 10^1$ & $5.57\cdot 10^0$ & $1.69\cdot 10^{-1}$ \\ \hline \hline

        \multicolumn{1}{ | c | }{\multirow{4}{*}{\rotatebox{90}{Active-volume}}} & Code distance & $28$ & $18$ & $16$ & $14$ \\ \cline{2-6}

        \multicolumn{1}{ |c | }{} & Physical qubits & $2.73\cdot 10^{27}$ & $6.48\cdot 10^{17}$ & $3.99\cdot 10^{15}$ & $4.29\cdot 10^{12}$ \\ \cline{2-6}

        \multicolumn{1}{ |c | }{} & Circuit time (hours) & $3.88\cdot 10^4$ & $5.98\cdot 10^{-1}$ & $4.69\cdot 10^{-2}$ & $1.54\cdot 10^{-3}$ \\ \cline{2-6}
        
        \multicolumn{1}{ |c | }{} & Final time (hours) & $8.37\cdot 10^5$ & $1.88\cdot 10^1$ & $1.64\cdot 10^0$ & $6.03\cdot 10^{-2}$ \\ \hline 

    \end{tabular}}
\end{table}

\subsubsection{NVSieve and GaussSieve with LSH/LSF}

Finally, we consider both $\mathtt{NVSieve}$ and $\mathtt{GaussSieve}$ with LSH/LSF. Once again, we focus on one Grover's search with worst-case list size $|S| = 2^{0.2352D + 0.102\log_2{D} + 2.45} \approx 2.15\cdot 10^{29}$ for $\mathtt{NVSieve}$ and $|L| = 2^{0.193D + 2.325} \approx 8.70\cdot 10^{23}$ for $\mathtt{GaussSieve}$. We assume the existence of only one solution in each Grover's search, and we focus on the first search loop in $\mathtt{GaussSieve}$. The average size of the list of candidates $C$ to be searched over with Grover's algorithm is $|C| = |L|\cdot p_2^\ast$ for LSH and $|C| = |L|\cdot t\cdot \mathcal{C}_D(\alpha)^2$ for LSF. 

The choice for the hashing parameter $k$ and the number of hash tables $t$ (and the angle $\alpha$ for LSF) is highly heuristic. For LSH the choice of $k$ is usually based on guaranteeing that nearby vectors collide with high probability in at least one hash table. This yields $k = \log_{3/2}{t} - \log_{3/2}\ln(1/\varepsilon)$ for angular LSH and $k=6(\ln{t} - \ln\ln(1/\varepsilon))/\sqrt{D}$ for spherical LSH. For spherical LSF, $k=1$ and $t = \ln(1/\varepsilon)/\mathcal{W}_D(\alpha,\alpha,\pi/3)$. Here $\varepsilon = 10^{-3}$. On the other hand, the value of $t$ for LSH is based on balancing the classical hashing time with the quantum searching time, while the parameter $\alpha$ is obtained by minimising the total runtime (classical hashing time plus quantum searching time). A precise choice for $t$ and $\alpha$ thus depends on all sieving steps and not just a single Grover's search. We refer the reader to \cref{sec:heuristics} below for a list of assumptions on the performance of $\mathtt{NVSieve}$ and $\mathtt{GaussSieve}$ that allows for precise expressions used to derive $t$ and $\alpha$. For now we just quote the values $k$, $t$, and $\alpha$ in \Cref{table:resources_D400_nvsieve,table:resources_D400_gausssieve}.

The analysis is very similar to the previous ones, so we omit most of the details and list the results in \Cref{table:resources_D400_nvsieve,table:resources_D400_gausssieve}. The expressions for $\mathsf{Toffoli}$-count, number of logical qubits, active volume, and reaction depth are basically the aforementioned ones but replacing $|L|$ or $|S|$ with $|C|$ within a Grover's search.

\Cref{table:resources_D400_nvsieve,table:resources_D400_gausssieve} show a rough estimate for one Grover's search with worst-case list size in each $\mathtt{NVSieve}$ and $\mathtt{GaussSieve}$ with and without LSH in dimension $D=400$. Even though only one Grover's search was taken into consideration, we can already grasp the order of magnitude of each resource, specially number of physical qubits and overall time, the most important ones. Moreover, it is possible to observe some of the advantages and disadvantages of each algorithm, e.g., the use of hashing has a significant impact on time and number of physical qubits as expected from searching a smaller list. However, a full and complete analysis can only come from considering all Grover's searches from all sieving steps, which we shall look at next. 

\subsection{Resource estimations via heuristic assumptions}
\label{sec:heuristics}

In this section, we employ the analysis procedure outlined above in order to gauge the required resources to fully carry out the $\mathtt{NVSieve}$ and $\mathtt{GaussSieve}$ aided by Grover's search. For the sake of comparison, we also consider a completely classical implementation where vector reductions are searched sequentially. Since these sieving algorithms involve several quantities which are difficult to precisely measure, we rely on heuristic and numerical observations from \cref{sec:heuristics_nvsieve,sec:heuristics_gausssieve} to build plausible worst-case assumptions on which the resource estimations can be performed. In the following, we assume that:
\begin{enumerate}
    \item The initial list size $|L|$ in $\mathtt{NVSieve}$ is $|L| = D\cdot 2^{0.2352D + 0.102\log_2{D} + 2.45}$.

    \item In the classical implementation of $\mathtt{NVSieve}$, the list $S$ or $C$ is scanned one full time in order to find a solution.
    
    \item In the quantum implementation of $\mathtt{NVSieve}$, there is only one solution to each Grover's search.
    
    \item The center list has size $|S| = 2^{0.2352D + 0.102\log_2{D} + 2.45}$ in each sieving step of $\mathtt{NVSieve}$ without LSH/LSF. The list size $|L|$ decreases by $|S|$ per sieving step.
    
    \item In each sieving step of $\mathtt{NVSieve}$ with LSH, $|S| = 2^{0.2352D + 0.102\log_2{D} + 2.45}$ vectors are inserted into $t$ hash tables, and the list of candidates has size $|C| = |S|\cdot p_2^\ast$, where $p_2^\ast$ is the average probability that far-away vectors collide. In $\mathtt{NVSieve}$ with LSF, $|S| = 2^{0.2352D + 0.102\log_2{D} + 2.45}$ vectors are  inserted into relevant filters out of $t$ buckets, and the list of candidates has size $|C| = |S|\cdot t\cdot \mathcal{C}_D(\alpha)^2 = |S|\cdot \mathcal{C}_D(\alpha)^2\cdot \ln(1/\varepsilon)/\mathcal{W}_D(\alpha,\alpha,\pi/3)$, where $\varepsilon = 10^{-3}$. The list size $|L|$ decreases by $|S|$ per sieving step.
    
    \item The maximum list size in $\mathtt{GaussSieve}$ is $2^{0.193D + 2.325}$, while the number of iterations $I$ grows as $2^{0.283D + 0.335}$.
    
    \item The list size $|L|$ in $\mathtt{GaussSieve}$ equals the maximum list size of $2^{0.193D + 2.325}$ for all iterations and its size therefore does not decrease.

    \item In the classical implementation of $\mathtt{GaussSieve}$, the list $L$ or $C$ in the first search loop (\cref{line:gauss_sieve_search1} in \cref{alg:gauss_sieve} and \cref{line:gauss_sieve_search1_lsh} in \cref{alg:gauss_sieve_lsh}) is scanned $10$ times: one vector reduction happens after every scan until no solutions are left after the $9$-th time. The list $L$ or $C$ in the second search loop (\cref{line:gauss_sieve_search2} in \cref{alg:gauss_sieve} and \cref{line:gauss_sieve_search2_lsh} in \cref{alg:gauss_sieve_lsh}) is scanned only once.
    
    \item In the quantum implementation of $\mathtt{GaussSieve}$, the first search loop (\cref{line:gauss_sieve_search1} in \cref{alg:gauss_sieve} and \cref{line:gauss_sieve_search1_lsh} in \cref{alg:gauss_sieve_lsh}) is performed $10$ times: $9$ times with $M=1$ solution, and $1$ final time with $M=0$ solutions. The second search loop (\cref{line:gauss_sieve_search2} in \cref{alg:gauss_sieve} and \cref{line:gauss_sieve_search2_lsh} in \cref{alg:gauss_sieve_lsh}) is performed only once with $M=0$ solutions.
    
    \item In $\mathtt{GaussSieve}$ with LSH, $|L| = 2^{0.193D + 2.325}$ vectors are inserted into $t$ hash tables and the list of candidates has size $|C| = |L|\cdot p_2^\ast$, where $p_2^\ast$ is the average probability that far-away vectors collide. In $\mathtt{GaussSieve}$ with LSF, $|L| = 2^{0.193D + 2.325}$ vectors are inserted into relevant filters out of $t$ buckets and the list of candidates has size $|C| = |L|\cdot t\cdot \mathcal{C}_D(\alpha)^2$.

    \item Angular LSH: $k = \log_{3/2}{t} - \log_{3/2}\ln(1/\varepsilon)$ and the average collision probability of far-away vectors $p_2^\ast$ is given by \cref{eq:probability_collision_angular_hash}. The total classical hashing time requires $2k\cdot t \cdot |L|$ multiplications and $k\cdot t\cdot |L|$ additions. In the quantum implementations, the number of hash tables is determined through the equality $D^2\cdot |L|\sqrt{|S|\cdot p_2^\ast} = k\cdot t\cdot |L|$ for $\mathtt{NVSieve}$ (there are $|L|/|S| = D$ sieving steps) and $D\cdot I\sqrt{|L|\cdot p_2^\ast} = k\cdot t\cdot |L|$ for $\mathtt{GaussSieve}$. 
    
    \item Spherical LSH: $k = 6(\ln{t} - \ln\ln(1/\varepsilon))/\sqrt{D}$ and the average collision probability of far-away vectors $p_2^\ast$ is given by \cref{eq:probability_collision_spherical_hash}. The total classical hashing time requires $D\cdot 2^{\sqrt{D}}\cdot k\cdot t \cdot |L|$ additions and multiplications. In the quantum implementations, the number of hash tables is determined through the equality $D^2\cdot|L|\sqrt{|S|\cdot p_2^\ast} = D\cdot 2^{\sqrt{D}}\cdot k\cdot t\cdot |L|$ for $\mathtt{NVSieve}$ (there are $|L|/|S| = D$ sieving steps) and $D\cdot I\sqrt{|L|\cdot p_2^\ast} = D\cdot 2^{\sqrt{D}}\cdot k\cdot t\cdot |L|$ for $\mathtt{GaussSieve}$. 
    
    \item Spherical LSF: $k=1$ and the number of filter buckets is $t = \ln(1/\varepsilon)/\mathcal{W}_D(\alpha,\alpha,\pi/3)$ with $\varepsilon = 10^{-3}$. The total classical time to place vectors into relevant filters requires $2\log_2{D}\cdot |L|\cdot t\cdot \mathcal{C}_D(\alpha)$ additions. In the quantum implementations, the parameter $\alpha\leq 1/2$ is determined by minimising $\log_2{D}\cdot |L|\cdot t\cdot \mathcal{C}_D(\alpha) + D^2\cdot |L|\sqrt{|S|\cdot t\cdot \mathcal{C}_D(\alpha)^2}$ for $\mathtt{NVSieve}$ and $\log_2{D} \cdot |L|\cdot t\cdot \mathcal{C}_D(\alpha) + D\cdot I\sqrt{|L|\cdot t\cdot \mathcal{C}_D(\alpha)^2}$ for $\mathtt{GaussSieve}$.

    \item Classical additions and multiplications require $1$ and $4$ computational cycles, respectively.
    
    \item The logical error probability (the error probability of correcting a logical qubit) and Grover's search error probability ($\delta$ in \cref{fact:grover_search_unknown}) are $\leq 10^{-3}$.
\end{enumerate}

\begin{table}[t]
    \scriptsize
    \centering
    \caption{Amount of classical arithmetic operations in the \emph{classical} implementation of $\mathtt{NVSieve}$ and $\mathtt{GaussSieve}$ with and without LSH/LSF. Here $|L|$ is the maximum list size of $\mathtt{NVSieve}$ or $\mathtt{GaussSieve}$, $I$ is the number of iterations of $\mathtt{GaussSieve}$, and $p_2^\ast$ is the average probability that a non-reducing vector collides with another vector in at least one of $t$ hash tables.}
    \label{table:classical_operations}
    \def\arraystretch{1.3}
    \begin{tabular}{ c | c | c | c | c |}
        \cline{2-5}
        & \multicolumn{2}{ c| }{Searching} &  \multicolumn{2}{ c| }{Hashing} \\ \hline
        \multicolumn{1}{ |c|  }{Sieve/Operations} & Additions & Multiplications & Additions & Multiplications \\ \hline\hline
        
        \multicolumn{1}{ |c|  }{$\mathtt{NVSieve}$} & $D\cdot|L|^2$ & $\frac{1}{2}D^2\cdot |L|^2$ & $0$ & $0$ \\ \hline

        \multicolumn{1}{ |c|  }{\makecell{$\mathtt{NVSieve}$ +\\ angular LSH}} & $D\cdot |L|^2\cdot p_2^\ast$ & $\frac{1}{2}D\cdot |L|^2\cdot p_2^\ast$ & $k\cdot t\cdot |L|$ & $2k\cdot t\cdot |L|$ \\ \hline

        \multicolumn{1}{ |c|  }{\makecell{$\mathtt{NVSieve}$ + \\ spherical LSH}} & $D\cdot |L|^2\cdot p_2^\ast$ & $\frac{1}{2}D\cdot |L|^2\cdot p_2^\ast$ & $D\cdot 2^{\sqrt{D}}\cdot k\cdot t\cdot |L|$ & $D\cdot 2^{\sqrt{D}}\cdot k\cdot t\cdot |L|$ \\ \hline

        \multicolumn{1}{ |c|  }{\makecell{$\mathtt{NVSieve}$ + \\ spherical LSF}} & $D\cdot |L|^2\cdot t\cdot \mathcal{C}_D(\alpha)^2$ & $\frac{1}{2}D\cdot |L|^2\cdot t\cdot \mathcal{C}_D(\alpha)^2$ & $2\log_2{D} \cdot |L| \cdot t\cdot \mathcal{C}_D(\alpha)$ & $0$ \\ \hline

        \multicolumn{1}{ |c|  }{$\mathtt{GaussSieve}$} & $(41D-19)\cdot I\cdot |L|$ & $21D\cdot I\cdot |L|$ & $0$ &  $0$ \\ \hline 

        \multicolumn{1}{ |c|  }{\makecell{$\mathtt{GaussSieve}$ +\\ angular LSH}} & $(41D-19)\cdot I\cdot |L|\cdot p_2^\ast$ & $21D\cdot I\cdot |L|\cdot p_2^\ast$ & $k\cdot t\cdot |L|$ &  $2k\cdot t\cdot |L|$ \\ \hline 

        \multicolumn{1}{ |c|  }{\makecell{$\mathtt{GaussSieve}$ +\\ spherical LSH}} & $(41D-19)\cdot I\cdot |L|\cdot p_2^\ast$ & $21D\cdot I\cdot |L|\cdot p_2^\ast$ & $D\cdot 2^{\sqrt{D}}\cdot k\cdot t\cdot |L|$ & $D\cdot 2^{\sqrt{D}}\cdot k\cdot t\cdot |L|$ \\ \hline

        \multicolumn{1}{ |c|  }{\makecell{$\mathtt{GaussSieve}$ +\\ spherical LSF}} & $(41D-19)\cdot I\cdot |L|\cdot t\cdot \mathcal{C}_D(\alpha)^2$ & $21D\cdot I\cdot |L|\cdot t\cdot \mathcal{C}_D(\alpha)^2$ & $2\log_2{D} \cdot |L| \cdot t\cdot \mathcal{C}_D(\alpha)$  &  $0$ \\ \hline 
    \end{tabular}
\end{table}

For convenience, under the above assumptions, in \cref{table:classical_operations} we collect all classical operations coming from hashing and searching for the classical implementation of the sieving algorithms.

\begin{figure}[t]
    \centering
    \begin{subfigure}[b]{0.49\textwidth}
        \includegraphics[width=\textwidth]{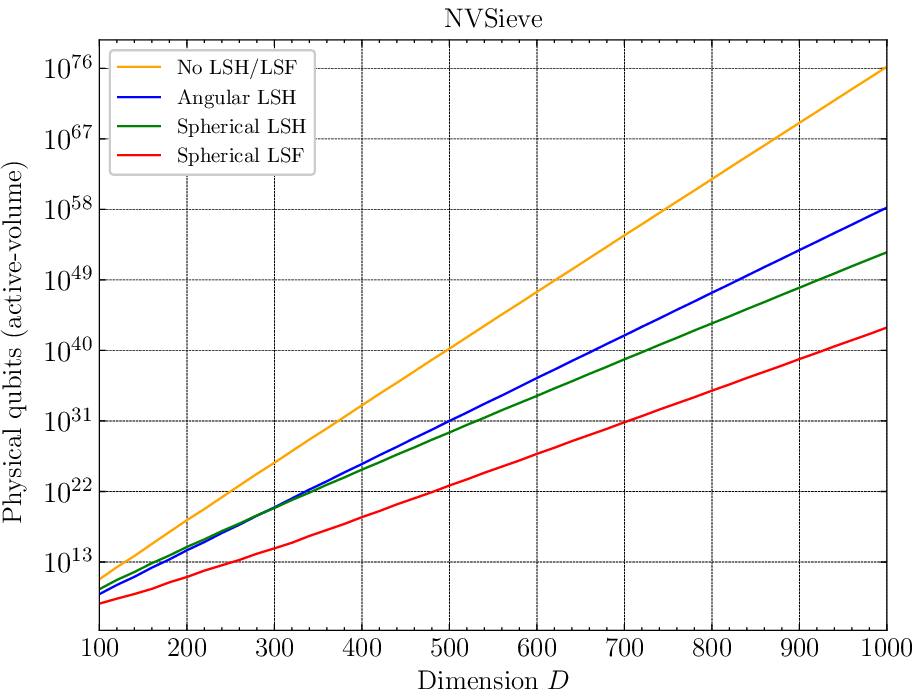}
        \caption{Active-volume physical qubits of $\mathtt{NVSieve}$}
    \end{subfigure}
    \begin{subfigure}[b]{0.49\textwidth}
        \includegraphics[width=\textwidth]{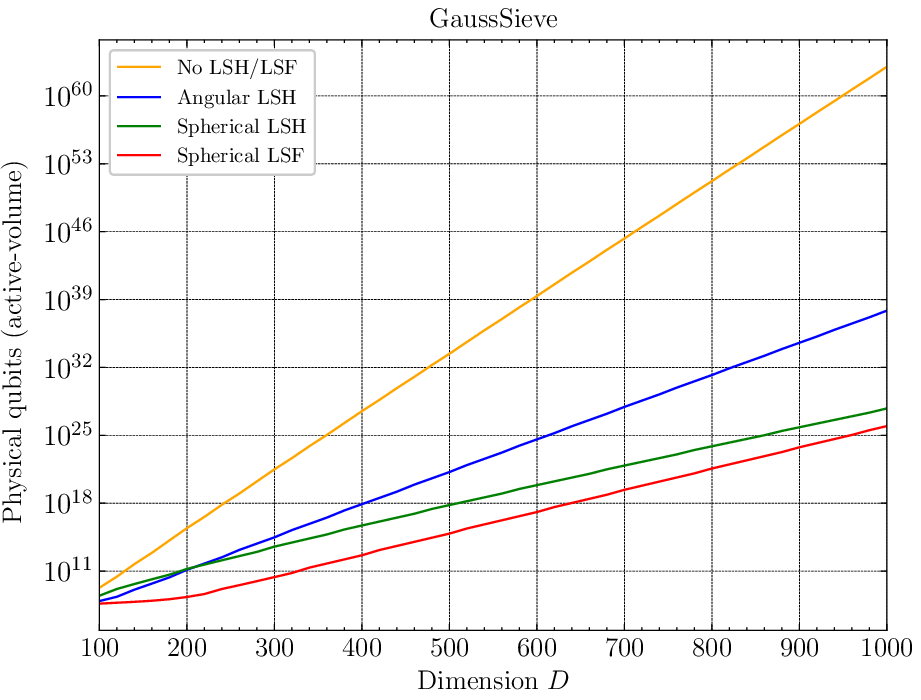}
        \caption{Active-volume physical qubits of $\mathtt{GaussSieve}$}
    \end{subfigure}
    \caption{Number of physical qubits of all Grover's searches in $\mathtt{NVSieve}$ and $\mathtt{GaussSieve}$ with and without LSH/LSF as a function of the lattice dimension $D$. We assume an underlying active-volume physical architecture. The quantities are computed based on heuristic assumptions described in the main text.}
    \label{fig:heuristic_assumptions}
\end{figure}

In \cref{fig:heuristic_assumptions,fig:heuristic_assumptions2} we compare the number of physical qubits and reaction limits from $\mathtt{NVSieve}$ and $\mathtt{GaussSieve}$ with and without LSH/LSF under an active-volume architecture. The estimated classical execution times using a $6$-GHz-clock-speed single-core classical computer are also included, where GHz means $10^9$ operations per second. For completeness, we also add the classical hashing time to the reaction limits coming from the Grover's search, although the difference is tiny. The use of locality-sensitive techniques greatly improves both quantities, specially the amount of physical qubits. It is noticeable the decrease in resources as more sophisticated hashing techniques are employed, from angular LSH to spherical LSH and LSF. Spherical LSH is more expensive than angular LSH in lower dimensions due to high lower-order terms, specially coming from the $O(2^{\sqrt{D}})$ hashing time. It is, however, asymptotically better than angular LSH as expected. At the range of proposed cryptographic dimensions $D \approx 400$, the best attack ($\mathtt{GaussSieve}$ with spherical LSF) requires around $\approx 10^{13}$ physical qubits and $\approx 10^{31}$ years to find a lattice's shortest vector. We note that most crossovers between classical and quantum time complexities happen after dimension $200$, or dimension $300$ for $\mathtt{GaussSieve}$ specifically.

\begin{figure}[p]
    \centering
    \begin{subfigure}[b]{0.81\textwidth}
        \includegraphics[width=\textwidth]{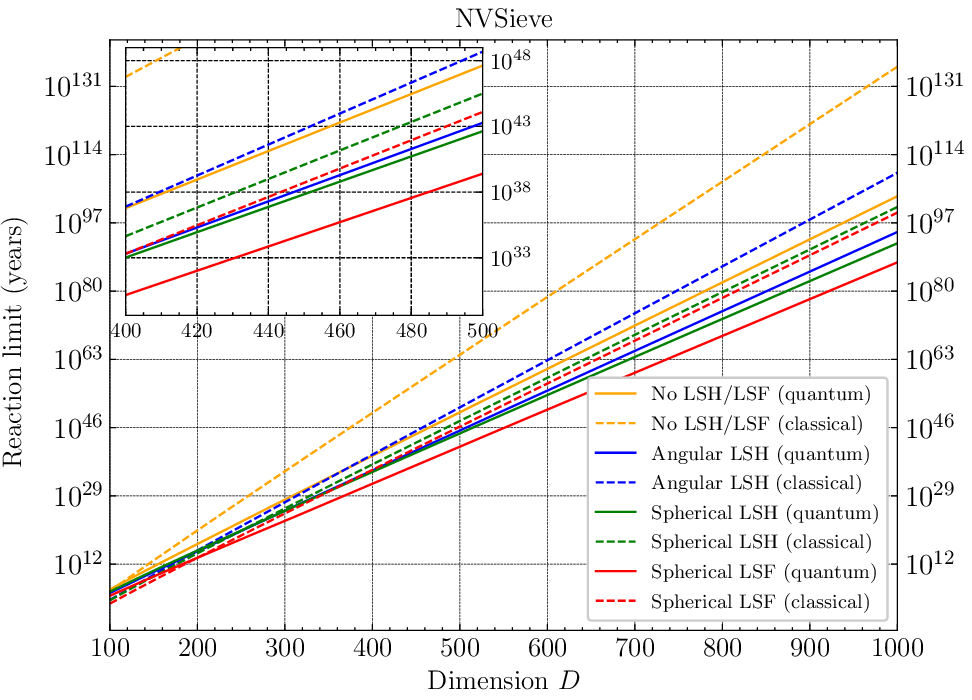}
        \caption{Reaction limit of $\mathtt{NVSieve}$}
    \end{subfigure}
    \begin{subfigure}[b]{0.81\textwidth}
        \includegraphics[width=\textwidth]{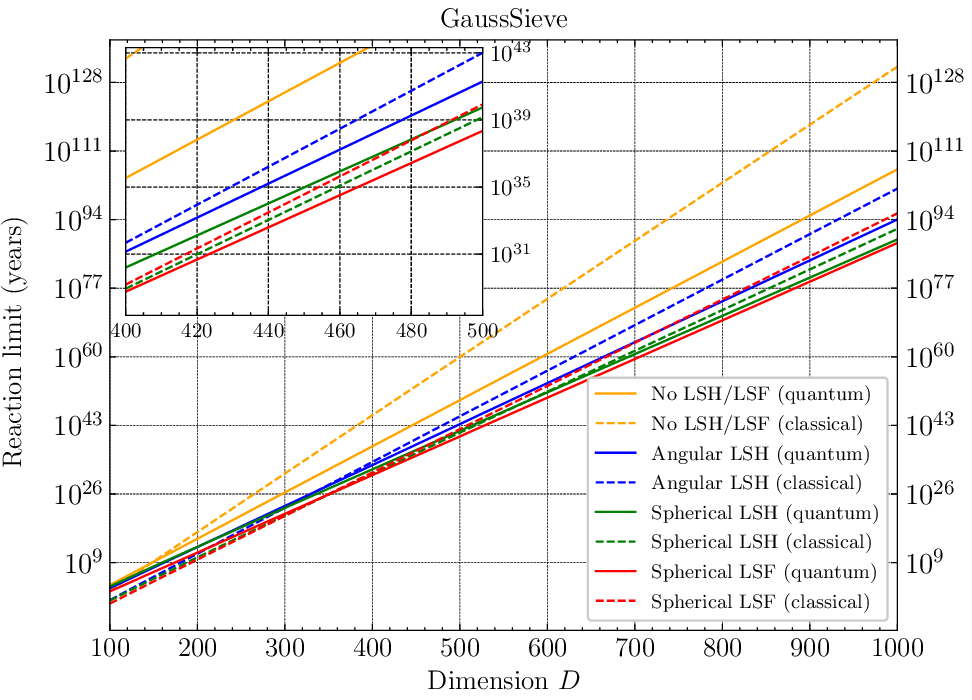}
        \caption{Reaction limit of $\mathtt{GaussSieve}$}
    \end{subfigure}
    \caption{Reaction limit of all Grover's searches in $\mathtt{NVSieve}$ and $\mathtt{GaussSieve}$ with and without LSH/LSF as a function of the lattice dimension $D$. We assume an underlying active-volume physical architecture. The reaction limits also include the classical hashing times. The quantities are computed based on heuristic assumptions described in the main text.}
    \label{fig:heuristic_assumptions2}
\end{figure}

In \cref{app:comparison} we revisit the heuristic assumptions made in this section and compare the performance of all Grover's searches under these assumptions to the performance using data from numerical simulations on classical hardware. In other words, we perform resource estimates using the evolution of the list $L$ in a real run of $\mathtt{GaussSieve}$ on classical hardware. As a brief summary, the time complexities reported in this section can probably be reduced by half under more realistic and thorough heuristic assumptions.

\subsection{The cost of QRAM}\label{sec:qramcost}

From the previous sections, specially from the cost expressions of \cref{sec:case_study}, it should be clear that $\mathsf{QRAM}$ is the most expensive component in our quantum circuits. The need to access an exponentially large dataset imposes a huge burden on size. Not only that, but the loss of sequential access to the dataset set by Grover's search implies that, when using any hashing technique, we must first gather all candidate vectors and store them separately in order to later use $\mathsf{QRAM}$. For these reasons, we analyse in this section the required resources for sieving algorithms under the scenario where $\mathsf{QRAM}$ has negligible cost, akin to Albrecht et al.~\cite{albrecht2020estimating}. This is done by repeating the procedure from the previous sections, but this time zeroing out all contributions from $\mathsf{QRAM}$ to $\mathsf{Toffoli}$-count, logical qubits, active volume, and reaction depth in the expressions from \cref{sec:case_study}. For simplicity, we focus on $\mathtt{GaussSieve}$, as it requires less resources than $\mathtt{NVSieve}$ and performs classically better in practice. The number of physical qubits under an active-volume architecture and reaction limit for $\mathtt{GaussSieve}$ with and without $\mathsf{QRAM}$ are compared in \cref{fig:resources_without_qram}.

The absence of $\mathsf{QRAM}$ has little impact on the reaction limit of $\mathtt{GaussSieve}$ for dimensions up to $1000$. According to \cref{sec:qram,sec:case_study}, $\mathsf{QRAM}$ is a shallow circuit with reaction depth of $2\lceil\log_2|L|\rceil - 2 = 390$ for $|L| = 2^{0.193\cdot 1000 + 2.325} \approx 6.29\cdot 10^{58}$, while the arithmetic part of one Grover iteration has reaction depth of $2(1 + \lceil\log_2{D}\rceil)(\kappa-1) + 2\kappa\log_2\kappa - 2\kappa - 2\log_2\kappa + 4 = 932$, hence why there is no noticiable change in the reaction limit from \cref{fig:resources_without_qram_b}. On the other hand, however, the number of physical qubits is drastically reduced from $\approx 10^{25}$ down to $\approx 10^{9}$ for $\mathtt{GaussSieve}$ with spherical LSF at dimension $D=1000$, for example. Such drastic change is expected, since a bucket-brigade-style $\mathsf{QRAM}$ with shallow reaction depth requires a number of logical qubits roughly equal to the size of the list stored in memory. We note that earlier resources estimates on Shor's algorithm placed the number of physical qubits to be in the range $10^{8}$-$10^{10}$~\cite{van2010distributed,jones2012layered,fowler2012surface,OGorman2017quantum,gheorghiu2019benchmarking}, which is comparable to our estimates of running $\mathtt{GaussSieve}$ without $\mathsf{QRAM}$ in high dimensions. From \cref{fig:resources_without_qram_a} it can be noted that the number of physical qubits has little dependence on the employed hashing techniques. Without $\mathsf{QRAM}$, the number of physical qubits comes mainly from the arithmetic modules, which are independent of the list size. Finally, the sudden changes in the number of physical qubits from \cref{fig:resources_without_qram_a} are due to integer increases in the code distance $d$ in order to maintain the error rates below $0.1\%$.

\begin{figure}[t]
    \centering
    \begin{subfigure}[b]{0.49\textwidth}
        \includegraphics[width=\textwidth]{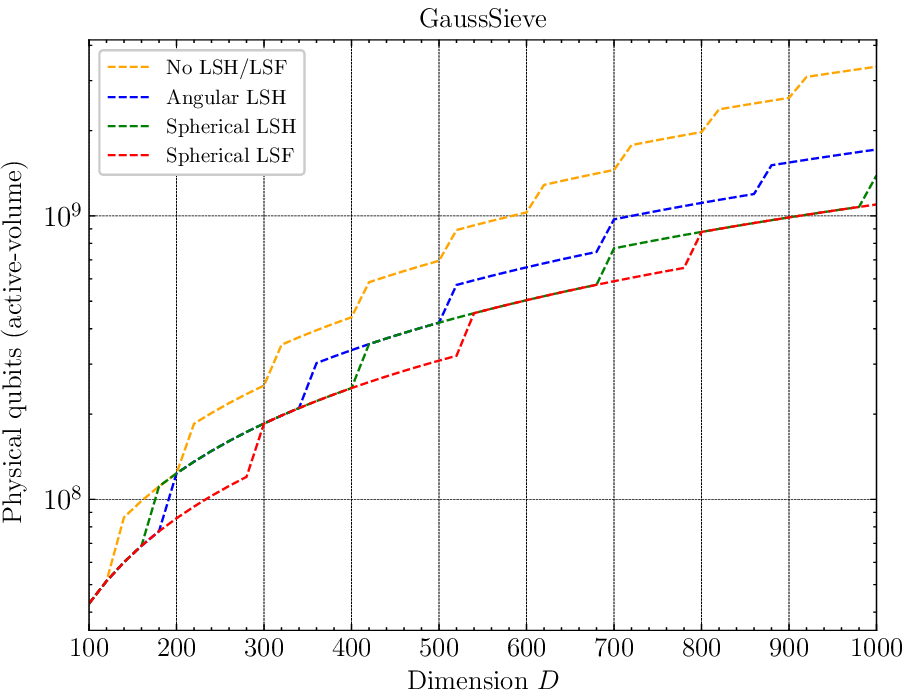}
        \caption{Physical qubits of $\mathtt{GaussSieve}$ without $\mathsf{QRAM}$}
        \label{fig:resources_without_qram_a}
    \end{subfigure}
    \begin{subfigure}[b]{0.49\textwidth}
        \includegraphics[width=\textwidth]{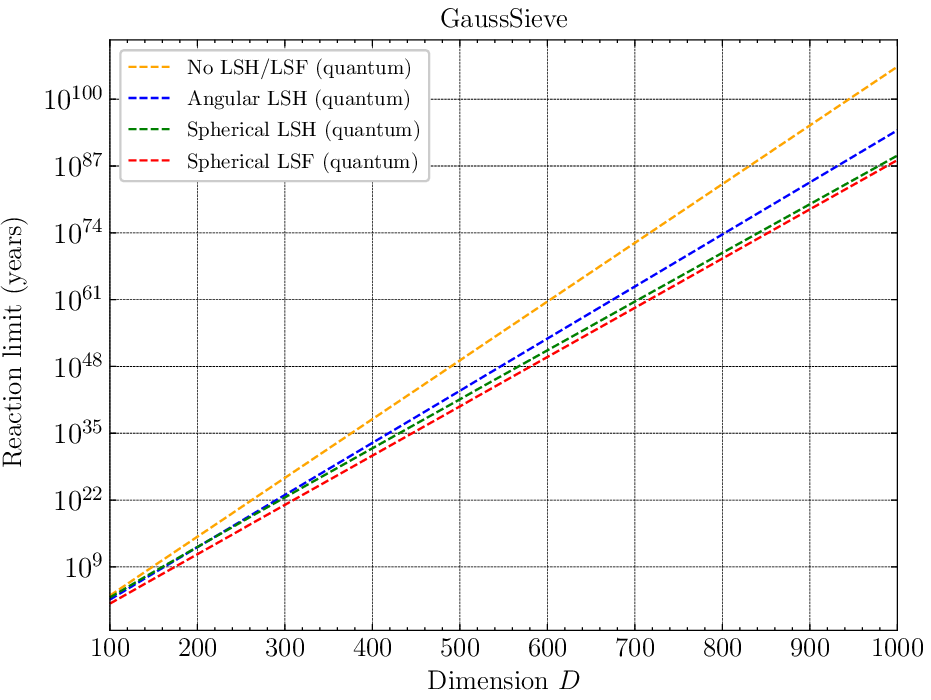}
        \caption{Reaction limit of $\mathtt{GaussSieve}$ without $\mathsf{QRAM}$}
        \label{fig:resources_without_qram_b}
    \end{subfigure}
    \caption{Number of physical qubits and reaction limit of $\mathtt{GaussSieve}$ with and without LSH/LSF as a function of the lattice dimension $D$ in the scenario where $\mathsf{QRAM}$ has negligible cost. We assume an underlying active-volume physical architecture. The quantities are computed based on heuristic assumptions described in the main text.}
    \label{fig:resources_without_qram}
\end{figure}

\subsection{Depth constraints: NIST standardisation}\label{sec:nist}

In many realistic situations, a quantum attacker would have bounded resources, e.g., be constrained by a total running time or circuit depth. In its call for proposals for the post-quantum cryptography standardisation process~\cite{nist_cfp}, NIST introduced the parameter $\mathtt{MAXDEPTH}$ to bound the circuit depth of any potential attacker, suggesting reasonably values in the range of $2^{40}$ to $2^{96}$ logical gates. As explained in their proposal~\cite[Section~4.A.5]{nist_cfp}, the value $2^{40}$ is ``the approximate number of gates that presently envisioned quantum computing architectures are expected to serially perform in a year''~\cite{jones2012layered}, while $2^{96}$ is ``the approximate number of gates that atomic scale qubits with speed of light propagation times could perform in a millennium.'' In this section, we revisit the results from \cref{sec:heuristics} and constrain the circuit depth. Since several quantities could be interpreted as the circuit depth, we set the parameter $\mathtt{MAXDEPTH}$ as a limit to the reaction depth of any Grover's search. This means that, for $\mathtt{MAXDEPTH} = 2^{40}$, any Grover's search would be time limited to $2^{40}~\mu\text{s} \approx 12.73~\text{days}$, assuming a reaction time of $1~\mu\text{s}$, while for $\mathtt{MAXDEPTH} = 2^{96}$, any Grover's search would be time limited to $2^{96}~\mu\text{s} \approx 2.51\cdot 10^{15}~\text{years}$. This, in turns, limits the number of Grover iterations. In order to meet the maximum reaction depth, we split the list $L$ in $\mathtt{GaussSieve}$ (list of centers $S$ in $\mathtt{NVSieve}$ and list of candidates $C$ when using LSH/LSF) into $F$ disjoint parts, each to be searched by a different instance of Grover algorithm. The number $F$ of sequential Grover's searches required to set a maximum reaction depth of $I$ in $\mathtt{GaussSieve}$ is thus determined by the equation
\begin{align*}
\begin{multlined}[\textwidth][b]
    I = \lceil 9.2\log_3(1/\delta) \sqrt{|L|/F}\rceil \big(2\lceil\log_2(|L|/F)\rceil  + 2(1+\lceil\log_2D\rceil)(\kappa-1) \\
    + 2\kappa\log_2\kappa - 2\kappa - 2\log_2\kappa + 2 + 2\lceil\log_2\lceil\log_2(|L|/F)\rceil\rceil \big),
\end{multlined}
\end{align*}
which simply follows from the reaction-depth expression from \cref{sec:case_study_gausssieve}. A similar equation to determining $F$ holds for $\mathtt{NVSieve}$. Here $2^{40} \leq I \leq 2^{96}$ as set by NIST. The value $F$ obtained from the above equation is then used to determine other quantities like number of physical qubits.

In \cref{fig:resources_nist} we depict the number of physical qubits and total reaction limit of $\mathtt{GaussSieve}$ with and without LSH/LSF in the scenario where each Grover's search has reaction depth at most $2^{40}$. As usual, the total number of physical qubits is the maximum number of physical qubits used by any Grover's search, while the total reaction limit is the sum of the reaction limits of all Grover's searches. For small dimensions, the reaction limit of Grover's search is smaller than $\mathtt{MAXDEPTH} = 2^{40}$, so there are no differences between \cref{fig:heuristic_assumptions,fig:resources_nist}. However, the depth restriction begins to take effect for dimensions higher than $\approx 250$. As a consequence, the number of physical qubits becomes mostly constant since only lists of at most a certain size can be searched. On the other hand, the reaction limit of the whole algorithm increases more rapidly with the dimension $D$, since employing $F$ sequential Grover's searches over list of size $|L|/F$ is less time efficient than employing a single Grover's search over the whole list $L$. The end result is a considerable decrease in number of physical qubits, while the time increases by a few orders of magnitude, specially in $\mathtt{GaussSieve}$ without LSH/LSF, whose circuit depth is capped in smaller dimensions. A similar effect would be observed for a different $\mathtt{MAXDEPTH}$. We remark that the reaction depth of Grover's search is smaller than $2^{96}$ for dimensions $D \lesssim 800$, hence why we omit an analysis for $\mathtt{MAXDEPTH} = 2^{96}$.

\begin{figure}[t]
    \centering
    \begin{subfigure}[b]{0.49\textwidth}
        \includegraphics[width=\textwidth]{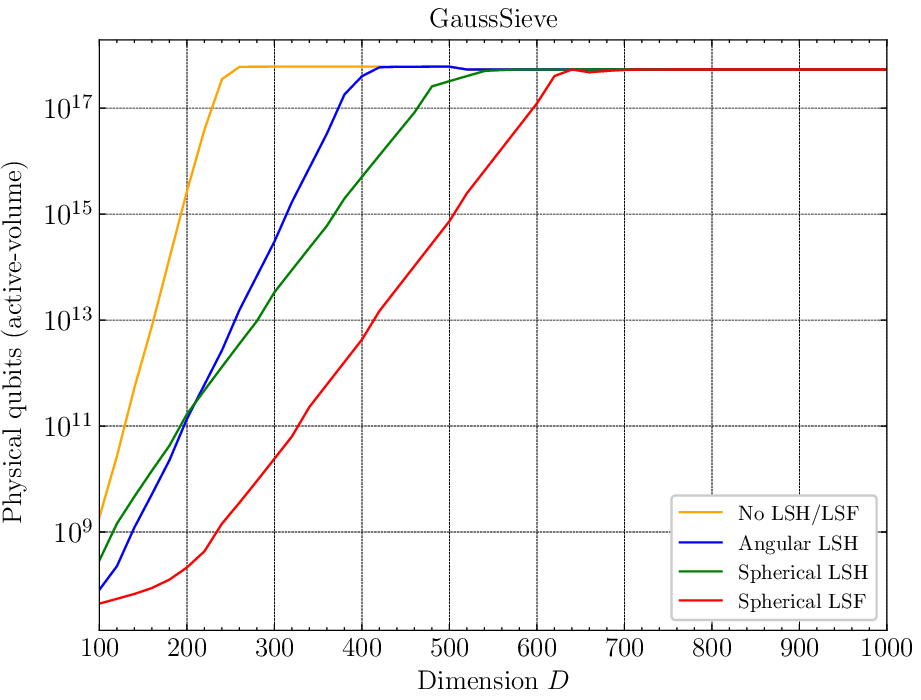}
        \caption{Physical qubits of $\mathtt{GaussSieve}$ with limited depth}
        \label{fig:resources_nist_a}
    \end{subfigure}
    \begin{subfigure}[b]{0.49\textwidth}
        \includegraphics[width=\textwidth]{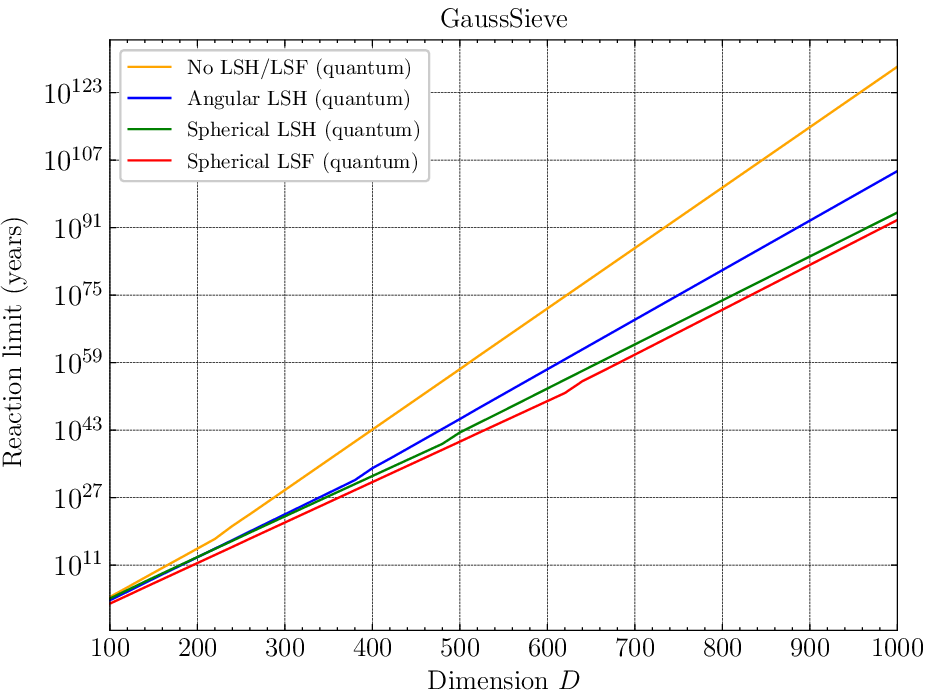}
        \caption{Reaction limit of $\mathtt{GaussSieve}$ with limited depth}
        \label{fig:resources_nist_b}
    \end{subfigure}
    \caption{Number of physical qubits and reaction limit of $\mathtt{GaussSieve}$ with and without LSH/LSF as a function of the lattice dimension $D$ in the scenario where the reaction depth of each Grover's seach is at most $2^{40}$. We assume an underlying active-volume physical architecture. The quantities are computed based on heuristic assumptions described in the main text.}
    \label{fig:resources_nist}
\end{figure}

\section{Discussions and open problems}
\label{sec:discussionandopen}

In this paper, we considered the most important sieving algorithms ($\mathtt{NVSieve}$ and $\mathtt{GaussSieve}$) and gave rigorous estimates on the time and space required to execute internal searching subroutines with Grover's search. Our estimation analysis took into consideration fixed-point quantum arithmetic, the cost of $\mathsf{QRAM}$, different physical architectures like baseline and active-volume ones, and quantum error correction. For the sake of comparison, we also consider equivalent classical implementations where Grover's search was replaced with sequential classical searching operations. We note that using BKZ to break the security of level-1 NIST candidate cryptosystems like Kyber-512, Falcon-512, and DiLithium require us to solve SVP in dimensions (block sizes of) over $400$. At this lattice dimension, even $\mathtt{GaussSieve}$ with spherical LSF under an active-volume architecture would require $\approx 10^{13}$ physical qubits and $\approx 10^{31}$ years to execute all Grover's search subroutines, which also takes into consideration classical hashing operations but ignores memory allocation. Most of the required qubits are due to $\mathsf{QRAM}$, meaning that any quantum advantage will only be possible if $\mathsf{QRAM}$ becomes substantially less costly. On the other hand, a single-core classical computer with $6$ GHs clock rate would also require $\approx 10^{31}$ years to solve SVP at dimension $400$, meaning that there is little advantage at dimensions of cryptographic interest. 

We have not explored the possibility of parallelising the list search by breaking it into smaller parts and employing different Grover's searches on each part. However, it is well known that Grover's search does not parallelise well~\cite{zalka99grover}, meaning that $F$ parallel Grover algorithms running on $F$ separate search spaces have a total width that is larger by a factor of $F$ compared to a single Grover algorithm on the whole search space while only reducing the depth by a factor of $\sqrt{F}$. We expect to observe a decrease in total runtime (\cref{fig:heuristic_assumptions}) via parallelisation by $n$ order of magnitude in exchange to an increase in number of physical qubits by roughly $2n$ orders of magnitude.

The hash parameter $k$ and the number of hash tables $t$ were chosen so that nearby vectors collide with high probability and the classical time hashing is balanced out by the quantum time searching. A very precise choice for $t$ would
require sorting out the constant factors in each of these complexities, which we did not consider. We leave it to future works a more careful analysis on the choice of $k,t,\alpha$.

We saw that the introduction of LSH or LSF requires a classical pre-search to gather all candidate vectors from the buckets with the same hash as a given vector and place them on a $\mathsf{QRAM}$. Albrecht et al.~\cite{albrecht2020estimating} (partially)\footnote{The problem is still present when considering LSF in their $\mathtt{ListDecodingSearch}$~\cite[Algorithm~4]{albrecht2020estimating}.} evaded such a problem by employing just one hash table and considering more than one bucket via a ``XOR and population count trick.'' By using their popcount filter and amplifying the amplitude of the vectors that pass such filter via quantum amplitude amplification~\cite{brassard2002quantum}, Grover's search is performed only on a subset of vectors which are close to the target vector with high probability. This lessens the cost coming from quantum arithmetic circuits in Grover's oracle. Even though Albrecht et al.\ obtained a cost expression for this ``filtered'' Grover's search, it requires strong bounds on the number of solutions $M$. In particular, the authors assumed that the number of solutions to the filtered search is known exactly beforehand, which we deem too strong of an assumption. It would be interesting to obtain more rigorous cost expressions on their filtered Grover's search (akin to Ref.~\cite{Cade2023quantifyinggrover}) and perform a complete resource estimation on sieving algorithms employing it.

Finally, we leave to future works a rigorous resource estimate on quantum-random-walk-based sieving~\cite{chailloux21,bonnetain2023finding} and enumeration algorithms and the consideration of metrics other than time and number of physical qubits like energy consumption.

\section*{Acknowledgements}
We thank Divesh Aggrawal, Martin Albrecht, Hugo Cable, Anupam Chattopadhyay, Craig Gidney, Andr\'as Gily\'en, Daniel Litinski, Markus M\"uller, and Adithya Sireesh for useful discussions. JFD acknowledges funding from ERC grant No.\ 810115-DYNASNET. Research at CQT is funded by the National Research Foundation, the Prime Minister’s Office, and the Ministry of Education, Singapore under the Research Centres of Excellence programme’s research grant R-710-000-012-135. We also acknowledge funding from the Quantum Engineering Programme (QEP 2.0) under grant NRF2021-QEP2-02-P05. This work was done in part while JFD was visiting the Simons Institute for the Theory of Computing, supported by NSF QLCI Grant No.\ 2016245.

\bibliographystyle{plain}
\bibliography{paper_sieving}

\appendix

\section{Comparison between heuristic assumptions and numerical simulations}
\label{app:comparison}

Some of the heuristic assumptions from \cref{sec:heuristics} are worst-case simplifications, e.g., the assumption that any Grover's search in $\mathtt{NVSieve}$ and $\mathtt{GaussSieve}$ has at most one solution, or that the list size is maximum throughout all iterations. In reality, we expect $\mathtt{NVSieve}$ and $\mathtt{GaussSieve}$ to perform better than described in \cref{sec:heuristics}: the list size should be quite smaller in many iterations than its maximum size at the end of the algorithm, and several solutions could exist when performing Grover's search. 

In this appendix, we compare the results of \cref{sec:heuristics} to those obtained from actual numerical simulations. To be more precise, we solved SVP on a random lattice of dimension $40 \leq D \leq 71$ using $\mathtt{GaussSieve}$ (without LSH/LSF) on a classical hardware and recorded the list $L_i$ and number of solutions $M_i$ at each step $i$. This creates a history of list and number-of-solution pairs $\{(L_i,M_i)\}_i$. Given a list size $|L_i|$ and a number of solutions $M_i$, it is then possible to estimate the amount of resources that would be required by Grover's search at that given step $i$ of $\mathtt{GaussSieve}$ by following \cref{sec:case_study}. The total amount of physical qubits employed by one particular $\mathtt{GaussSieve}$ run is the maximum number of physical qubits that would be required for any search step $i$, while the total quantum time complexity due to all Grover algorithms is the sum of the quantum time complexities of each individual search step $i$. Since $\mathtt{GaussSieve}$ is a randomised algorithm, we repeat this procedure a few times and take averages of the final number of physical qubits and quantum time complexity\footnote{For $40\leq D \leq 70$ we repeated the procedure $10$ times, while for $D = 71$ we repeated it $3$ times.}. The results are displayed in \cref{fig:resources_simulations}.

\cref{fig:resources_simulations_a} compares the number of physical qubits, both under baseline and active-volume architectures, that result from following the heuristic assumptions of \cref{sec:heuristics} and from considering the history of list and number of solutions $\{(L_i,M_i)\}_i$ of an average run of $\mathtt{GaussSieve}$. There is little difference between both approaches in the number of physical qubits, which is to be expected since the number of physical qubits is a function of the maximum list size and its heuristic value of $2^{0.193D + 2.325}$ is a fitting of actual numerical data. More interestingly, though, \cref{fig:resources_simulations_b} compares several time complexities (reaction limit and circuit time under baseline and active-volume architectures) between heuristic and numerical data. As anticipated, we can observe that $\mathtt{GaussSieve}$ has lower quantum time complexities in ``practice'' than under the heuristic assumptions of \cref{sec:heuristics}. The improvement in time complexity is around $50\%$, meaning that $\mathtt{GaussSieve}$ with Grover's search should be two times faster than reported in \cref{sec:results} for dimensions $40 \leq D \leq 71$. It is not unreasonable to extend such advantage to larger dimensions and to $\mathtt{GaussSieve}$ with hashing techniques.

\begin{figure}[t]
    \centering
    \begin{subfigure}[b]{0.49\textwidth}
        \includegraphics[width=\textwidth]{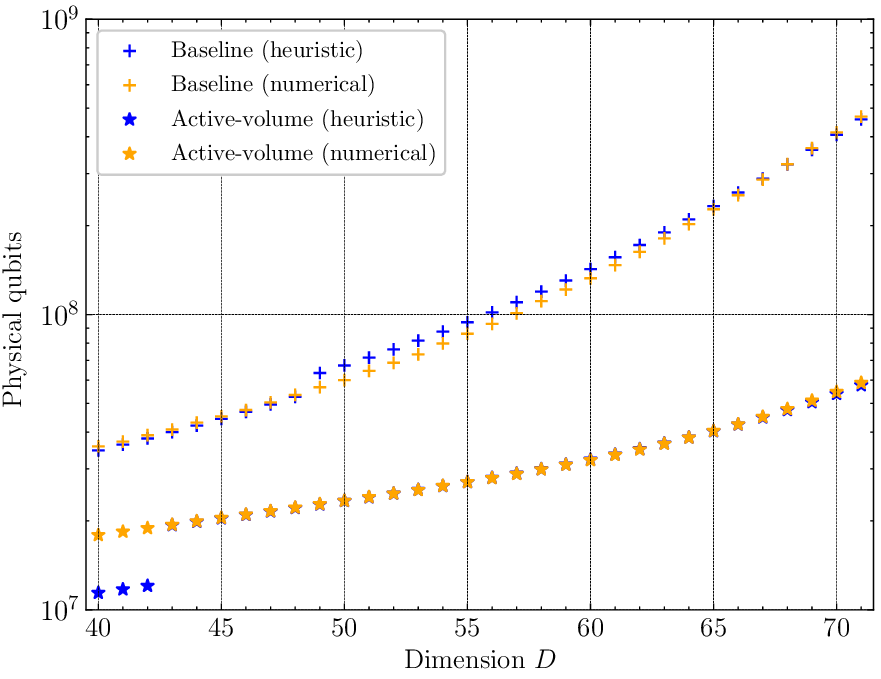}
        \caption{Physical qubits of $\mathtt{GaussSieve}$}
        \label{fig:resources_simulations_a}
    \end{subfigure}
    \begin{subfigure}[b]{0.49\textwidth}
        \includegraphics[width=\textwidth]{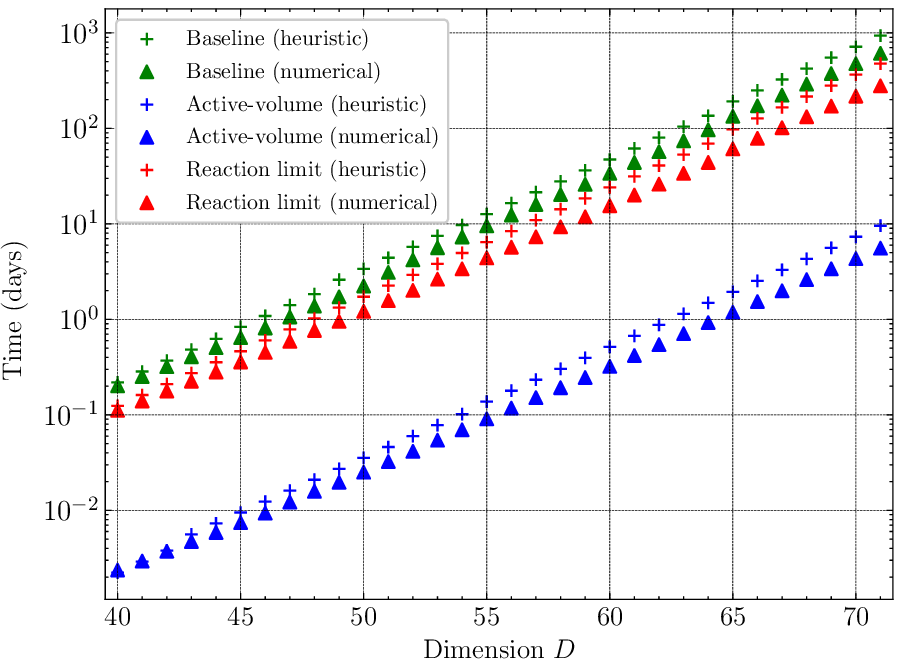}
        \caption{Time estimate of $\mathtt{GaussSieve}$}
        \label{fig:resources_simulations_b}
    \end{subfigure}
    \caption{Comparison between quantum resources of $\mathtt{GaussSieve}$ obtained through heuristic assumptions and through numerical data. (a) Comparison between physical qubits under baseline and active-volume architectures. (b) Comparison between reaction limit and circuit time under baseline and active-volume architectures. The heuristic data is obtained through the assumptions from \cref{sec:heuristics}, while the numerical data is obtained by running $\mathtt{GaussSieve}$ on a classical hardware and employing its internal parameters (list size and number of solutions) at every step.}
    \label{fig:resources_simulations}
\end{figure}

\section{Extra results}

In this appendix we provide more results that were omitted from \cref{sec:results}, e.g., number of physical qubits and circuit time under baseline and active-volume physical architectures for both $\mathtt{NVSieve}$ and $\mathtt{GaussSieve}$ with and without LSH/LSF. \cref{fig:extra_result1,fig:extra_result2,fig:extra_result3} describe results for $\mathtt{NVSieve}$ and $\mathtt{GaussSieve}$ with $\mathsf{QRAM}$, without $\mathsf{QRAM}$, and with Grover's searches reaction-depth limited to $2^{40}$, respectively.

\begin{figure}[t]
    \centering
    \begin{subfigure}[b]{0.49\textwidth}
        \includegraphics[width=\textwidth]{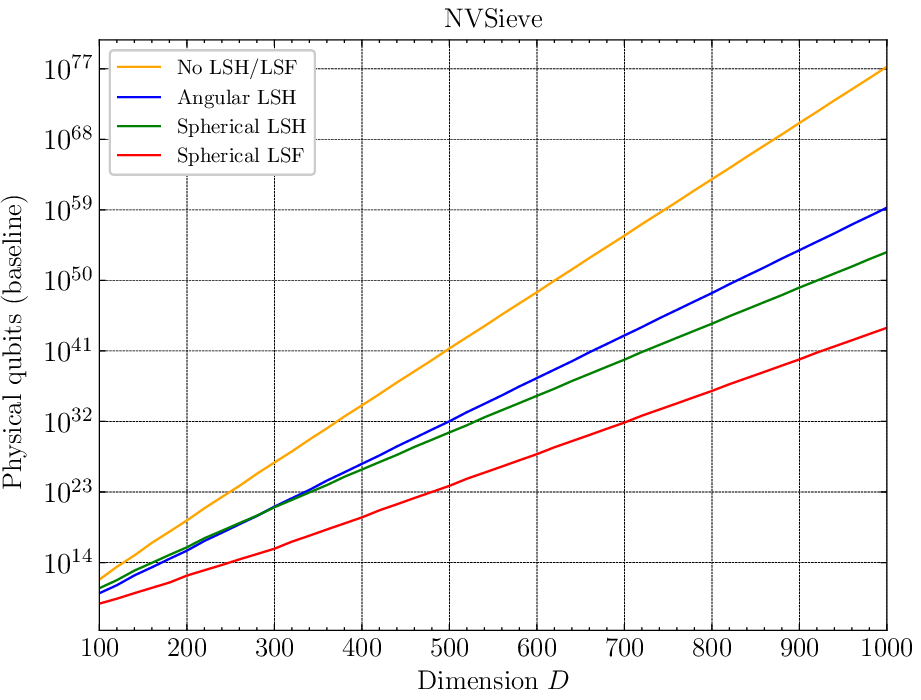}
        \caption{Baseline physical qubits of $\mathtt{NVSieve}$}
    \end{subfigure}
    \begin{subfigure}[b]{0.49\textwidth}
        \includegraphics[width=\textwidth]{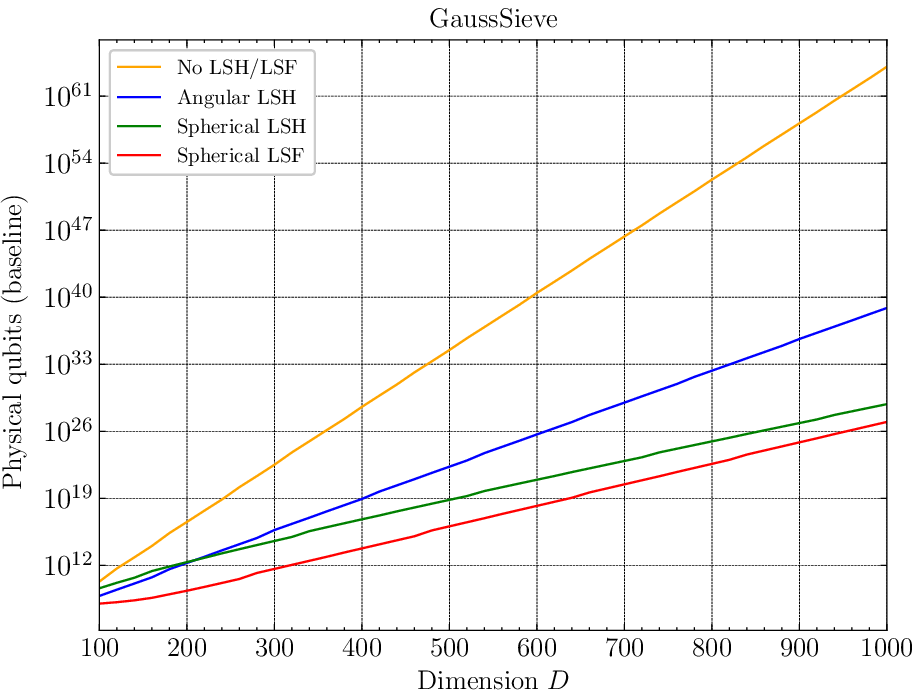}
        \caption{Baseline physical qubits of $\mathtt{GaussSieve}$}
    \end{subfigure}
    \begin{subfigure}[b]{0.49\textwidth}
        \includegraphics[width=\textwidth]{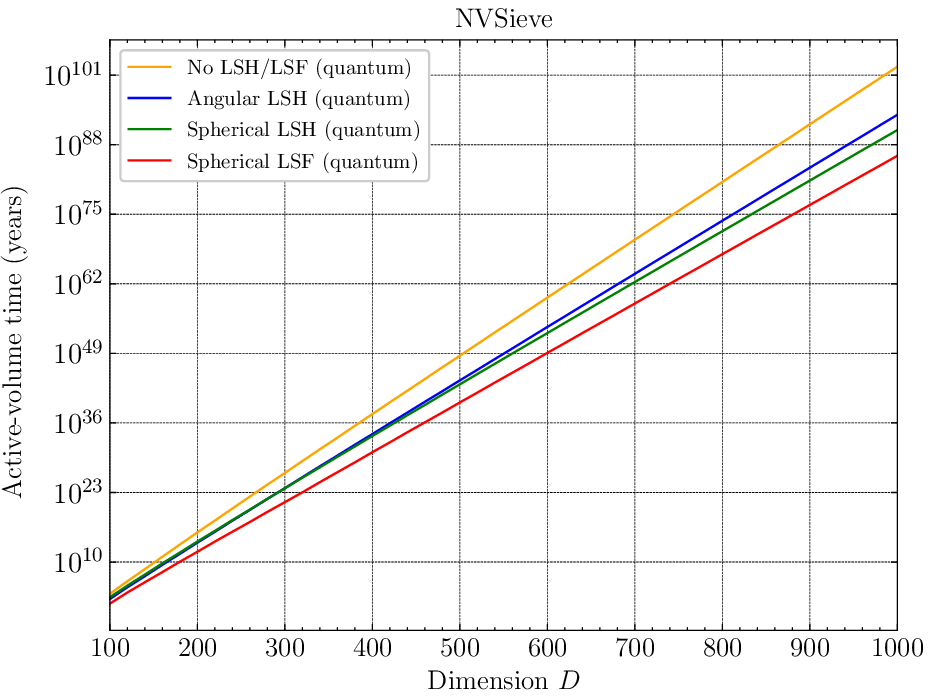}
        \caption{Active-volume circuit time of $\mathtt{NVSieve}$}
    \end{subfigure}
    \begin{subfigure}[b]{0.49\textwidth}
        \includegraphics[width=\textwidth]{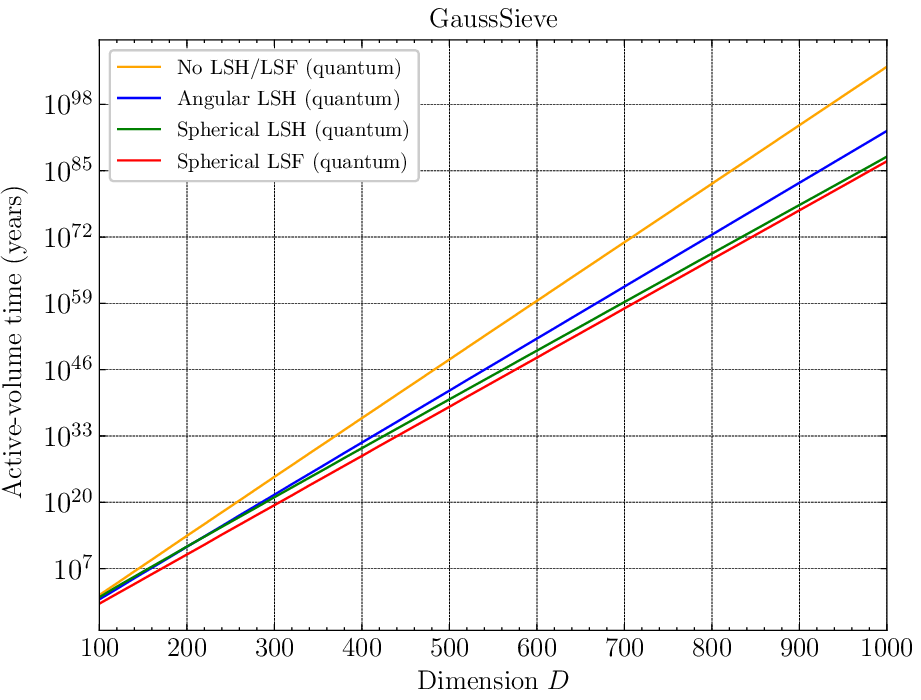}
        \caption{Active-volume circuit time of $\mathtt{GaussSieve}$}
    \end{subfigure}
    \begin{subfigure}[b]{0.49\textwidth}
        \includegraphics[width=\textwidth]{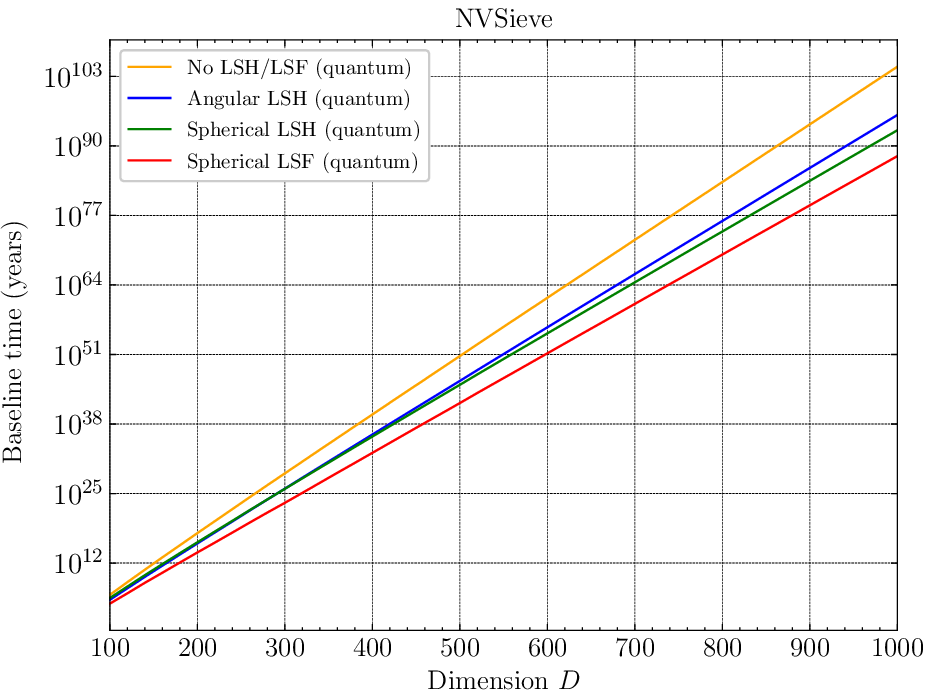}
        \caption{Baseline circuit time of $\mathtt{NVSieve}$}
    \end{subfigure}
    \begin{subfigure}[b]{0.49\textwidth}
        \includegraphics[width=\textwidth]{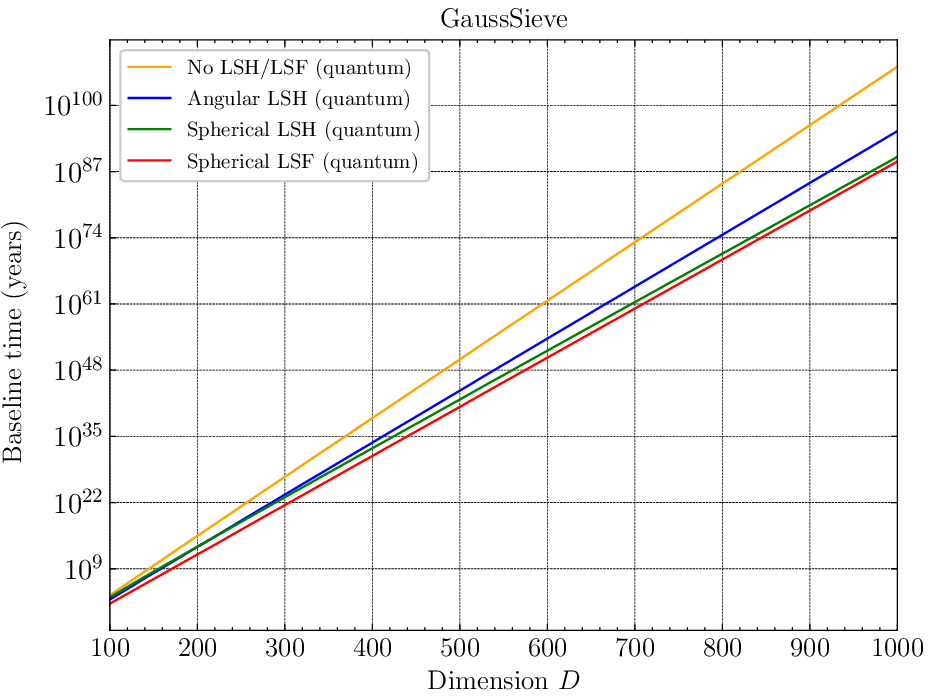}
        \caption{Baseline circuit time of $\mathtt{GaussSieve}$}
    \end{subfigure}
    \caption{Number of physical qubits and circuit times in $\mathtt{NVSieve}$ and $\mathtt{GaussSieve}$ with and without LSH/LSF under baseline and active-volume physical architectures as a function of lattice dimension~$D$. Reaction limits and circuit time also include classical hashing time.}
    \label{fig:extra_result1}
\end{figure}

\begin{figure}[t]
    \centering
    \begin{subfigure}[b]{0.49\textwidth}
        \includegraphics[width=\textwidth]{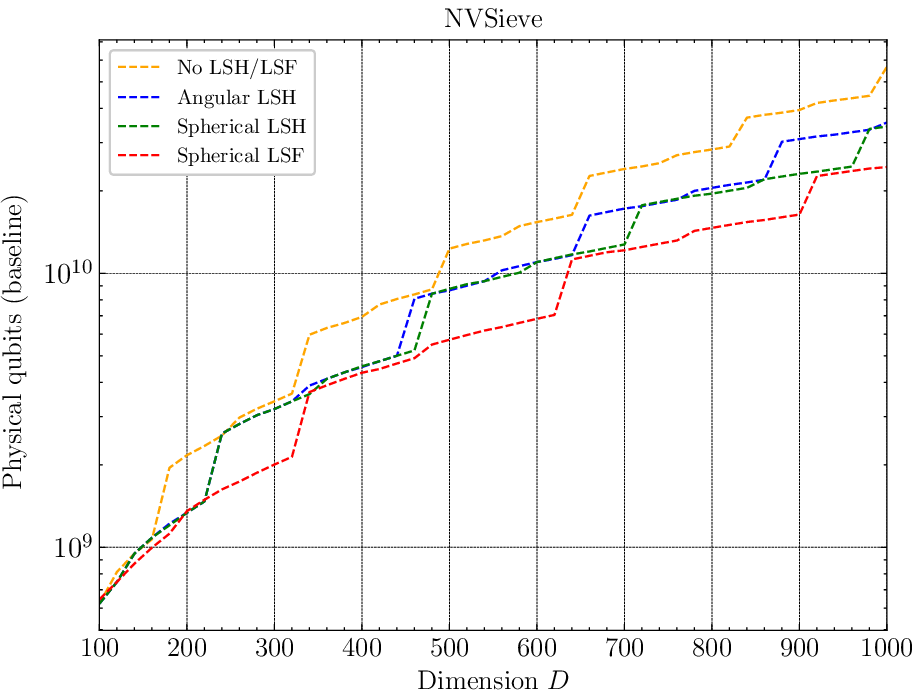}
        \caption{\centering Baseline physical qubits of $\mathtt{NVSieve}$ without $\mathsf{QRAM}$}
    \end{subfigure}
    \begin{subfigure}[b]{0.49\textwidth}
        \includegraphics[width=\textwidth]{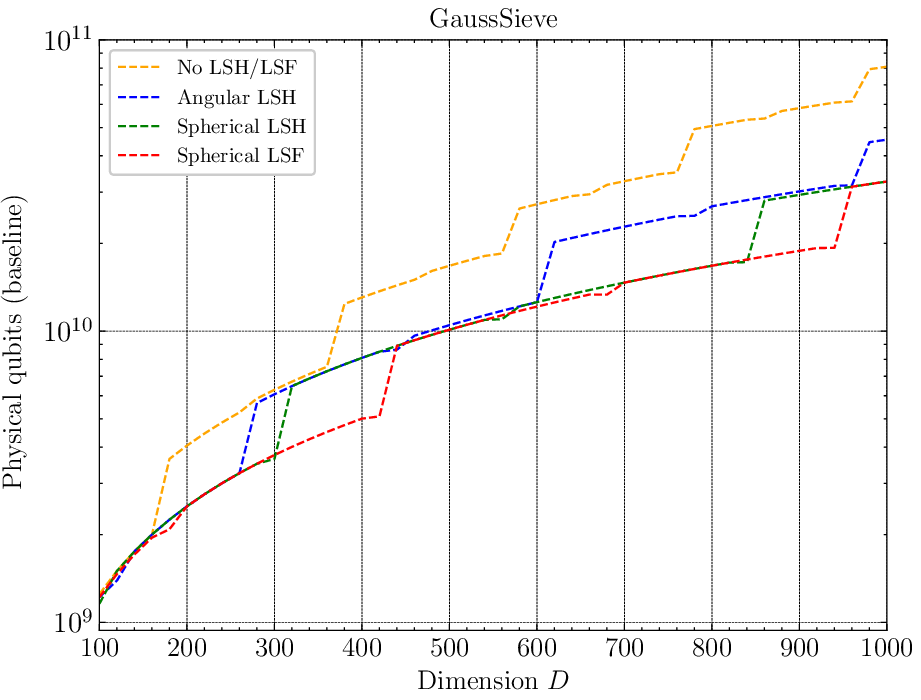}
        \caption{\centering  Baseline physical qubits of $\mathtt{GaussSieve}$ without $\mathsf{QRAM}$}
    \end{subfigure}
    \begin{subfigure}[b]{0.49\textwidth}
        \includegraphics[width=\textwidth]{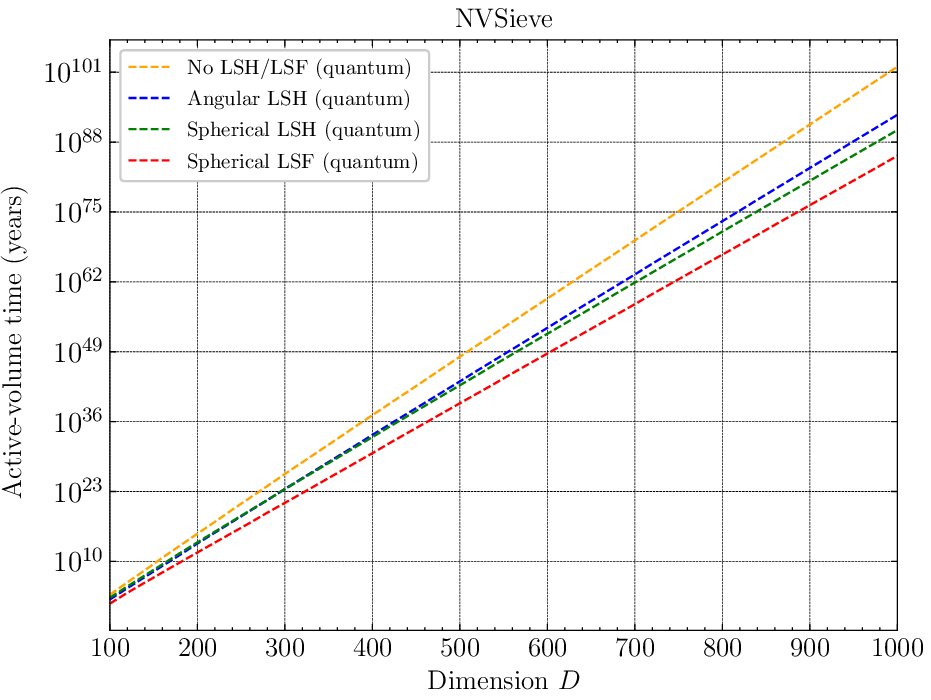}
        \caption{\centering Active-volume circuit time of $\mathtt{NVSieve}$ without $\mathsf{QRAM}$}
    \end{subfigure}
    \begin{subfigure}[b]{0.49\textwidth}
        \includegraphics[width=\textwidth]{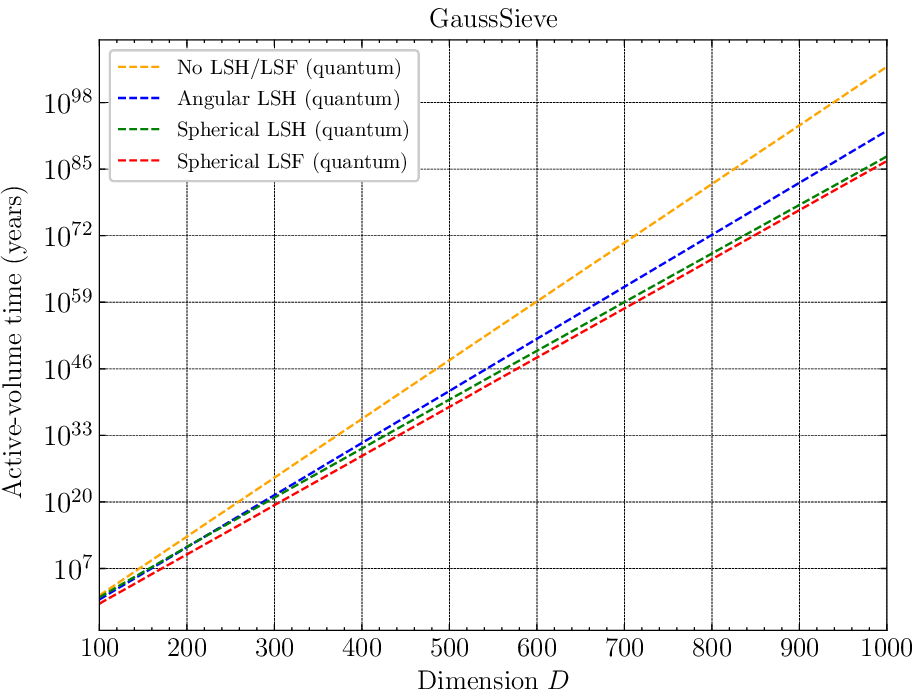}
        \caption{\centering Active-volume circuit time of $\mathtt{GaussSieve}$ without $\mathsf{QRAM}$}
    \end{subfigure}
    \begin{subfigure}[b]{0.49\textwidth}
        \includegraphics[width=\textwidth]{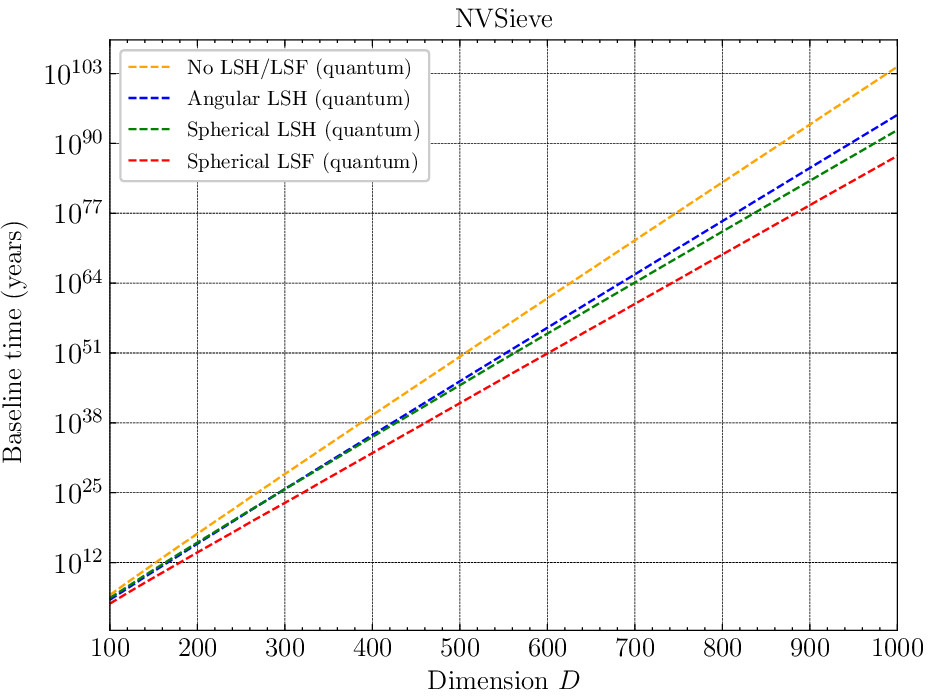}
        \caption{Baseline circuit time of $\mathtt{NVSieve}$ without $\mathsf{QRAM}$}
    \end{subfigure}
    \begin{subfigure}[b]{0.49\textwidth}
        \includegraphics[width=\textwidth]{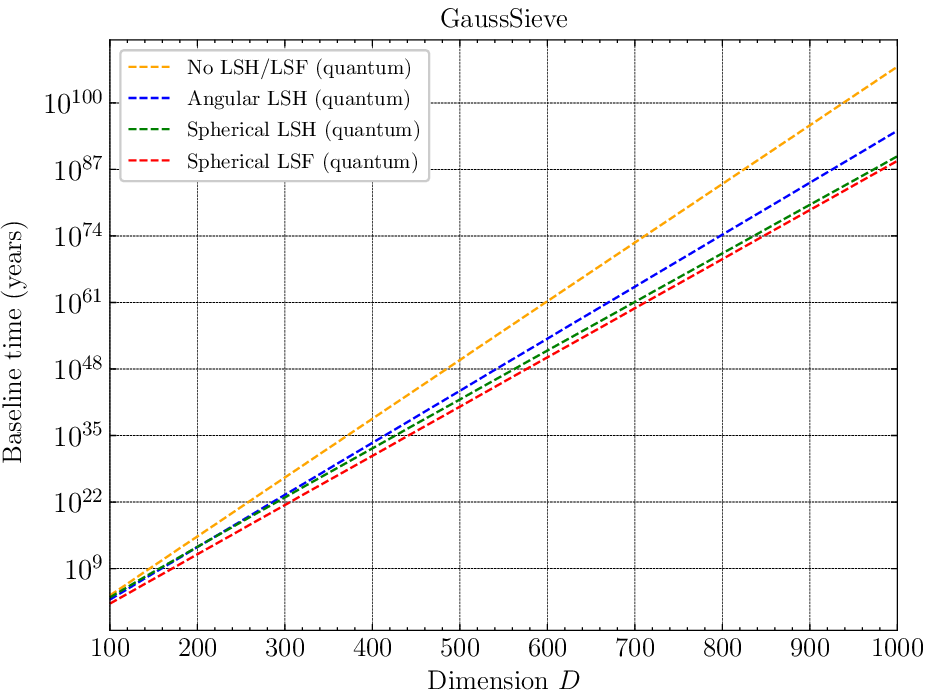}
        \caption{Baseline circuit time of $\mathtt{GaussSieve}$ without $\mathsf{QRAM}$}
    \end{subfigure}
    \caption{Number of physical qubits and circuit times in $\mathtt{NVSieve}$ and $\mathtt{GaussSieve}$ with and without LSH/LSF under baseline and active-volume physical architectures as a function of lattice dimension~$D$ in the scenario where $\mathsf{QRAM}$ has negligible cost. Reaction limits and circuit time also include classical hashing time.}
    \label{fig:extra_result2}
\end{figure}

\begin{figure}[t]
    \centering
    \begin{subfigure}[b]{0.49\textwidth}
        \includegraphics[width=\textwidth]{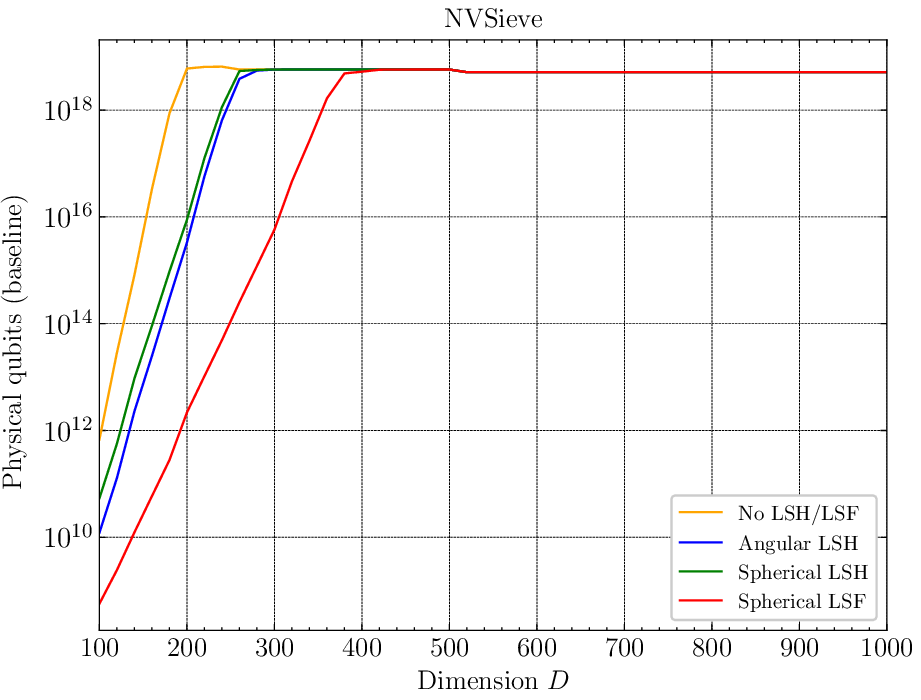}
        \caption{\centering Baseline physical qubits of $\mathtt{NVSieve}$ with limited depth}
    \end{subfigure}
    \begin{subfigure}[b]{0.49\textwidth}
        \includegraphics[width=\textwidth]{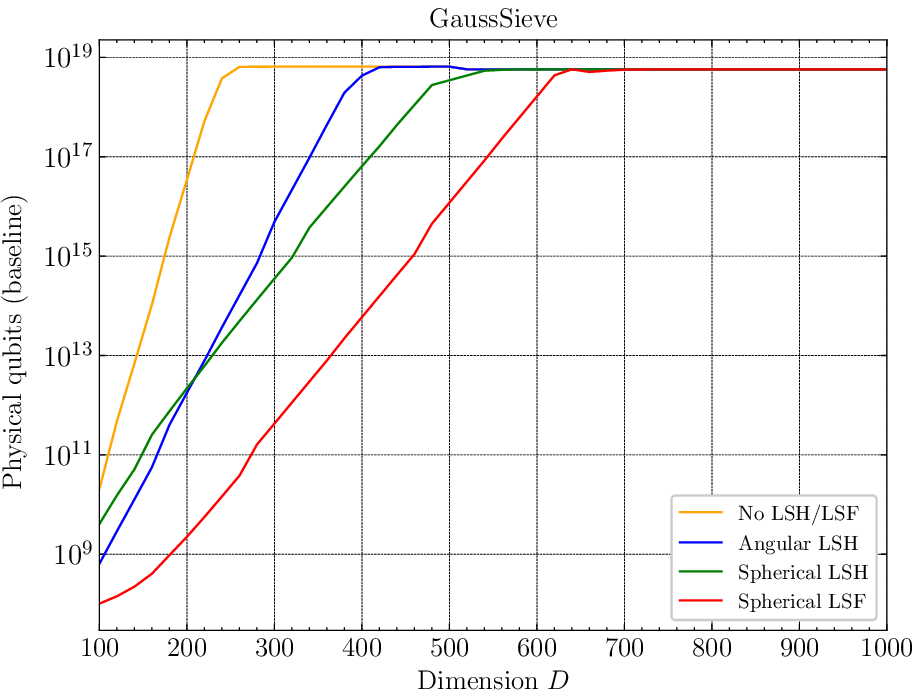}
        \caption{\centering Baseline physical qubits of $\mathtt{GaussSieve}$ with limited depth}
    \end{subfigure}
    \begin{subfigure}[b]{0.49\textwidth}
        \includegraphics[width=\textwidth]{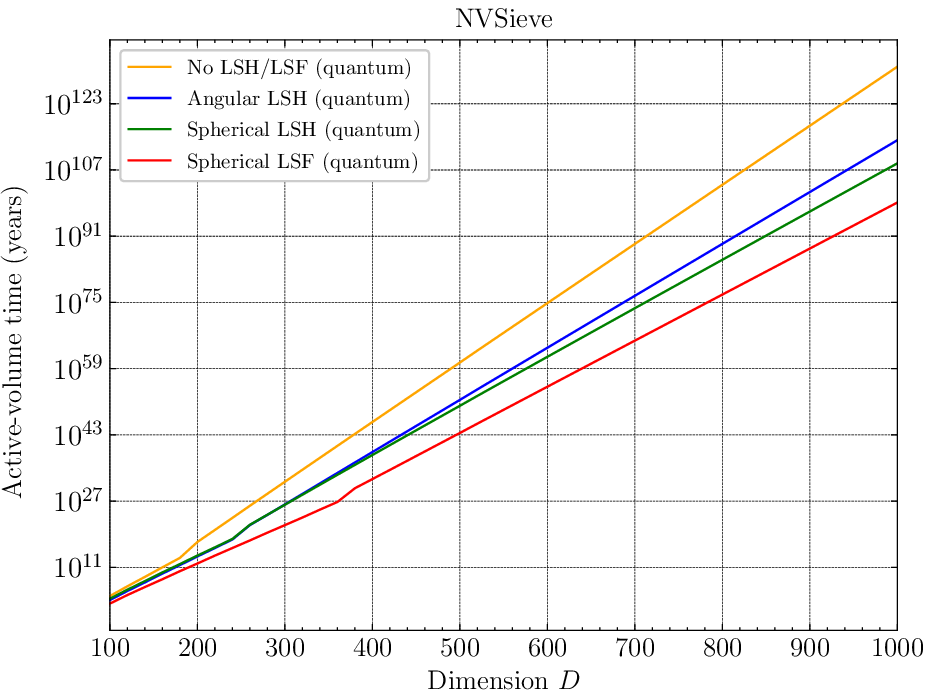}
        \caption{\centering Active-volume circuit time of $\mathtt{NVSieve}$ with limited depth}
    \end{subfigure}
    \begin{subfigure}[b]{0.49\textwidth}
        \includegraphics[width=\textwidth]{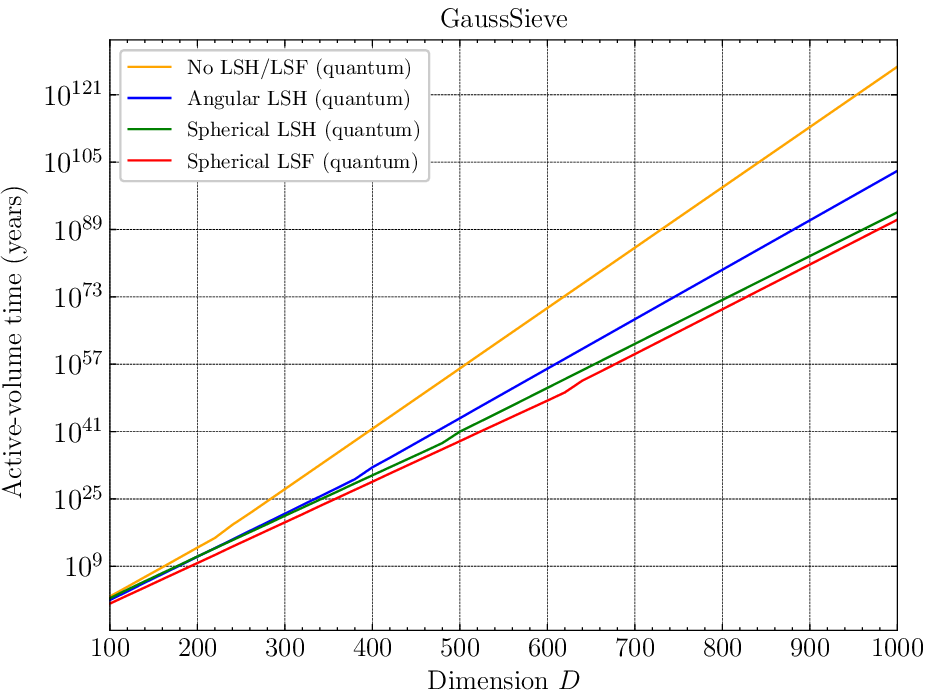}
        \caption{\centering Active-volume circuit time of $\mathtt{GaussSieve}$ with limited depth}
    \end{subfigure}
    \begin{subfigure}[b]{0.49\textwidth}
        \includegraphics[width=\textwidth]{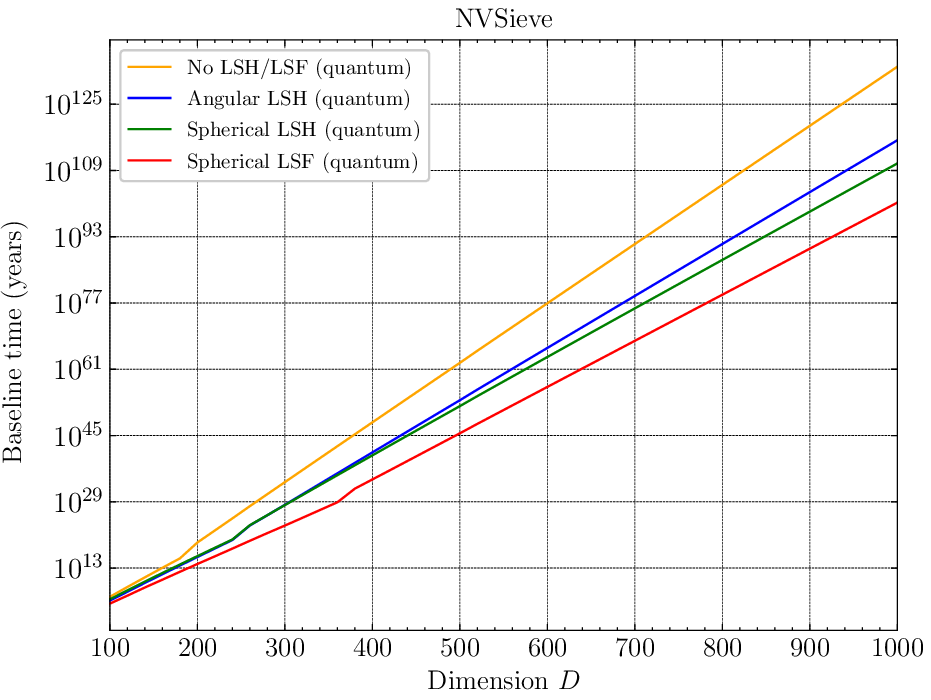}
        \caption{\centering Baseline circuit time of $\mathtt{NVSieve}$ with limited depth}
    \end{subfigure}
    \begin{subfigure}[b]{0.49\textwidth}
        \includegraphics[width=\textwidth]{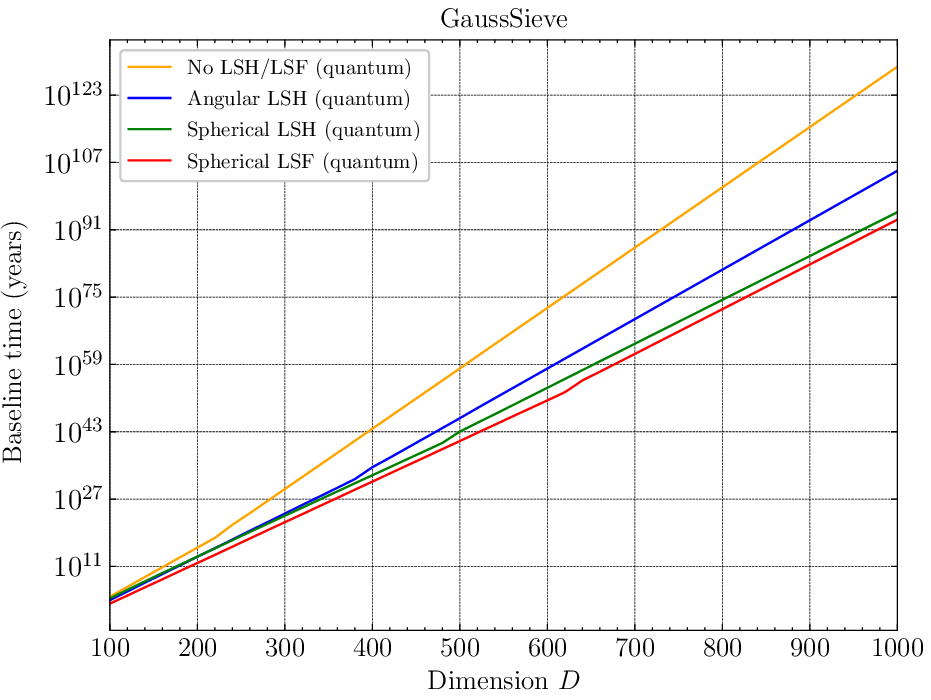}
        \caption{\centering Baseline circuit time of $\mathtt{GaussSieve}$ with limited depth}
    \end{subfigure}
    \caption{Number of physical qubits and circuit times in $\mathtt{NVSieve}$ and $\mathtt{GaussSieve}$ under baseline and active-volume physical architectures as a function of lattice dimension $D$ in the scenario where the reaction depth of each Grover's search is at most $2^{40}$. Circuit time includes classical hashing time.}
    \label{fig:extra_result3}
\end{figure}

\end{document}